\documentclass[11pt]{article}

\usepackage{amsmath}

 \usepackage{natbib}
 
\usepackage{newtxtext}
\usepackage[subscriptcorrection]{newtxmath}

\usepackage[T1]{fontenc}
\usepackage[english]{babel}
\usepackage{graphicx}
\usepackage{dsfont}
\usepackage{amsmath}
\usepackage{xcolor}
\usepackage{colortbl}
\usepackage{amsfonts}
\usepackage{amsmath}
\usepackage{endnotes}						
\usepackage[utf8]{inputenc}
\inputencoding{latin1}
\inputencoding{utf8}
\usepackage{tikz}
\usepackage{tikzscale}
\usetikzlibrary{calc}
\usetikzlibrary{bayesnet}
\usepackage{subcaption}
\usepackage{mathrsfs}
\usepackage{bbm}
\usepackage{setspace} 
\usepackage{textcomp}
\usepackage{pdflscape}
\usepackage{multirow}
\usepackage{multicol}

\usepackage{epsfig}
\usepackage[normalem]{ulem}
\usepackage{wasysym}

\usepackage{amsmath}

\usetikzlibrary{bayesnet}
\usetikzlibrary{shapes,fit,arrows,positioning,quotes}

\tikzstyle{vc} = [circle, draw, text centered,minimum size=14mm]
\tikzstyle{vs} = [draw, text centered,minimum size=12mm]

\tikzstyle{ef} = [draw, -]
\tikzstyle{ed} = [draw,dotted, -]

\tikzstyle{ebf} = [draw, latex-latex]
\tikzstyle{ebd} = [draw,dotted, latex-latex]

\tikzstyle{edf} = [draw, -latex]
\usepackage{bbm}

\newcommand{\R}{\ensuremath{\mathbb{R}}}
\newcommand{\X}{\ensuremath{\mathcal{X}}}

\newcommand{\diag}{\rm{diag}}

\newcommand{\Exp}{\mathds{E}}
\newcommand{\Prob}{\mathds{P}}

\newcommand{\1}{\ensuremath{\mathds{1}_{n_x}}}
\newcommand{\Z}{\mathcal{Z}}

\newcommand{\C}{\mathscr{C}}
\newcommand{\U}{\mathcal{U}}

\newcommand{\E}{\mathcal{E}}
\newcommand{\G}{\mathcal{G}}
\newcommand{\EU}{\mathcal{E}_{\mathcal{U}}}

\newcommand{\GU}{\mathcal{G}_{\mathcal{U}}}
\newcommand{\GB}{\mathcal{G}_{\mathcal{B}}}
\newcommand{\bigCI}{\mathrel{\text{\scalebox{1.07}{$\perp\mkern-10mu\perp$}}}}
\newcommand*\circlee[1]{\tikz[baseline=(char.base)]{
		\node[shape=circle,draw,inner sep=2pt] (char) {#1};}}
\newcommand{\circleee}{\hspace*{-1mm} \circlee{}}

\newcommand{\ci}{\!\perp \! \! \! \perp\!}
\newcommand{\nci}{\mbox{\protect $\: \perp \hspace{-1.75ex}\perp $} \hspace{-1.3mm}/ \hspace{1.5mm}}

\graphicspath{{./Figures/}}

\usepackage[plain,noend]{algorithm2e}

\def\*#1{\bm{#1}} 

\def\@fnsymbol#1{\ensuremath{\ifcase#1\or *\or \dagger\or \ddagger\or
   \mathsection\or \mathparagraph\or \|\or **\or \dagger\dagger
   \or \ddagger\ddagger \else\@ctrerr\fi}}
\newcommand{\ssymbol}[1]{^{\@fnsymbol{#1}}}

\let\OLDthebibliography\thebibliography
\renewcommand\thebibliography[1]{
  \OLDthebibliography{#1}
  \setlength{\parskip}{0pt}
  \setlength{\itemsep}{3pt plus 0.3ex}
}

\def\*#1{\bm{#1}} 

\usepackage[left=2.3cm, right=2.3cm, top=3cm]{geometry}

\newtheorem{definition}{Definition}
\newtheorem{theorem}{Theorem}
\newtheorem{lemma}{Lemma}
\newtheorem{corollary}{Corollary}
\newtheorem{proposition}{Proposition}
\newtheorem{proof}{Proof}
\newtheorem{example}{Example}

\newcommand{\ml}{\textcolor{black}}
\newcommand{\aap}{\textcolor{black}}

\usepackage{authblk}
\title{\bf  {Profile Graphical Models}}
\author[1,2,*]{Alejandra Avalos-Pacheco}
\author[3]{Monia Lupparelli}
\author[3]{Francesco C. Stingo}
\affil[1]{Institute of Applied Statistics, Johannes Kepler University Linz}
\affil[2]{Harvard-MIT Center for Regulatory Science, Harvard Medical School}
\affil[3]{Depart.\ of Statistics, Computer Science, Applications ``G. Parenti'', University of Florence}
\affil[*]{\textit {\textcolor{blue}{alejandra.avalos\textunderscore pacheco@jku.at}}}
\date{October 2024}

\begin{document}




\def\spacingset#1{\renewcommand{\baselinestretch}%
{#1}\small\normalsize} \spacingset{1}

\setcounter{Maxaffil}{0}
\renewcommand\Affilfont{\itshape\small}

\spacingset{1.42} 

\maketitle
\begin{abstract}
We introduce a novel class of graphical models, termed profile graphical models, that represent, within a single graph, how an external factor influences the dependence structure of a multivariate set of variables. 
This class is quite general and includes multiple graphs and chain graphs as special cases. 
Profile graphical models capture the conditional distributions of a multivariate random vector given different levels of a risk factor, 
and learn how the conditional independence structure among variables may vary across  these risk profiles; we formally define this family of models and establish their corresponding Markov properties. 
We derive key structural and probabilistic properties that underpin a more powerful inferential framework than existing approaches, underscoring that our contribution extends beyond a novel graphical representation.
Furthermore, we show that the resulting profile undirected graphical models are independence-compatible with two-block LWF chain graph models.
We then develop a Bayesian approach for Gaussian undirected profile graphical models based on continuous spike-and-slab priors to learn shared sparsity structures across different levels of the risk factor. 
We also design a fast EM algorithm for efficient inference. 
Inferential properties are explored through simulation studies, including the comparison with competing methods. 
The practical utility of this class of models is demonstrated through the analysis of protein network data from various subtypes of acute myeloid leukemia.
Our results show a more parsimonious network and greater patient heterogeneity than its competitors, highlighting its enhanced ability to capture subject-specific differences.
\end{abstract}

\noindent %

\spacingset{1.45}

\section{Introduction}

Multivariate regression models can be represented by chain graphs \citep{lauritzen1989graphical,frydenberg1990,andersson2001alternative, wersad2012}, which capture the conditional independence structure between multiple response and explanatory variables. In their simplest form, responses and predictors form two distinct chain components, with missing edges corresponding to conditional independencies under suitable Markov properties. Among the various types of chain graphs \citep{drton2009discrete}, we focus on the LWF chain graph \citep{frydenberg1990}, which defines a smooth statistical model and provides a flexible framework for modeling dependencies among outcomes given covariates.

Model selection for chain graphs has attracted significant attention, with recent developments in penalized likelihood \citep{rothman2010sparse,yin2011sparse,lee2012simultaneous}, two-step procedures \citep{cai2012covariate,chen2016asymptotically}, and Bayesian approaches \citep{bhadra2013joint,consonni2017objective}. However, these models remain limited when the goal is to characterize how an explanatory variable affects the joint dependence structure among outcomes, rather than each outcome individually. The conditional independence models encoded by chain graphs provide information only through missing edges, leaving unexplored how existing dependencies change with external factors. This issue is particularly relevant in situations where associations between variables may reverse or shift across different conditions, as in the well-known effect reversal \citep{coxwer2003} and Simpson’s paradox \citep{simpson1951}.

An alternative line of work addresses this issue indirectly by modeling subgroups or subpopulations through multiple graphical models \citep{guo2011joint,danaher2014joint,peterson2015bayesian} or context-specific independencies \citep{hojsgaard2003split,corander2003labelled,nyman2014stratified,nyman2016context}. Yet, these approaches typically do not incorporate external factors directly into the model and often restrict context-specific variations to adjacent vertices.

Building on these ideas, we propose a novel class of graphical models—profile undirected graphs—which preserve the interpretability of chain graphs while extending them to model how dependence structures among responses vary with an external factor. Our contributions are fourfold: (i) We introduce profile undirected graphs as a general framework for modeling all profile outcome distributions, i.e., conditional distributions of responses given any level of a risk factor. (ii) We derive the corresponding Markov properties based on a unified connected set rule. (iii) We establish formal compatibility, in terms of independence models, between the proposed profile graphs and specific chain graph structures. (iv) We develop parameterizations for Gaussian models and incorporate continuous spike-and-slab priors to learn shared sparsity patterns across levels of the external factor.
For efficient inference, we implement a fast EM algorithm and demonstrate the usefulness of our proposed methodology through extensive simulations including comparison with competing methods and a cancer genomics application.
We provide a profile graphical model of protein networks that evolve across disease subtypes, thereby uncovering subtype-specific dependency patterns invisible to standard chain graph or multiple graph analysis.


\section{Theoretical framework basic setup}\label{sec.setup}

Let $G=(V,E)$ be a graph defined by a set of vertices $a \in V$ and a set of edges $(a,b) \in E$ joining pairs of vertices $a,b \in V$, and let $Y_{V}=(Y_{a})_{a \in V}$ be a random vector of variables indexed by the finite set $V$ with $|p=V|$. A graph, associated to a random vector $Y_V$, is generally used to represent conditional independence structures under suitable Markov properties. Typically, missing edges in $G$ correspond to conditional independencies for the joint distribution of $Y_V$. 
Also, let us consider the random categorical variable $X$ with strictly positive probability distribution representing an external factor with respect to (in the sequel, wrt) the random vector $Y_V$ of outcome/response variables. The variable $X$ takes level $x \in \X$, with $q=|\X|$. Our interest lies in the effect of $X$ on the joint independence structure of $Y_V$ and, in particular, in exploring via a graphical modelling approach how this structure may change under different levels $x \in \X$, which we call \textit{profiles}.
Chain graphs are generally used to model the effects of background variables on joint response variables. 
In the simplest form, a two-block chain graph $C=[\{C_1,C_2\},E]$ is defined by a set of vertices partitioned in chain components $C_1$ and $C_2$, and a set of edges $E$. 
Depending on the set of Markov properties specified for the chain graph we may have different independence models for the joint distribution of random vectors $(Y_{C_{t}})_{t \in \{1,2\}}$, associated to the chain components $\{C_t\}_{t \in \{1,2\}}$. 
In particular we focus on the class of LWF chain graph models \citep{frydenberg1990}; these models correspond to multivariate regression models with suitable independence constraints corresponding by missing edges, both within and between chain components.
Any pair of vertices $a,b \in C_t$ within the same chain component with $t=1,2$ and $a \neq b$, can be joined by undirected edges; vertices between chain components, $a \in C_1$ and $b \in C_2$, are joined by directed edges preserving the same direction such that cycles are not allowed. For our purpose, the set of vertices $C_1$ and $C_2$ are associated, respectively, to the random vector $Y_V$ of response variables and to the background variable $X$, so that  $C_1=V$ and $|C_2|=1$. In principle, the chain component $C_2$ may include a multiple categorical random vector; in this case  $X$   represents a random variable with state space given by the combination of a multiple factor levels.
For any  $x \in \X$, let $Y_V(x)$ be a \textit{x-profile outcome vector}, that is  the  random vector $Y_V|\{X=x\}$ conditioned on a specific profile $x$ of the factor $X$, and let $P(Y_V(x))$ 
be the corresponding \textit{x-profile  probability distribution} of $Y_V(x)$, that is the  conditional probability distribution $P(Y_V|\{X=x\})$. Note that $P(\cdot)$ can be a probability density function or a probability mass function, depending on the continuous or discrete nature of the multivariate random variable $Y_V(x)$, with $x \in \X$.  For sake of simplicity, in the sequel we omit the prefix $x$ to denote both the profile outcome vector and the profile outcome distribution.
Then, for a given  multivariate random vector $Y_V$ and an external factor $X$, let $Y_{V|\X}=[Y_V(x)]_{x \in \X}$ be the finite set of all profile outcome vectors and let $P(Y_{V|\X})=[P(Y_V(x))]_{x \in \X}$ be the corresponding set of all profile outcome distributions. For any $A \subseteq V$, $Y_{A|\X}=[Y_A(x)]_{x \in \X}$ is set of marginal profile outcome vectors with corresponding profile probability distributions  $P(Y_{A|\X})=[P(Y_A(x))]_{x \in \X}$. 
A definition of profile-independence follows.
\begin{definition}
Given a graph $G=(V,E)$ and a partition $A,B,C \subseteq V$, the \textit{profile conditional independence} $Y_{A}(x) \ci  Y_{B}(x)|Y_C(x)$ and the \textit{profile marginal independence} $Y_{A}(x) \ci  Y_{B}(x)$ correspond, respectively, to the factorizations
\begin{eqnarray}\label{eq:x-profile.cond.ind}
P[Y_{A}(x),Y_{B}(x)|Y_{C}(x)]&=&P[Y_{A}(x)|Y_{C}(x)]\times P[Y_{B}(x)|Y_{C}(x)],\\
\label{eq:x-profile.marg.ind}
P[Y_{A}(x),Y_{B}(x)]&=&P[Y_{A}(x)]\times P[Y_{B}(x)],
\end{eqnarray}
of the joint profile distribution $Y_V(x)$, for any $x \in \X$. 
\end{definition}
The following lemma holds by definition of conditional independence.
\begin{lemma}
    If the profile independence statements in Equations \eqref{eq:x-profile.cond.ind} and \eqref{eq:x-profile.marg.ind} hold for any level $x \in \X$, then these equations imply that $Y_A \ci Y_B|\{Y_C,X\}$ and $Y_A \ci Y_B|X$, respectively.
\end{lemma}

Finally, let us consider a collection of  \textit{multiple graphs} $G_{V|\X}=\{G(x)=(V,E(x))\}_{x \in \X}$ associated to the profile outcome distributions $P(Y_{V|\X})$.
Under suitable Markov properties, any graph $G(x)$  represents an independence model for the profile outcome vector $Y_V(x)$, for any $x \in \X$. In particular, missing edges wrt $G(x)$ correspond to profile conditional independencies for the joint distribution of $Y_V(x)$, with $x \in \X$. Graphs $G(x) \in G_{V|\X}$ may have different skeletons.

We remark that chain graph models do not allow to explore how the independence structure of $Y_V$  may considerably vary for any profile $x \in \X$. Multiple graphs do not allow to model the effect of $X$ on each outcome $Y_a \in Y_V$.
In essence, the idea is to provide a single graph able to embed, at the same time, information about the profile independence structure for any $Y_V(x) \in Y_{V|\X}$ and about the conditional independence between $X$ and any outcome $Y_a \in Y_V$. In the following sections, we derive the properties of this new class of graphical models that fill the gap between the class of chain graphs and the one of multiple graphs and highlight the connection between these two classes. 
We also exploit our results to define a more powerful, with respect to state-of-the art approaches, inferential procedure. 

\section{Profile undirected graphical models}\label{sec.profile}
\subsection{Profile undirected graphs}

We introduce the class of undirected profile graphs. A profile undirected graph  $\GU=(V,\E)$ is defined by the set $V$ of vertices and a set of $\Z$-labelled edges $\E$ which are labelled according to a subset $\Z \subseteq \X$. Let $(a,b)^{\Z}$ be the generic element of $\E$ associated to any pair $a,b \in V$, where the presence or absence of the edge between $a$ and $b$ is determined by the subset $\Z$ of the state space $\X$. For each pair $a,b \in V$, the corresponding edge $(a,b)^{\Z} \in \E$ will belong to one of the following three categories: (i) if $\Z = \X$, vertices $a$ and $b$ are not joined by any edge, (ii) if $\Z$ is a nonempty proper subset of $\X$, $\Z \subset \X$ and $\Z \neq \emptyset$, vertices $a$ and $b$ are joined by a dotted $\Z$-labelled edge; (iii) if $\Z = \emptyset$, vertices $a$ and $b$ are joined by a full edge and, for sake of simplicity, the $\emptyset$-label is not displayed in the graph.
Under suitable Markov properties, the profile graph $\GU$ provides an independence model for the joint distributions of a random vector $Y_{V|\X}$ of profile outcomes. In particular, a missing edge in $\GU$ corresponds to a profile conditional independence for each profile $x \in \X$. 
A  $\Z$-labelled dotted  edge in $\GU$ corresponds to profile conditional independencies holding only for the profiles $x \in  \Z$, with $\Z \subset \X$ and $\Z \neq \emptyset$.

Further technical definitions are given. For any couple of vertices $a,b \in V$, we say that $b$ is an $x$-\textit{neighbour} of $a$ and \textit{vice versa}, if they are joined by a  $\Z$-labelled  edge  such that  $x \notin \Z$, with $\Z \subset \X$.   Let $nb_x(a)$ be the set of all $x$-neighbours of $a$, with $a \in V$ and  $x \in \X$.  For any pair $a,b \in V$ and $x \in \X$, an $x$-$path$ between $a$ and $b$ is given by a sequence of $(a,b)^{\Z}$ edges, for any $\Z \subset \X$, such that $x \notin \Z$ for all edges in the sequence.  Given any nonempty subset $C$ of $V$, $C$  is said to be $x$-\textit{connected} if any pair $a,b \in C$ is joined by a $x$-path, with $x \in \X$.  
Any nonempty subset $D$ of $V$ is said to be $x$-\textit{disconnected} if it is not $x$-connected, with $x \in \X$ and let $K_1,\dots,K_r$ be the $x$-connected components of $D$.
For any triple $A,B,C$ of disjoint subsets of $V$ and $x \in \X$, we say that $C$ $x$-\textit{separates} $A$ from $B$ if every $x$-path from any vertex $a \in A$ to any vertex $b \in B$ intersects $C$.  Technical $x$-definitions above can be simply extended to $\Z$-definitions for any subset $\Z$ of $\X$ if they hold for all $x \in \Z$.

\begin{figure}
	\begin{center}
		\begin{tikzpicture}[node distance = {1.2cm and 1.2cm}, v/.style = {draw, circle,minimum size=0.7cm},,mylabel/.style={thin,  align=center,fill=white,bend right}]
		
		\node at (1,-3.2) {$\G_U$};
		
		\node (a) [v] {a};
		\node (b) [v, below = of a] {b};
		\node (c) [v, right = of a] {c};
		\node (d) [v, below = of c] {d};
		
		\draw (b) to node {} (d);
		
		\draw[ed] (a) to node[mylabel]{\tiny 2}  (b);
		\draw[ed] (b) to node [sloped] [mylabel]{\tiny 1,2} (c);
		\draw[ed] (a) to node [sloped] [mylabel]{\tiny 0} (c);

		\end{tikzpicture}
		\hspace*{1.3cm} 
		\begin{tikzpicture} [node distance = {1.2cm and 1.2cm}, v/.style = {draw, circle,minimum size=0.7cm},mylabel/.style={thin,  align=center,fill=white,bend right}]
		
		\node at (1,-3.2) {$U(0)$};

		\node (a) [v] {a};
		\node (b) [v, below = of a] {b};
		\node (c) [v, right = of a] {c};
		\node (d) [v, below = of c] {d};
		
		\draw (a) to node {} (b);
		\draw (b) to node {} (c);
		\draw (b) to node {} (d);

		\end{tikzpicture} \hspace*{0.4cm}
		\begin{tikzpicture} [node distance = {1.2cm and 1.2cm}, v/.style = {draw, circle,minimum size=0.7cm}]
		
		\node at (1,-3.2) {$U(1)$};

		\node (a) [v] {a};
		\node (b) [v, below = of a] {b};
		\node (c) [v, right = of a] {c};
		\node (d) [v, below = of c] {d};
		
		\draw (a) to node {} (b);
		\draw (a) to node {} (c);
		\draw (b) to node {} (d);
		
		\end{tikzpicture} \hspace*{0.4cm}
		\begin{tikzpicture} [node distance = {1.2cm and 1.2cm}, v/.style = {draw, circle,minimum size=0.7cm}]
		
		\node at (1,-3.2) {$U(2)$};

		\node (a) [v] {a};
		\node (b) [v, below = of a] {b};
		\node (c) [v, right = of a] {c};
		\node (d) [v, below = of c] {d};
		
		\draw (a) to node {} (c);
		\draw (b) to node {} (d);
		
		\end{tikzpicture}
	\end{center}
	
	\caption{Given $V=\{a,b,c,d\}$, $\GU$ is a  profile undirected graph for the profile outcome vectors $Y_{V|\X}=[(Y_V(x))_{x \in \X}]$ with $\X=\{0,1,2\}$. Any $U(x)$ is the induced undirected graph for the profile outcome vector $Y_V(x)$, with $x \in \X$.}
	
	\label{fig:fig1} 
	
\end{figure}
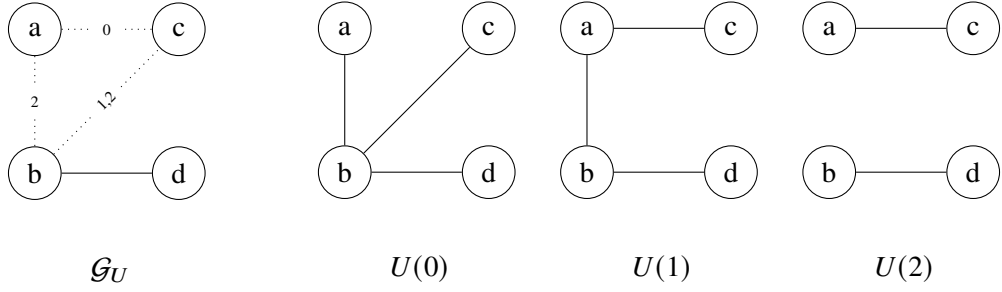


\begin{example} 
	Consider the profile undirected graph $\mathcal{G}_U$ in the left panel of Figure \ref{fig:fig1}.  Vertices $a$ and $c$ are both $\{1,2\}$-neighbours, because they are joined by a dotted edge with label $\Z=\{0\}$ that does not contain neither  $1$ or $2$. Vertices $b$ and $d$ are $\X$-neighbours because they are joined by a full edge.  The sequence of edges $\{(a,c)^{\{0\}},(a,b)^{\{2\}},(b,d)^{\emptyset}\}$ is a $\{1\}$-path, since $1$ is not included in any label of the edges in the sequence. The same sequence is not a $\{2\}$-path since the label of the couple $(a,b)$  contains $2$. The set $V$ is $\{1\}$-connected, because every pair of vertices in $V$ are joined by  a $\{1\}$-path. The same set is $\{2\}$-disconnected with $\{2\}$-connected components $\{a,c\}$ and $\{b,d\}$, because does not exist a $\{2\}$-path between $a$ and $b$. 
	Vertices $c$ and $d$ are $\{1\}$-separated by $a$ because the only $\{1\}$-path $\{(a,c)^{\{0\}},(a,b)^{\{2\}},(b,d)\}$ between $c$ and $d$ intersects $a$; vertex $a$ does not $\{0\}$-separates $c$ and $d$ because there exists the $\{0\}$-path $\{(b,c)^{\{1,2\}},(b,d)^{\emptyset}\}$ between them that does not intersects $a$.
\end{example}

\subsection{Profile undirected Markov properties}
\label{sec.profile-und}

In this section, we first define Markov properties for profile undirected graphs and then derive a few results that connect them.

\begin{definition}
The probability distributions $P[Y_{V|\X}]$ of the profile outcome vectors $Y_{V|\X}$ satisfy the \textit{profile undirected Pairwise Markov Property} ($\U$-PMP)  wrt the graph $\GU=(V,\E_{\U})$ if, for any  $(a,b)^{\Z} \in \EU$ with $\Z \subseteq \X$, 
\begin{equation}
Y_a(x) \bigCI Y_b(x) | Y_{V\setminus \{a,b\}}(x), \qquad x \in \Z.
\end{equation}
\end{definition}
\begin{definition}
The probability distributions $P[Y_{V|\X}]$ of the profile outcome vectors $Y_{V|\X}$ satisfy the \textit{profile undirected Global Markov Property} ($\U$-GMP)  wrt the graph $\GU=(V,\E_{\U})$ if, 
for any triple $A,B,C$ of disjoint subsets of $V$ such that $C$ $x$-separates $A$ from $B$ in $\GU$, 
\begin{equation}
Y_{A}(x) \bigCI  Y_{B}(x)|Y_{C}(x), \qquad x \in \X. 
\end{equation} 
\end{definition}
\begin{definition}
The probability distributions $P[Y_{V|\X}]$ of the profile outcome vectors $Y_{V|\X}$ satisfy the \textit{profile undirected Connected Set Markov Property} ($\U$-CSMP)  wrt the graph $\GU=(V,\E_{\U})$ if, for any $x$-disconnected set $D$ of $V$, with $K_1,\dots,K_r$ $x$-connected components of $D$,
\begin{equation}
Y_{K_1}(x) \bigCI \dotso \bigCI Y_{K_r}(x) | Y_{V\setminus D}(x), \qquad x \in \X.
\end{equation}
\end{definition}
\begin{example}
	
	Consider the left panel including the graph $\GU$ in Figure \ref{fig:fig1}. $P[Y_{V|\X}]$ satisfy the $\U$-PMP wrt $\GU$ if $Y_{b}(x) \bigCI  Y_{c}(x)|\{Y_{a}(x),Y_{d}(x)\}$ for  $x \in \{1,2\}$, since  $(b,c)^{\{1,2\}} \in \EU$. $P[Y_{V|\X}]$ satisfy the $\U$-GMP wrt $\GU$ if $Y_{c}(1) \bigCI  \{Y_{b}(1),Y_{d}(1)\}|Y_{a}(1)$  because $a$ $\{1\}$-separates $c$ from $\{b,d\}$. Consider the subset $D=\{a,b,c\}$ of $V$; $P[Y_{V|\X}]$ satisfy the $\U$-CSMP wrt $\GU$ if $\{Y_{a}(2),Y_{c}(2)\} \bigCI  Y_{b}(2)|Y_{d}(2)$ because $D$ is $\{2\}$-disconnected set with two $\{2\}$-connected components $\{a,c\}$ and $b$.
\end{example}

We prove that all the independence statements encoded in a profile undirected graph under the global Markov property can be derived by applying the connected set rule {that gives insight on the connectivity and on the path setting of the graph across profiles}.

\begin{theorem}\label{t1}
	Let $\GU=(V,\EU)$ be a profile undirected graph model associated to the profile outcome vectors $Y_{V|\X}$ with probability distributions $P[Y_{V|\X}]$. The $\U$-GMP is satisfied if and only if the $\U$-CSMP is satisfied wrt $\GU$. 
\end{theorem} 
The proof of Theorem 1 is given in the Supplementary Material, along with all other proofs. The local Markov property for profile undirected graph is also included in the Supplementary Material. Given a profile undirected graph $\GU=(V,\EU)$ for the profile outcome vectors $Y_{V|\X}$,  the corresponding class  of multiple undirected graphs associated to each random vector $Y_V(x) \in Y_{V|\X}$ can be defined.

\begin{definition}\label{def.multUnd}
	Given a profile undirected graph $\GU=(V,\EU)$ for the profile outcome vectors $Y_{V|\X}$, let $ U_{V|\X}=\{U(x)=(V, E_U(x))\}_{x \in \X}$ be the induced class of multiple undirected graphs, where, for any $U(x) \in U_{V|\X}$, the couple $a,b \in V$ is joined by an undirected edge if $x \notin \Z$ in the corresponding edge $(a,b)^{\Z} \in \EU$, with $\Z \subseteq \X$.
\end{definition}

Then, a missing edge in $\GU$ corresponds to a missing edge in $U(x)$, for any $x \in \X$; a $\Z$-labelled dotted edge in $\GU$ corresponds to a missing edge in $U(x)$ if $x \in \Z$, and to a full edge in $U(x)$ if $x \notin \Z$; a full edge in $\GU$ corresponds to a full edge in $U(x)$, for any $x \in \X$. 

\begin{example}
	Consider Figure \ref{fig:fig1}. Given the profile undirected graph $\GU$, let $U_{V|\X}=\{U(0),U(1),U(2)\}$ be the induced class of multiple undirected graphs. The couple $a,d$ is disjoined in $\GU$ and in any $U(x) \in U_{V|\X}$. The couple $b,c$ is joined by a $\{1,2\}$-labelled dotted edge in $\GU$ then is joined by a full edge in $U(0)$ and is disjoined in $U(1),U(2)$. The couple $b,d$ is joined by a full edge in $\GU$ and in any $U(x) \in U_{V|\X}$. 
\end{example}


\noindent Pairwise, local, and global Markov property of probability distributions associated to undirected graphs are well known \citep{lauritzen1996graphical}. The following corollary, derived directly from Theorem \ref{t1}, shows that the full set of conditional independencies implied by the global Markov property for any undirected graph can be also derived by applying the connected set rule.
\begin{corollary}\label{c1} Given an undirected graph model $U(x)=(V,E(x))$ associated to the profile outcome vectors $Y_{V|\X}$, the probability distributions $P[Y_{V}(x)]$ satisfies  the global Markov property wrt $U(x)$ if and only if the connected set Markov property is satisfied for every $x$-disconnected set $D \subseteq V$, with $x \in \X$. 
\end{corollary}
The following proposition shows that the full set of independencies encoded in the induced undirected graph model for any  $Y_V(x) \in Y_{V|\X}$ can be derived from the profile undirected graph model for the joint distributions of $Y_{V|\X}$.


\begin{proposition}\label{p2} Consider a profile undirected graph  $\GU=(V,\EU)$ associated to the profile outcome vectors $Y_{V|\X}$ and the induced class of multiple undirected graphs $U_{V|\X}$.  If the probability distributions $P[Y_{V|\X}]$ satisfy the $\U$-CSMP wrt $\GU$, the probability distribution $P[Y_V(x)]$ of each profile vector $Y_V(x) \in Y_{V|\X}$ satisfies the global Markov property wrt the induced undirected graph $U(x) \in U_{V|\X}$. 
\end{proposition}

\noindent In the following proposition we show that $\U$-GMP, $\U$-CSMP and $\U$-PMP are equivalent for the class of profile undirected graph models in case of strictly positive probability distributions. This result directly derives from Proposition \ref{p2}.

\begin{proposition}\label{p1} Let $\GU=(V,\EU)$ be a profile undirected graph associated to the profile outcome vectors $Y_{V|\X}$ with strictly positive probability distributions $P[Y_{V|\X}]$. The $\U$-GMP is satisfied if and only if the $\U$-PMP is satisfied wrt $\GU$.
\end{proposition} 

\section{Profile undirected graphs and LWF chain graphs} \label{subsec.chainvsund}
For any profile undirected graph $\GU$, we derive an induced class of two-block LWF chain graphs $\C_U=\{C_U\}$, with generic element $C_U=[\{V,X\}, E_{C_U}]$, such that the joint distribution $P(Y_V,X)$ under $C_U$ is compatible, in terms of independence models, with the set of profile distributions $P[Y_{V|\X}]$ under $\GU$. Preliminary definitions are required to specify the class of  LWF chain graphs induced by a profile graph.


A joint probability distribution $P(Y_V,X)$ satisfies the LWF Global Markov property (LWF-GMP) wrt the LWF chain graph $C_U=[\{V,X\}, E_{C_U}]$ if \citep{frydenberg1990,drton2009discrete}: 
\begin{itemize}
	\item[] for any disconnected set $D \subseteq V$ with connected components $K_1,\dots,K_r$,
	\begin{equation}\label{eq:LWF global undirected}
	Y_{K_1} \bigCI \dotso \bigCI Y_{K_r}|\{Y_{V \setminus D},X\};
	\end{equation}
	\item[] for any subset $A \subseteq V$ such that there is a missing arrow between any vertex $a \in A$ and $X$,
	\begin{equation}\label{eq:LWF global arrow}
	Y_A \bigCI X | Y_{V \setminus A}.
	\end{equation}
\end{itemize}
We remark that  Equation \eqref{eq:LWF global undirected} directly derives from Theorem \ref{t1}.
\ml{A definition of Markov-compatibility between a profile undirected  graph and an  LWF chain graphs is given.}
\ml{
\begin{definition}
    An LWF chain graph $C_U=[\{V,X\}, E_{C_U}]$ is Markov-compatible with a profile undirected graph $\GU$ if the LWF-GMP in \eqref{eq:LWF global undirected} for $C_U$ is implied by the $\U$-GMP for $\GU$.
\end{definition}
}
\ml{We now derive the class of LWF-graphs induced by a profile undirected graph such that Markov-compatibility is satisfied.
\begin{theorem}\label{thr2}
	Consider a profile undirected graph $\G_U =(V, \EU)$ associated to the profile outcome vectors $Y_{V|\X}$. If the probability distributions $P[Y_{V|\X}]$ satisfy the $\U$-GMP for $\G_U$,  then  also  the LWF-GMP in \eqref{eq:LWF global undirected}  is satisfied for the induced class $\C_U$ of two-block LWF chain graphs, where any $C_U=[\{V,X\},E_{C_U}]$ belongs to $\C_U$ if 
	\begin{itemize}
		\item[(i)] any couple $a,b \in V$ is joined by an undirected edge in $C_U$ if $\Z \subset \X$ for the  pair $(a,b)^{\Z} \in \EU$; 
		\item[(ii)] for any couple $a,b \in V$,  $a$ and $b$ are both reached by an arrow in $C_U$ starting from $X$  if $\Z \subset \X$ and $\Z \neq \emptyset$ for the pair $(a,b)^{\Z} \in \EU$.
	\end{itemize} 
\end{theorem}
}

Necessary conditions (i) and (ii) in Theorem \ref{thr2} 
ensure that it will always exist at least one Markov-compatible LWF chain graph for any given profile undirected graph, {specifically, the chain graph with no missing arrows is always compatible}. Condition (i) is related to the missing/non-missing undirected edges for any induced chain graph; it states that dotted and full edges in profile undirected graphs correspond to full edges in chain graphs. Condition (ii) is related to missing/non-missing directed edges for any induced chain graph; it states that vertices joined by a dotted edge in a profile undirected graph cannot be disjoined from $X$ in the induced chain graph. Since condition (ii) may not be intuitive, the following counterexample shows that this is a necessary condition.

{
\begin{example}\label{ex.nec-cond}
	Let $V=\{a,b,c\}$ be a set of response variables and $X$ a   factor with state-space $\X=\{0,1\}$ and let  $\GU=(V, \EU)$ be a profile undirected graph with
    $\EU=\{(a,b)^{\{0\}},(a,c)^{\emptyset},(b,c)^{\X}\}$ 
    where the pair $a,b$ is joined by a $\{0\}$-dotted edge that implies
	\begin{equation}
	\label{eqcompch2}
	Y_{a}(0) \ci Y_{b}(0)|Y_{c}(0) \qquad {\rm and} \qquad Y_{a}(1) \nci Y_{b}(1)|Y_{c}(1).
	\end{equation} We explore Markov-compatibility for some chain graphs. Consider a chain graph $C_{U}=\{(V,X), E_{C_U}\}$ with $E_{C_{U}}=\{(a,b),(b,c),(X,c)\}$, where vertices $a$ and $b$ are both disjoined from $X$. For the condition \eqref{eq:LWF global arrow}, we have $\{Y_a,Y_b\} \ci X | Y_{c}$, i.e., $$P(Y_{a}(0),Y_{b}(0)|Y_{c}(0)) = P(Y_{a}(1),Y_{b}(1)|Y_{c}(1)),$$ 
	that is not compatible with condition (\ref{eqcompch2}) for the profile graph, then $C_U$ does not belong to the induced class $\C_U$. Consider the chain graph $C^{'}_U$  with $E_{C^{'}_{U}}=\{(a,b),(b,c),(X,b),(X,c)\}$, where  only $a$ is joined to $X$. Equation \eqref{eq:LWF global arrow} implies that $Y_a \ci X|\{Y_b,Y_c\}$, that is, $$P[Y_a(0)|Y_b(0),Y_c(0)]=P[Y_a(1)|Y_b(1),Y_c(1)],$$  
    that is not compatible with condition \eqref{eqcompch2} for the profile graph implying that $P[Y_a(1)|Y_b(1),Y_c(1)] \neq [Y_a(0)|Y_b(0)Y_c(0)]$.  It follows that $C_U^{'}$ does not belong to the induced class $\C_U$. 
	Consider a third chain graph $C^{''}_U$  with the set $E_{C^{''}_{U}}=\{(a,b),(b,c),(X,a),(X,b),(X,c)\}$ of edges, where $a$ and $b$ are both joined to $X$. This chain graph is compatible with condition \eqref{eqcompch2} for the profile graph $\GU$. $C^{''}_U$ satisfies both conditions (i) and (ii) in Theorem \ref{thr2}  and belongs to the induced class $\C_U$. Consider the last chain graph $C^{'''}_U$ with the set $E_{C^{''}_{U}}=\{(a,b),(b,c),(X,a),(X,b),\}$ of edges, where only $a$ and $b$ are both joined to $X$. This graph implies $Y_c \ci X|\{Y_a,Y_b\}$, i.e., $P[Y_c(1)|Y_a(1),Y_b(1)] \neq [Y_c(0)|Y_a(0)Y_b(0)]$ that is compatible with condition \eqref{eqcompch2} for the profile graph $\GU$. $C^{'''}_U$ satisfies both conditions (i) and (ii) in Theorem \ref{thr2}. The induced class  $\C_U=\{C^{''}_U, C^{'''}_U\}$ of chain graphs includes two elements where $C^{'''}_U$ has no missing arrows and $C^{'''}_U$ has a missing arrow in correspondence of the vertex $c$ which is not involved in dotted edges.
    \end{example}
    }

In essence, given a profile  graph $\GU$, the induced class $\C_U$ includes  LWF chain graphs $C_U=\{(V,X),E_{C_U}\}$ where the chain component $V$ has the same skeleton of $\GU$, and differ only according to which arrows are missing. Within this class, we can identify the \textit{maximum element}, i.e., the chain graph with no missing arrows and the \textit{minimum element}, i.e., the chain graph with a set of arrows that point to all vertices $a \in V$ such that $nb_{\Z}(a) \neq \emptyset$ with $Z \subset \X$ and $Z \neq \emptyset$. 

\begin{figure}[t]
	
	\begin{center}
		\begin{tikzpicture}[node distance = {1.2cm and 1.2cm}, v/.style = {draw, circle,minimum size=0.7cm},amp/.style = {regular polygon, regular polygon sides=4,
			draw, fill=white, text width=0.4em,
			inner sep=1mm, outer sep=0mm,
			minimum size=0.9cm},mylabel/.style={thin,  align=center,fill=white,bend right}]
		
		\node at (1,-3.5) {$\G_U$};

		\node (a) [v] {a};
		\node (b) [v, below = of a] {b};
		\node (c) [v, right = of a] {c};
		\node (d) [amp, below = 1.24cm of c] {d};
		
		\draw (b) to node {} (d);
		
		\draw[ed] (a) to node [mylabel]{\tiny 2} (b);
		\draw[ed] (b) to node [sloped] [mylabel]{\tiny 1,2} (c);
		\draw[ed] (a) to node [sloped] [mylabel]{\tiny 0} (c);

		\end{tikzpicture} \hspace*{1.3cm}
		\begin{tikzpicture}[node distance = {1.2cm and 1.2cm}, v/.style = {draw, circle,minimum size=0.7cm}]

		\node at (0,-3.5) {$C_U$};
		
		\node (a) [v] {a};
		\node (b) [v, below = of a] {b};
		\node (c) [v, right = of a] {c};
		\node (d) [v, below = of c] {d};
		\node (x) [v, below left=0.5cm and 1.5cm of a] {X};
		
		\plate {}{(x)}{} 
		\plate {}{(a)(b)(c)(d)}{} 
		
		\draw  (a) to node {} (b);
		\draw  (b) to node {} (d);
		
		\draw  (b) to node {} (c);
		\draw  (a) to node {} (c);
		
		\draw[edf] (x) to node {} (a);
		\draw[edf] (x) to node {} (b);
		\draw[edf] (x) to node {} (c);

		\end{tikzpicture} 
		
		\caption{A profile undirected graph with a compatible LWF chain graph.}\label{fig:fig3} 
	\end{center}
\end{figure}
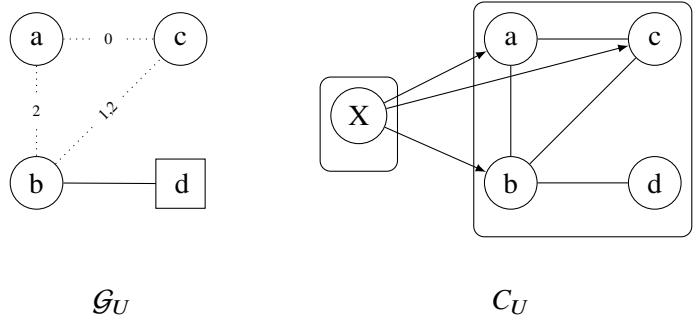


In order to account also for the LWF-GMP in \eqref{eq:LWF global arrow} and to establish a one-to-one relationship between profile undirected graphs and LWF chain graphs, we generalize the class of profile undirected graphs. Given a profile undirected graph $\GU=(V,\EU)$, consider the partition $V=V \circleee \cup V_{\square}$ of the vertex set so that we distinguish between two types of vertices, a \textit{circle vertex} $a \hspace*{0.4mm} \text{\circleee} \in V \circleee$ and a \textit{square vertex} $a_{\square} \in V_{\square}$, drawn as $\Circle$ and $\square$, respectively. For every $a_{\square} \in V_{\square}$, we assume $Y_a \bigCI X|Y_{V \setminus a}$; that is the univariate profile distribution of $Y_a(x)$ is invariant for any $x \in \X$, given the remaining variables $Y_{V\setminus a}$; 
otherwise if $a \hspace*{0.4mm} \text{\circleee} \in V \circleee$, we assume $Y_a \nci X|Y_{V \setminus a}$.

The profile graph in this generalized representation includes information also about the independence structure between subsets of response variables $Y_A$, with $A \subseteq V$, and the external factor $X$. In particular, for any $A \subseteq V_{\square}$, we assume that $Y_A \bigCI X|Y_{V \setminus A}$. Then, given a profile undirected graph $\GU=(V \circleee, V_{\square}, \EU)$, the compatible two-block LWF chain graph $C_U=[\{V,X\},E_{C_U}]$ in the class $\C_U$ is unique and is defined by a chain graph where the undirected graph of the response component $V$ has the same skeleton of $\GU$ and there are missing arrows between $X$ and any square vertex $a_{\square} \in V_{\square}$. 
{Square and circle vertices provide insights into how the dependence structure varies across profiles. Square vertices identify variables with stable pairwise dependencies, while circular vertices denote variables characterized by pairwise context-specific independencies. The fewer the square nodes, the greater the difference across profiles. }

\begin{example}Consider the profile undirected graph and the induced chain graph in Figure \ref{fig:fig3}. Vertices $a,b,c$ are circled vertices while $d$ is a square vertex wrt $\GU$, i.e., $\{a,b,c\} \in V \circleee$ and $d \in V_{\square}$. Then, both the profile undirected graph and the chain graph imply the  independence statement $Y_d \bigCI X | \{Y_a,Y_b,Y_c\} $. Also, both  graphs imply that $\{Y_a,Y_c\} \ci Y_{d}|\{Y_{b},X\}$. Unlike the profile graph, the chain graph does not provide information about the effect of $X$ on the $Y_V$ association structure, e.g., $Y_a(2) \ci Y_b(2)| \{Y_c(2),Y_d(2)\}$.
\end{example}

\section{Gaussian profile undirected graphical model}
\label{sec:Gaussian}

We can now define the class of \emph{Gaussian profile undirected graphical model} by imposing zero-constraints over the model parameters; these constraints naturally follow from the Markov equivalence between profile graphs and multiple graphs, and the compatibility between profile graphs and chain graphs established previously. 

For all $x \in \X$, let  $Y_{V}(x) \sim N(\alpha+\beta_{x},\Sigma_x)$ where $[\alpha_{a}+\beta_{ax}]_{a \in V}=\mathbb{E}[Y_{a}(x)]_{a \in V}$ is the \emph{profile marginal mean vector} and $\Sigma_x$ is the \emph{profile covariance matrix} with entries $\sigma_{ab,x}$, $x \in \X$.
Let $\zeta_{ax}, \, a \in V$, be the linear effect of the external factor on the \emph{profile conditional mean vector} $\mathbb{E}[Y_{a}(x)|Y_{V \setminus a}(x)]_{a \in V}$ and let $\Omega_x = \Sigma_x^{-1}$ be the \emph{profile precision matrix} with entries $\omega_{ab,x}$, $x \in \X$; note that $\zeta_{x}=\Omega_x\beta_{x}$ where $\zeta_{x}=[\zeta_{ax}]_{a \in V}$ and $\beta_{x}=[\beta_{ax}]_{a \in V}$ \citep{andersson2001alternative}. 
\ml{
\begin{definition}
    The Gaussian profile undirected graphical model for $Y_{V|\X}$ wrt $\GU=(V, \EU)$  is such that, 
\begin{itemize}
	\item[(i)] for any $a \in V_{\square}$,  $\zeta_{ax}=0$ for all $ x \in \X$,
	\item[(ii)] for any $(a,b)^{\Z} \in \EU$, with $\Z \subseteq \X$,  $\omega_{ab,x}=0$, for each $x \in \Z$.
\end{itemize}
\end{definition}
}
Estimation of Gaussian profile undirected graphical models can build on existing methods for Gaussian chain or multiple graph inference. Penalty terms that promote shared network structures across profiles—similar to the group graphical lasso, joint graphical lasso \citep{danaher2014joint}, or GemBag \citep{YangXinming2021GGEo} may be appropriate when the x-profile vectors follow distributions with comparable graph structures. In Section \ref{Section:BayesianModel}, we extend the model of \citet{YangXinming2021GGEo} by introducing profile indicators $r_{ab,x}$ that specify the sparsity of the corresponding entries $\omega_{ab,x}$ and profile coefficients $\beta_{a,x}$. 
The entries $\omega_{ab,x}$ estimated to be zero define the zero patterns that correspond to a given profile undirected graph $\GU=(V,\EU)$.

\subsection{Bayesian model formulation}
\label{Section:BayesianModel}
To express our model in a Bayesian framework, we first introduce binary global-level indicators $\gamma_{ij}$ such that $\gamma_{ij}:=1$ indicates that nodes $i$ and $j$ are connected in at least one of the $q$ graphs, and $\gamma_{ij}:=0$ denotes no such connection. We place independent Bernoulli priors $\gamma_{ij}\sim \text{Bernoulli}(\gamma_{ij} \mid p_1)$ on these indicators.

Similarly, we introduce global-level coefficient indicators $\theta_{i}$. The indicator $\theta_{i}$ denotes whether at least one of the coefficients {that represent the effect of the external factor} is non-zero ($\theta_{i}:=1$) or zero ($\theta_{i}:=0$). These indicators will also affect the profile-level sparsity in the corresponding entries of the precision matrices across the $q$ graphs. We assign a Bernoulli prior $\theta_{i}\sim \text{Bernoulli}(\theta_{i} \mid p_2)$ at each indicator.

We introduce profile indicators $r_{ij,x}$ to capture the sparsity of the corresponding entries $\omega_{ij,x}$. To encourage similarity among these indicators, we place joint priors on $r_{ij,0},\dots,r_{ij,q-1} \mid \gamma_{ij}$. 
These distributions encourage across-profile information sharing while allowing within-profile heterogeneity. 
In this paper, we assume the following hierarchical structure for this distribution:
%
\begin{equation}
\begin{split}
\label{eq:rx}
\Prob(r_{ij,x},\dots,r_{ij,q-1} \mid \gamma_{ij}, \theta_i, \theta_j) =& \gamma_{ij}\theta_i \theta_j \prod_{x = 0}^{q-1}\text{Bernoulli}(r_{ij,x} \mid p_3) \\
&+ (1-\gamma_{ij})\prod_{x = 0}^{q-1}\delta_0(r_{ij,x}) \\
&+ \gamma_{ij} (1-\theta_i \theta_j) \text{Bernoulli}({r}_{ij,0} \mid p_4) \delta_{(r_{ij,1},\dots,r_{ij,q-1})}(r_{ij,0}),
\end{split}
\end{equation}
where $\delta_{\cdot}(\cdot)$ denotes a point mass at zero. 
Under this setup, when the global-level indicator $\gamma_{ij}=0$, all the ${r}_{ij,k}$ are set to zero. 
When $\gamma_{ij}=1$, each ${r}_{ij,k}$ can still independently take the value 0 with probability $1-p_3$ if $\theta_i=\theta_j=1$, or with probability $1-p_4$ if $\theta_i=\theta_j=0$ and $r_{ij,0}=r_{ij,1}=,\dots,=r_{ij,q-1}$.
This provides a flexible approach to sharing information across profiles while also accounting for the effect of the external factor, e.g., the inclusion or exclusion of arrows.

We place an exponential prior on the positive diagonal entries of the $q$ precision matrices to induce proper shrinkage: $\omega_{ii,x}\sim \text{Exponential}(\tau)$, and a spike-and-slab prior on the upper triangular entries $\omega_{ij,x}(i<j)$: $\omega_{ij,x}\mid r_{ij,x} \sim r_{ij,x} \Prob(\omega_{ij,x}\mid \nu_1) +(1- r_{ij,x}) \Prob(\omega_{ij,x}\mid \nu_0)$, with $\nu_1>\nu_0>0$. 
Following \citet{YangXinming2021GGEo}, we adopt a spike-and-slab Lasso prior \citep{Rockova2018}, where $\Prob(\omega_{ij,x}\mid \nu_1)$ represents the slab component with a large variance, allowing for large signals, and $\Prob(\omega_{ij,x}\mid \nu_0)$ represents the spike component with a small variance, encouraging values close to zero. Finally we set a normal spike-and-slab prior \citep{George93} on the profile coefficients: $\beta_{i,x}\mid \theta_i \sim \theta_i \Prob(\beta_{i,x}\mid \lambda_1) +(1- \theta_i) \Prob(\beta_{i,x}\mid \lambda_0)$, with $\lambda_1>\lambda_0>0$. 
Figure~\ref{fig:DAGpugm} provides a graphical representation of our Bayesian model.

\begin{figure}[h!]
\centering
\resizebox{0.5\textwidth}{!}{
  \begin{tikzpicture}
\node[obs] (y) {$Y_{v_x}$} ;
\node[latent, right=of y] (alpha) {$\alpha$} ;
\node[latent, left=of y] (omegaij) {$\omega_{ij,x}$} ;
\node[latent, above=of omegaij] (omegaii) {$\omega_{ii,x}$} ;
\node[latent, left=of omegaij] (r) {$r_{ij,x}$} ;
\node[latent, left=of r] (gamma) {$\gamma_{ij}$} ;
\node[latent, below=of omegaij] (beta) {$\beta_{i,x}$} ;
\node[latent, below=of gamma] (theta) {$\theta_{i}$} ;
\node[const, above=of omegaii] (tau) {$\tau$} ;
\node[const, left=of tau] (p4) {$p_4$} ;
\node[const, left=of p4] (p3) {$p_3$} ;
\node[const, below=of alpha] (l0) {$\lambda_0$} ;
\node[const, below=of l0] (l1) {$\lambda_1$} ;
\node[const, right=of tau] (v0) {$\nu_0$} ;
\node[const, right=of v0] (v1) {$\nu_1$} ;
\node[const, left=of theta] (p2) {$p_2$} ;
\node[const, left=of gamma] (p1) {$p_1$} ;
\plate[inner sep=0.1cm] {plate1} {(y) (omegaij) (omegaii) (beta) (r)} {$0\leq x \leq q-1$};
\plate[inner sep=0.1cm] {plate4} {(beta) (omegaij) (omegaii)(theta)(r) (plate1)} {$1 \leq i < j \leq p$}
\edge {beta,alpha,omegaij,omegaii} {y} ;
\edge {theta,gamma} {r} ;
\edge {tau} {omegaii} ;
\edge {theta} {beta} ;
\edge {p3, p4} {r} ;
\edge {l0, l1} {beta} ;
\edge {v0, v1} {omegaij} ;
\edge {r} {omegaij} ;
\edge {p1} {gamma} ;
\edge {p2} {theta}
\end{tikzpicture}
}
  \caption{Directed acyclic graph (DAG) for Gaussian undirected profile graphical model with Spike-and-slab prior on the profile covariance matrices and the profile coefficients}
\label{fig:DAGpugm}
  \end{figure}
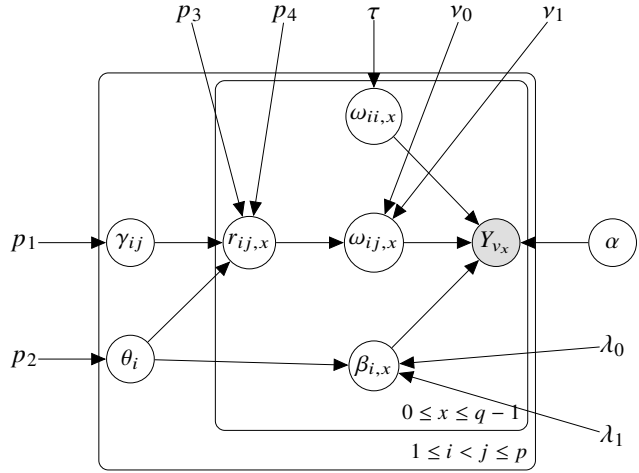

%
Parameter estimation for our proposed Bayesian Gaussian profile undirected graphical model is conducted using Expectation-Maximisation (EM) algorithm \citep{Dempster1977} along the lines of the EM algorithm of \citet{YangXinming2021GGEo}. 
Let $\Delta = (\alpha, \beta_x, \Omega_x)$ be the unknown parameters.
The EM algorithm aims to maximise the log-posterior $\log(\Prob(\Delta \mid Y_{V_x})$ by working with the complete-data log-posterior $\log(\Prob(\Delta \mid Y_{V_x},\theta,R)$.
The EM algorithm 
has two steps: the E-step calculates the expected values $\Exp[r_{ij,x}\mid \widehat{\Delta},Y_{V_x}] $ and $\Exp[\theta_i \mid \widehat{\Delta},Y_{V_x}] $ where $\Delta^{(t)}=\widehat{\Delta}$ are the current values at $t$ of $\Delta$. 
The M-step maximises $Q(\widehat{\Delta}) = \Exp_{\theta,R \mid Y,\widehat{\Delta}}[\log(\Prob(\widehat{\Delta} ,\theta,R\mid Y_{V_x})]= \langle \log(\Prob(\widehat{\Delta},\theta,R \mid Y_{V_x}) \rangle$  w.r.t. $\widehat{\Delta}=(\widehat{\alpha}, \widehat{\beta}_x, \widehat{\Omega}_x)$, giving new updates $\Delta^{(t+1)}=(\alpha^{(t+1)}, \beta_x^{(t+1)}, \Omega_x^{(t+1)})=\widehat{\Delta}$.
Supplementary Material~\ref{App:EM} provides a detail derivation of our EM algorithm. 

\section{Results}
\subsection{Simulations}
We conduct simulation studies to evaluate the performance of the proposed Bayesian Gaussian undirected profile graphical model (BPUGM). 
We consider $q = 4$ levels of the covariate $X$, with $X = \{0,1,2,3\}$, and examine performance under varying numbers of nodes $p \in \{20,50,100\}$ and three levels of sparsity $s$ (the larger the value of $s$, the sparser the graphs).
Four structural scenarios are investigated: (i) all four levels have distinct graph structures; (ii) $\{G(0)=G(1)\} \neq \{G(2)=G(3)\}$; (iii) $\{G(0)=G(1)=G(2)\} \neq \{G(3)\}$; and (iv) all levels share the same graph structure. 
For each scenario, data are generated independently for each $x \in \mathcal{X}$ from a multivariate normal distribution $\mathcal{N}(\beta_x,\Sigma_x)$, where $\Sigma_x=\Omega_x^{-1}$ and $\beta_x=\Sigma_x\zeta_x$, with sample size $n_x=50$. 
Additional details on the data generating mechanism are provided in the Supplementary Materials. 

We compare BPUGM with GemBag \citep{YangXinming2021GGEo}, the fused graphical lasso (FGL), and the group graphical lasso (GGL) \citep{danaher2014joint}. 
Graph structure estimation accuracy is assessed using the area under the ROC curve (AUC), accuracy, sensitivity, and specificity of $\Lambda(x)$ for all $x \in \mathcal{X}$. 
Results are averaged over 100 simulated datasets, with standard errors reported. Summaries are presented in Figures~4 to 6 and detailed values in Tables~1 to 4 of the Supplementary Materials.

Across all scenarios, values of $p$, and sparsity levels, BPUGM consistently achieves the highest AUC, indicating superior edge discrimination compared with competing methods. 
The improvement is most pronounced in Scenarios~2 and~3, where partial sharing of graph structures across levels allows BPUGM to effectively borrow strength across groups. 
In Scenario~1, where graph structures are fully distinct, BPUGM remains competitive and stable as $p$ increases. In Scenario~4, where all graphs are identical, BPUGM performs comparably to FGL and GGL, while maintaining slightly higher sensitivity without loss of specificity.
Overall, BPUGM provides a favorable balance between sensitivity and specificity, leading to higher accuracy. 
These results demonstrate the advantages of the proposed Bayesian profile graphical modeling framework for joint estimation of multiple related graph structures under finite-sample conditions.

\begin{figure}
  \includegraphics[width=\linewidth]{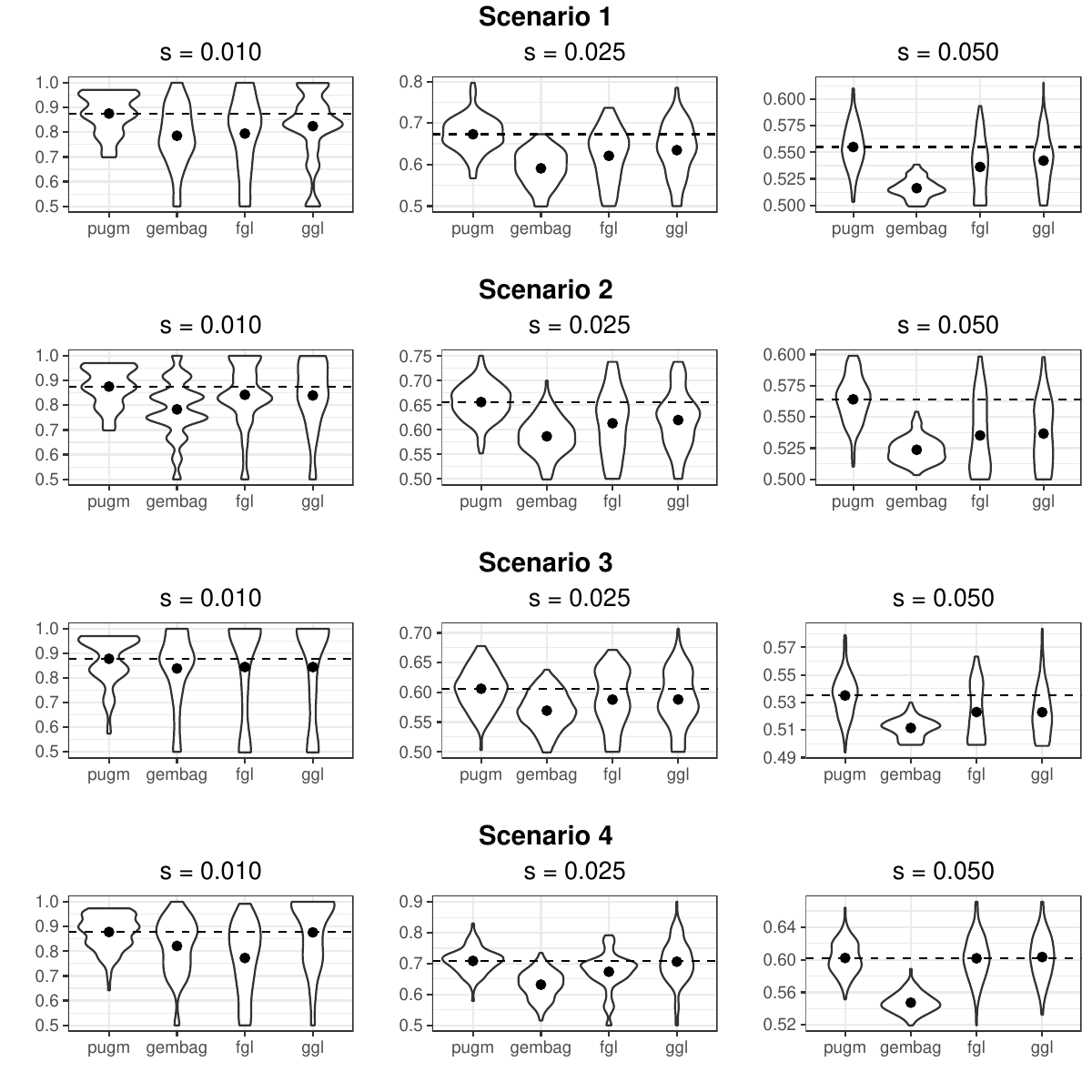}
  \caption{AUC over  \textbf{100} datasets, $P=20$, $K=4$, $N_x=20$ for the four different scenarios}
  \label{fig:P20}
\end{figure}

\begin{figure}
  \includegraphics[width=\linewidth]{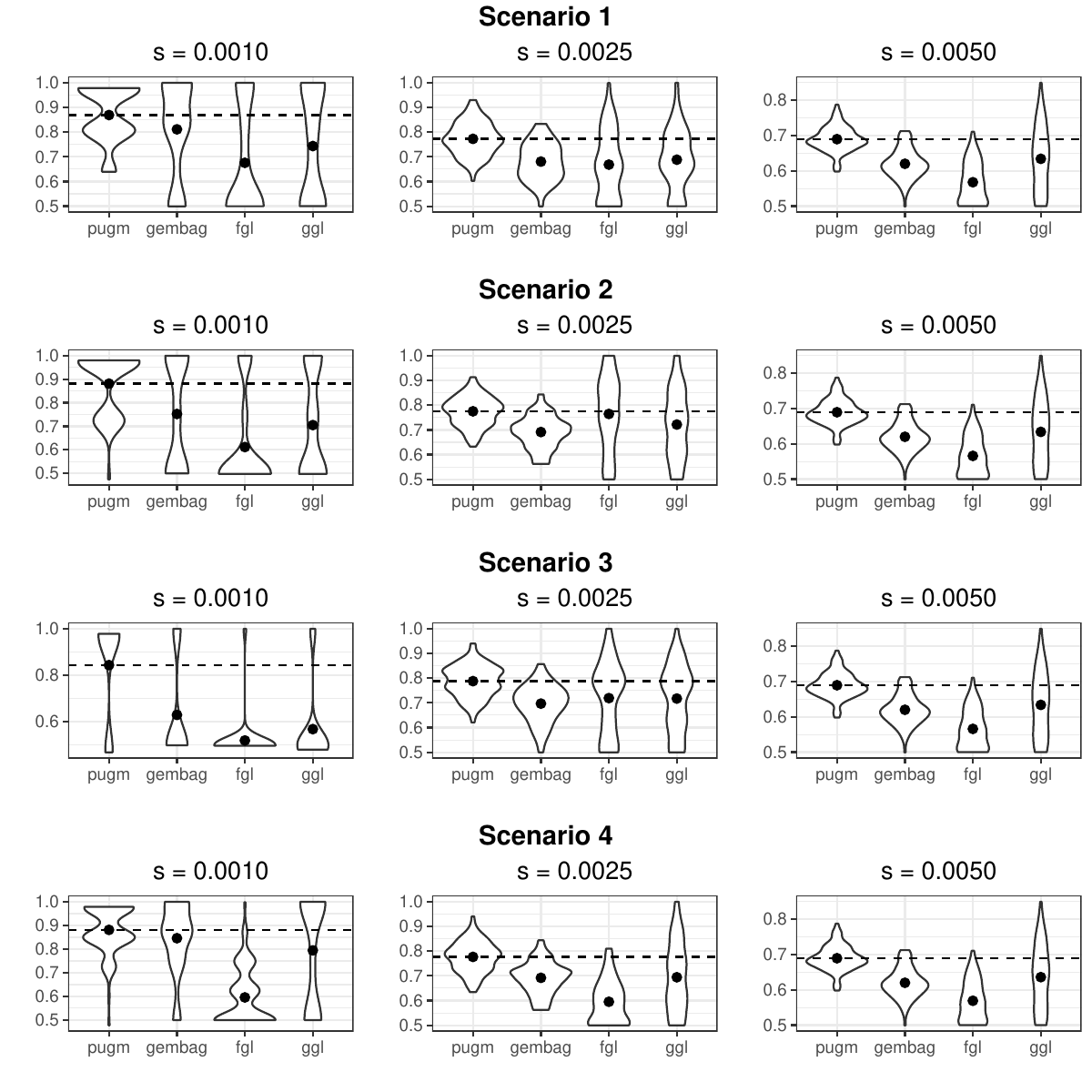}
  \caption{AUC over  \textbf{100} datasets, $P=50$, $K=4$, $N_x=50$ for the four different scenarios}
  \label{fig:P50}
\end{figure}

\begin{figure}
  \includegraphics[width=\linewidth]{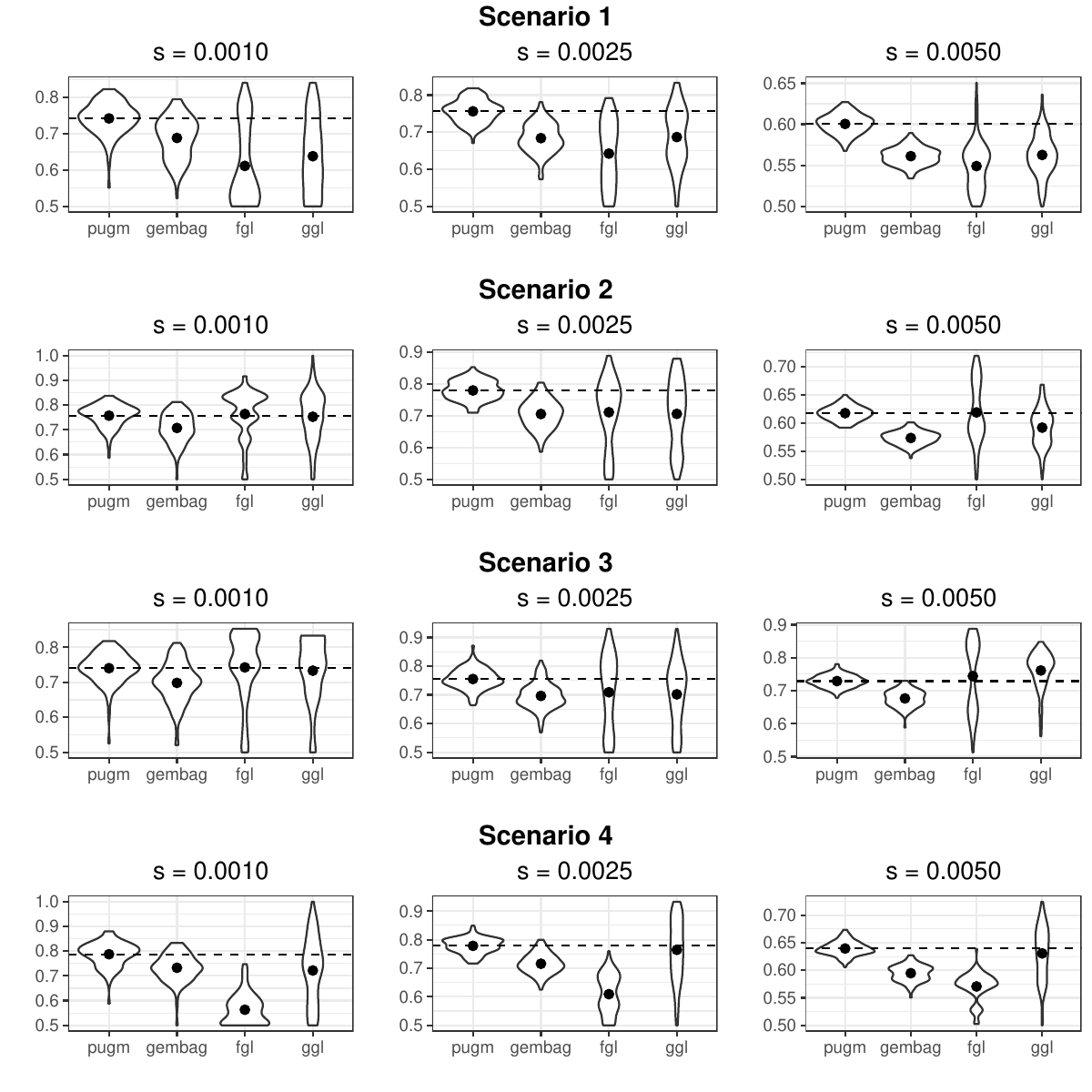}
  \caption{AUC over  \textbf{100} datasets, $P=100$, $K=4$, $N_x=50$ for the four different scenarios}
  \label{fig:P100}
\end{figure}

\subsection{AML protein data}

We analyze protein expression data from patients affected by acute myeloid leukemia (AML) with the goal of reconstructing and comparing protein networks across disease subtypes; comparing the networks for these groups provides insight into the differences in protein signaling that may affect whether treatments for one subtype will be effective in another one.

A set of protein levels, collected using the reverse phase protein array (RPPA) technology, is observed in a sample of 213 newly diagnosed AML patients \citep{kornblau2009functional}\footnote{{http://bioinformatics.mdanderson.org/Supplements/Kornblau-AML-RPPA/aml-rppa.xls}}.  Patients are classified by subtype according to the French-American-British (FAB) classification system. We consider 4 different profiles given by 4 AML subtypes, for which a reasonable sample size is available: M0 (17 subjects), M1 (34 subjects), M2 (68 subjects), and M4 (59 subjects). These profiles, based on criteria including cytogenetics and cellular morphology, show varying prognosis. We expect to observe different protein interactions in the subtypes. We focus on 18 proteins relevant to the apoptosis and cell cycle regulation KEGG pathways \citep{kanehisa2011kegg}. 

Our interest is modelling the effect of the AML subtype on the joint independence structure of the protein levels. Profile undirected graphical models are an encompassing tool that coherently and jointly performs all inferential tasks of interest of learning how the protein dependency structure changes across subtypes as well as the mean protein levels. Therefore, considering the $p=18$ protein levels following a multivariate Gaussian distribution and $q=4$ different profiles of AML, where the levels $x \in \X=\{0,1,2,3\}$ denote the subtypes M0, M1, M2, M4 respectively, we estimate and select the profile undirected graphical model represented in Figure \ref{fig:AML01}. For the sake of comparison, we represent the corresponding multiple-graph in Figure \ref{fig:MultiGemBag}; this graph is arguably harder to read. Most importantly, the many profile specific independencies are obviously missed by the graph.

For instance, from the selected profile graph we learn that for  the profiles $x \in \{0,1\}$, $Y_{\text{AKTp.308}}(x) \bigCI Y_{\text{BCI.2}}(x)$ $|Y_{V\setminus \{\text{AKTp.308,BCI.2}\}}(x)$; for any profile $x \in \X$, $Y_{\text{AKTp.308}}(x) \bigCI Y_{\text{BAD}}(x)|Y_{V\setminus \{\text{AKTp.308,BAD}\}}(x)$ . The level of proteins BAX, GSK3 and XIAP 
are independent to the AML subtypes. 

The selected profile graph has only three square vertices and only three arrows can be removed as all edges are dotted with the exception of one. This means that  the dependence structure of $Y_V$ is expected to substantially vary across profile and  modeling the direct effect of $X$ on single $Y_a \in Y_V$  becomes relevant to capture how the conditional dependence structure of $Y_V$ changes in different profiles. \aap{We then compute the maximal connected components of each profile graph, corresponding to the set of maximal paths. These quantities summarise the heterogeneity in protein connectivity and path structure across profiles. The cardinalities of the resulting sets differ across the four levels, taking values 4, 9, 13 and 12, respectively, indicating variation in structural complexity between profiles.}

The estimated profile graphs obtained with our method exhibit varying levels of sparsity across groups, compared with those from the GemBag approach. 
For groups 0 and 1, the inferred graphs are substantially sparser than those produced by GemBag, suggesting that explicitly modeling the effect of $X$ on $Y_{V|\X}$ may explain a large portion of the dependence structure within $Y_{V|\X}$, thereby reducing the number of residual conditional associations. 

Table \ref{Tab:GemBagVsBPUGM} provides a direct comparison of edge detection results between GemBag and BPUGM. For groups 0 and 1, all edges identified by BPUGM are also detected by GemBag, while for groups 2 and 3 the majority of BPUGM edges are likewise recovered. 
This pattern suggests that BPUGM yields a more parsimonious representation of the conditional dependence structure.
Moreover, group-specific graphs generated by our method (Figure \ref{fig:MultiPUGM}) appear more heterogeneous than those obtained with GemBag (Figure \ref{fig:MultiGemBag}), indicating a greater ability to capture group-specific differences. 

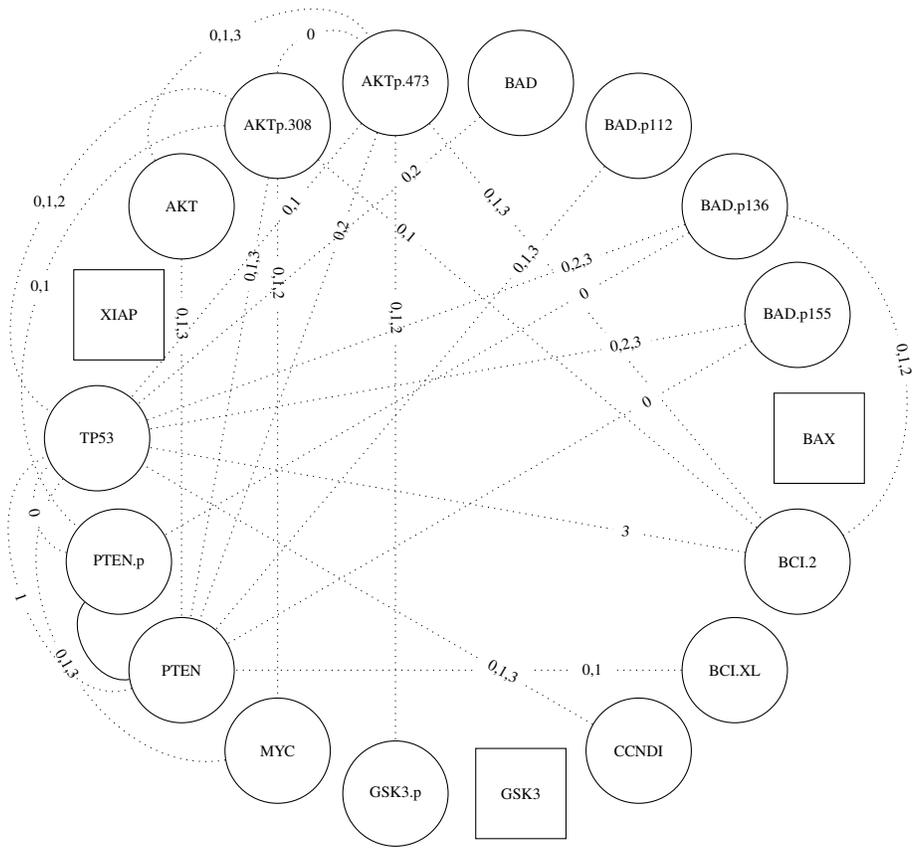
\begin{figure}
	\begin{center}
		\hspace*{-1cm}
		\begin{tikzpicture}[scale=1.2,mylabel/.style={thin,  align=center,fill=white,bend right}]
		
		\node[vc] (1) at (7*360/18: 4cm) {\tiny AKT};
		\node[vc] (2) at (6*360/18: 4cm) {\tiny AKTp.308};
		\node[vc] (3) at (5*360/18: 4cm) {\tiny AKTp.473};
		\node[vc] (4) at (4*360/18: 4cm) {\tiny BAD};
		\node[vc] (5) at (3*360/18: 4cm) {\tiny BAD.p112};
		\node[vc] (6) at (2*360/18: 4cm) {\tiny BAD.p136};
		\node[vc] (7) at (1*360/18: 4cm) {\tiny BAD.p155};
		\node[vs] (8) at (18*360/18: 4cm) {\tiny BAX};
		\node[vc] (9) at (17*360/18: 4cm) {\tiny BCI.2};
		\node[vc] (10) at (16*360/18: 4cm) {\tiny BCI.XL};
		\node[vc] (11) at (15*360/18: 4cm) {\tiny CCNDI};
		\node[vs] (12) at (14*360/18: 4cm) {\tiny GSK3};
		\node[vc] (13) at (13*360/18: 4cm) {\tiny GSK3.p};
		\node[vc] (14) at (12*360/18: 4cm) {\tiny MYC};
		\node[vc] (15) at (11*360/18: 4cm) {\tiny PTEN};
		\node[vc] (16) at (10*360/18: 4cm) {\tiny PTEN.p};
		\node[vc] (17) at (9*360/18: 4cm) {\tiny TP53};
		\node[vs] (18) at (8*360/18: 4cm) {\tiny XIAP};
		
		\draw[ed]  (1) to [bend left=90,pos=0.5] node[mylabel]{\tiny 0,1,3} (3);
		\draw[ed]  (1) to [sloped,pos=0.18] node[mylabel]{\tiny 0,1,3} (15);

		\draw[ed]  (2) to [bend left=70,pos=0.5] node[mylabel]{\tiny 0} (3);
		\draw[ed]  (2) to [sloped,pos=0.2] node[mylabel]{\tiny 0,1} (9);
        \draw[ed]  (2) to [sloped,pos=0.2] node[mylabel]{\tiny 0,1,2} (14);
        \draw[ed]  (2) to [sloped,pos=0.2] node[mylabel]{\tiny 0,1,3} (15);
        \draw[ed]  (2) to [bend right=70,pos=0.5] node[mylabel]{\tiny 0,1} (16);
        \draw[ed]  (2) to [bend right=90,pos=0.5] node[mylabel]{\tiny 0,1,2} (17);
		
		\draw[ed]  (3) to [sloped,pos=0.2] node[mylabel]{\tiny 0,1,3} (9);
        \draw[ed]  (3) to [sloped,pos=0.3] node[mylabel]{\tiny 0,1,2} (13);
        \draw[ed]  (3) to [sloped,pos=0.2] node[mylabel]{\tiny 0,2} (15);
		\draw[ed]  (3) to [sloped,pos=0.3] node[mylabel]{\tiny 0,1} (17);

        \draw[ed]  (4) to [sloped,pos=0.2] node[mylabel]{\tiny 0,2} (17);

        \draw[ed]  (5) to [sloped,pos=0.2] node[mylabel]{\tiny 0,1,3} (15);

       \draw[ed]  (6) to [bend left=70,sloped,pos=0.5] node[mylabel]{\tiny 0,1,2} (9);
		
		\draw[ed]  (6) to [sloped,pos=0.2] node[mylabel]{\tiny 0} (16);
        \draw[ed]  (6) to [sloped,pos=0.2] node[mylabel]{\tiny 0,2,3} (17);
		
		\draw[ed]  (7) to [sloped,pos=0.2] node[mylabel]{\tiny 0} (15);

        \draw[ed]  (7) to [sloped,pos=0.2] node[mylabel]{\tiny 0,2,3} (17);

        \draw[ed]  (9) to [sloped,pos=0.2] node[mylabel]{\tiny 3} (17);

        \draw[ed]  (10) to [sloped,pos=0.2] node[mylabel]{\tiny 0,1} (15);

        \draw[ed]  (11) to [sloped,pos=0.2] node[mylabel]{\tiny 0,1,3} (17);

        \draw[ed]  (14) to [bend left=70,sloped,pos=0.5] node[mylabel]{\tiny 0,1,3} (17);
        
		\draw[ef] [bend left=70] (15) to (16); 
        \draw[ed]  (15) to [bend left=90,sloped,pos=0.5] node[mylabel]{\tiny 1} (17);

        \draw[ed]  (16) to [bend left=70,sloped,pos=0.5] node[mylabel]{\tiny 0} (17);

		\end{tikzpicture}
	\end{center}	
	\caption{The selected profile undirected graph model for protein data}\label{fig:AML01}
\end{figure}

\begin{table}[h]
\centering
\scalebox{.9}{
\begin{tabular}{l|cc|cc|cc|cc|cc}
\hline
 & \multicolumn{8}{c}{BPUGM} \\
 \hline
 & \multicolumn{2}{c|}{$U(0)$}& \multicolumn{2}{c|}{$U(1)$} & \multicolumn{2}{c|}{$U(2)$}& \multicolumn{2}{c|}{$U(3)$}& \multicolumn{2}{c}{Total}\\
 \hline
GemBag & Edge & No edge & Edge & No edge & Edge & No edge & Edge & No edge & Edge & No edge \\
\hline
Edge    & 3   & 12  & 10 &   4  &  14 &   0 &  13 &   0&  40 &  16\\
No edge & 0   & 138 & 0 & 139 &   4 & 135&   3 & 137&   7 & 549\\
\hline
\end{tabular}}
\caption{Comparison of GemBag and BPUGM edge detection}
\label{Tab:GemBagVsBPUGM}
\end{table}

Finally, we repeated the experiment 100 times, randomly removing 10\% of the data each time to assess the robustness of the inferred models. 
For each repetition, we computed the balanced accuracy between the multiple-graphs obtained from the reduced dataset and the multiple-graphs estimated from the full dataset, to be considered as the ground truth for this metric. 
Table \ref{Tab:GemBagVsBPUGM02} shows that BPUGM exhibited higher overall robustness than GemBag, achieving a higher overall balanced accuracy (0.9736 vs. 0.9339). 
In particular, BPUGM substantially outperformed GemBag in the two smallest-sample groups ($U(0)$ and $U(1)$), where it inferred sparser graphs, while performance was comparable between methods in the larger groups ($U(2)$ and $U(3)$).
This consistent improvement demonstrates the advantage of BPUGM in model selection, highlighting its greater stability and ability to recover graph structures closer to the original model under data perturbations.

\begin{table}[ht]
\centering
\begin{tabular}{c|rr|rr}
\hline
  & \multicolumn{2}{c|}{\textbf{BPUGM}}  & \multicolumn{2}{c}{\textbf{GemBag}}\\ 
 \hline
  & Mean & SE & Mean & SE \\ 
  \hline
$U(0)$ & 0.9803 & 0.0425 & 0.8607 & 0.0502 \\ 
$U(1)$ & 0.9622 & 0.0395 & 0.9509 & 0.0394 \\ 
$U(2)$ & 0.9772 & 0.0172 & 0.9736 & 0.0262 \\ 
$U(3)$ & 0.9758 & 0.0235 & 0.9829 & 0.0210 \\ 
Overall & 0.9736 & 0.0143 & 0.9339 & 0.0192 \\ 
   \hline
\end{tabular}
\caption{Mean balanced accuracy and standard error (SE) between the multiple-graphs inferred from reduced datasets and those estimated from the full dataset, computed across 100 repetitions, comparing BPUGM and GemBag. }
\label{Tab:GemBagVsBPUGM02}
\end{table}

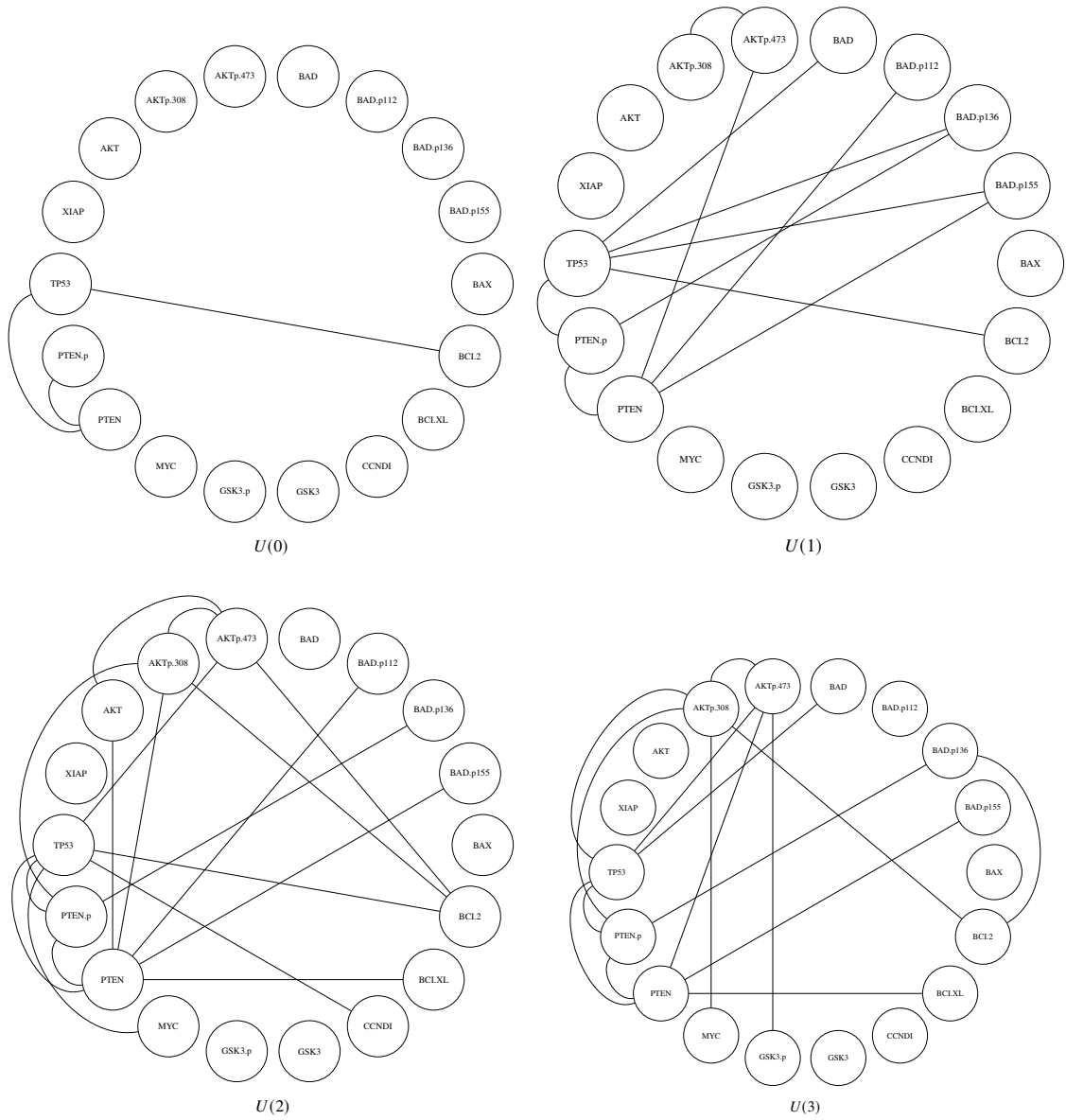
\begin{figure}[t]
	\begin{center}
    \resizebox{.45\textwidth}{!}{
    \begin{tikzpicture}[scale=1.2,mylabel/.style={thin,  align=center,fill=white,bend right}]
		
		\node[vc] (1) at (7*360/18: 4cm) {\tiny AKT};
		\node[vc] (2) at (6*360/18: 4cm) {\tiny AKTp.308};
		\node[vc] (3) at (5*360/18: 4cm) {\tiny AKTp.473};
		\node[vc] (4) at (4*360/18: 4cm) {\tiny BAD};
		\node[vc] (5) at (3*360/18: 4cm) {\tiny BAD.p112};
		\node[vc] (6) at (2*360/18: 4cm) {\tiny BAD.p136};
		\node[vc] (7) at (1*360/18: 4cm) {\tiny BAD.p155};
		\node[vc] (8) at (18*360/18: 4cm) {\tiny BAX};
		\node[vc] (9) at (17*360/18: 4cm) {\tiny BCI.2};
		\node[vc] (10) at (16*360/18: 4cm) {\tiny BCI.XL};
		\node[vc] (11) at (15*360/18: 4cm) {\tiny CCNDI};
		\node[vc] (12) at (14*360/18: 4cm) {\tiny GSK3};
		\node[vc] (13) at (13*360/18: 4cm) {\tiny GSK3.p};
		\node[vc] (14) at (12*360/18: 4cm) {\tiny MYC};
		\node[vc] (15) at (11*360/18: 4cm) {\tiny PTEN};
		\node[vc] (16) at (10*360/18: 4cm) {\tiny PTEN.p};
		\node[vc] (17) at (9*360/18: 4cm) {\tiny TP53};
		\node[vc] (18) at (8*360/18: 4cm) {\tiny XIAP};

        \node at (0,-5) {$U(0)$};
		
        \draw[ef]  (9) to (17);

        \draw[ef]  (15) to [bend left=70]  (16); 
        \draw[ef]  (15) to [bend left=90] (17);
        
        \end{tikzpicture}
    }
    \resizebox{.45\textwidth}{!}{
    \begin{tikzpicture}[scale=1.2,mylabel/.style={thin,  align=center,fill=white,bend right}]
		
		\node[vc] (1) at (7*360/18: 4cm) {\tiny AKT};
		\node[vc] (2) at (6*360/18: 4cm) {\tiny AKTp.308};
		\node[vc] (3) at (5*360/18: 4cm) {\tiny AKTp.473};
		\node[vc] (4) at (4*360/18: 4cm) {\tiny BAD};
		\node[vc] (5) at (3*360/18: 4cm) {\tiny BAD.p112};
		\node[vc] (6) at (2*360/18: 4cm) {\tiny BAD.p136};
		\node[vc] (7) at (1*360/18: 4cm) {\tiny BAD.p155};
		\node[vc] (8) at (18*360/18: 4cm) {\tiny BAX};
		\node[vc] (9) at (17*360/18: 4cm) {\tiny BCI.2};
		\node[vc] (10) at (16*360/18: 4cm) {\tiny BCI.XL};
		\node[vc] (11) at (15*360/18: 4cm) {\tiny CCNDI};
		\node[vc] (12) at (14*360/18: 4cm) {\tiny GSK3};
		\node[vc] (13) at (13*360/18: 4cm) {\tiny GSK3.p};
		\node[vc] (14) at (12*360/18: 4cm) {\tiny MYC};
		\node[vc] (15) at (11*360/18: 4cm) {\tiny PTEN};
		\node[vc] (16) at (10*360/18: 4cm) {\tiny PTEN.p};
		\node[vc] (17) at (9*360/18: 4cm) {\tiny TP53};
		\node[vc] (18) at (8*360/18: 4cm) {\tiny XIAP};

        \node at (0,-5) {$U(1)$};
		
		\draw[ef]  (2) to [bend left=70] (3);
		
		\draw[ef]  (3) to (15);
		
        \draw[ef]  (4) to (17);

        \draw[ef]  (5) to (15);

        \draw[ef]  (6) to (16);
        \draw[ef]  (6) to (17);
		
		\draw[ef]  (7) to (15);
        \draw[ef]  (7) to (17);

        \draw[ef]  (9) to (17);

        \draw[ef]  (15) to [bend left=70] (16); 
        
        \draw[ef]  (16) to [bend left=70] (17);

		\end{tikzpicture}
    }

    \resizebox{.45\textwidth}{!}{
    \begin{tikzpicture}[scale=1.2,mylabel/.style={thin,  align=center,fill=white,bend right}]
		
		\node[vc] (1) at (7*360/18: 4cm) {\tiny AKT};
		\node[vc] (2) at (6*360/18: 4cm) {\tiny AKTp.308};
		\node[vc] (3) at (5*360/18: 4cm) {\tiny AKTp.473};
		\node[vc] (4) at (4*360/18: 4cm) {\tiny BAD};
		\node[vc] (5) at (3*360/18: 4cm) {\tiny BAD.p112};
		\node[vc] (6) at (2*360/18: 4cm) {\tiny BAD.p136};
		\node[vc] (7) at (1*360/18: 4cm) {\tiny BAD.p155};
		\node[vc] (8) at (18*360/18: 4cm) {\tiny BAX};
		\node[vc] (9) at (17*360/18: 4cm) {\tiny BCI.2};
		\node[vc] (10) at (16*360/18: 4cm) {\tiny BCI.XL};
		\node[vc] (11) at (15*360/18: 4cm) {\tiny CCNDI};
		\node[vc] (12) at (14*360/18: 4cm) {\tiny GSK3};
		\node[vc] (13) at (13*360/18: 4cm) {\tiny GSK3.p};
		\node[vc] (14) at (12*360/18: 4cm) {\tiny MYC};
		\node[vc] (15) at (11*360/18: 4cm) {\tiny PTEN};
		\node[vc] (16) at (10*360/18: 4cm) {\tiny PTEN.p};
		\node[vc] (17) at (9*360/18: 4cm) {\tiny TP53};
		\node[vc] (18) at (8*360/18: 4cm) {\tiny XIAP};

        \node at (0,-5) {$U(2)$};
		
		\draw[ef]  (1) to [bend left=90] (3);
		\draw[ef]  (1) to (15);

		\draw[ef]  (2) to [bend left=70] (3);
		\draw[ef]  (2) to (9);
        \draw[ef]  (2) to (15);
        \draw[ef]  (2) to [bend right=70] (16);
        
		\draw[ef]  (3) to (9);
        \draw[ef]  (3) to (17);

        \draw[ef]  (5) to (15);

       \draw[ef]  (6) to (16);
        
		\draw[ef]  (7) to (15);

        \draw[ef]  (9) to (17);

        \draw[ef]  (10) to (15);

        \draw[ef]  (11) to (17);

        \draw[ef]  (14) to [bend left=70] (17);
        
		\draw[ef]  (15) to [bend left=70] (16); 
        \draw[ef]  (15) to [bend left=90] (17);

        \draw[ef]  (16) to [bend left=70] (17); 
        
        \end{tikzpicture}
    }
    \resizebox{.45\textwidth}{!}{
    \begin{tikzpicture}[scale=1.2,mylabel/.style={thin,  align=center,fill=white,bend right}]
		
		\node[vc] (1) at (7*360/18: 4cm) {\tiny AKT};
		\node[vc] (2) at (6*360/18: 4cm) {\tiny AKTp.308};
		\node[vc] (3) at (5*360/18: 4cm) {\tiny AKTp.473};
		\node[vc] (4) at (4*360/18: 4cm) {\tiny BAD};
		\node[vc] (5) at (3*360/18: 4cm) {\tiny BAD.p112};
		\node[vc] (6) at (2*360/18: 4cm) {\tiny BAD.p136};
		\node[vc] (7) at (1*360/18: 4cm) {\tiny BAD.p155};
		\node[vc] (8) at (18*360/18: 4cm) {\tiny BAX};
		\node[vc] (9) at (17*360/18: 4cm) {\tiny BCI.2};
		\node[vc] (10) at (16*360/18: 4cm) {\tiny BCI.XL};
		\node[vc] (11) at (15*360/18: 4cm) {\tiny CCNDI};
		\node[vc] (12) at (14*360/18: 4cm) {\tiny GSK3};
		\node[vc] (13) at (13*360/18: 4cm) {\tiny GSK3.p};
		\node[vc] (14) at (12*360/18: 4cm) {\tiny MYC};
		\node[vc] (15) at (11*360/18: 4cm) {\tiny PTEN};
		\node[vc] (16) at (10*360/18: 4cm) {\tiny PTEN.p};
		\node[vc] (17) at (9*360/18: 4cm) {\tiny TP53};
		\node[vc] (18) at (8*360/18: 4cm) {\tiny XIAP};

        \node at (0,-5) {$U(3)$};
		
		\draw[ef]  (2) to [bend left=70] (3);
		\draw[ef]  (2) to (9);
        \draw[ef]  (2) to (14);
        \draw[ef]  (2) to [bend right=70] (16);
        \draw[ef]  (2) to [bend right=90] (17);
		
		\draw[ef]  (3) to (13);
        \draw[ef]  (3) to (15);
		\draw[ef]  (3) to (17);

        \draw[ef]  (4) to (17);

        \draw[ef]  (6) to [bend left=70] (9);
		\draw[ef]  (6) to (16);
		
		\draw[ef]  (7) to (15);

        \draw[ef]  (10) to (15);

        \draw[ef] (15) to [bend left=70] (16); 
        \draw[ef]  (15) to [bend left=90] (17);

        \draw[ef]  (16) to [bend left=70] (17);

		\end{tikzpicture}
    }
    \end{center}
    \caption{Induced undirected multiple-graph for the profile outcome vector $Y_V(x)$} 
    \label{fig:MultiPUGM}
\end{figure}

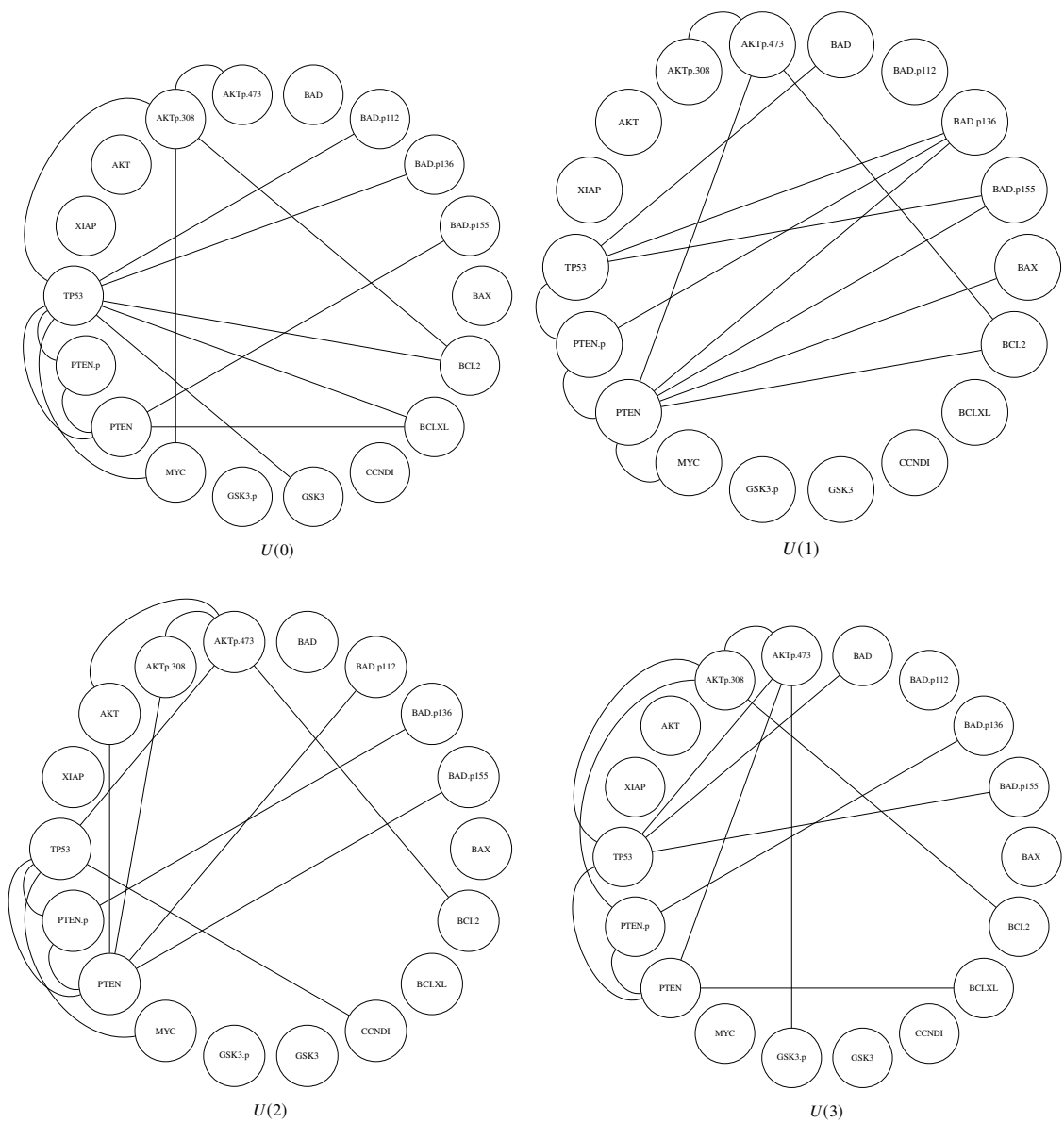
\begin{figure}[t]
	\begin{center}
    \resizebox{.45\textwidth}{!}{
    \begin{tikzpicture}[scale=1.2,mylabel/.style={thin,  align=center,fill=white,bend right}]
		
		\node[vc] (1) at (7*360/18: 4cm) {\tiny AKT};
		\node[vc] (2) at (6*360/18: 4cm) {\tiny AKTp.308};
		\node[vc] (3) at (5*360/18: 4cm) {\tiny AKTp.473};
		\node[vc] (4) at (4*360/18: 4cm) {\tiny BAD};
		\node[vc] (5) at (3*360/18: 4cm) {\tiny BAD.p112};
		\node[vc] (6) at (2*360/18: 4cm) {\tiny BAD.p136};
		\node[vc] (7) at (1*360/18: 4cm) {\tiny BAD.p155};
		\node[vc] (8) at (18*360/18: 4cm) {\tiny BAX};
		\node[vc] (9) at (17*360/18: 4cm) {\tiny BCI.2};
		\node[vc] (10) at (16*360/18: 4cm) {\tiny BCI.XL};
		\node[vc] (11) at (15*360/18: 4cm) {\tiny CCNDI};
		\node[vc] (12) at (14*360/18: 4cm) {\tiny GSK3};
		\node[vc] (13) at (13*360/18: 4cm) {\tiny GSK3.p};
		\node[vc] (14) at (12*360/18: 4cm) {\tiny MYC};
		\node[vc] (15) at (11*360/18: 4cm) {\tiny PTEN};
		\node[vc] (16) at (10*360/18: 4cm) {\tiny PTEN.p};
		\node[vc] (17) at (9*360/18: 4cm) {\tiny TP53};
		\node[vc] (18) at (8*360/18: 4cm) {\tiny XIAP};

        \node at (0,-5) {$U(0)$};
		
        \draw[ef]  (2) to [bend left=70] (3);
		\draw[ef]  (2) to (9);
        \draw[ef]  (2) to (14);
        \draw[ef]  (2) to [bend right=90] (17);
		
		\draw[ef]  (5) to (17);

        \draw[ef]  (6) to (17);
		
		\draw[ef]  (7) to (15);

        \draw[ef]  (9) to (17);

        \draw[ef]  (10) to (15);
        \draw[ef]  (10) to (17);

        \draw[ef]  (12) to (17);

        \draw[ef]  (14) to [bend left=70] (17);
        
		\draw[ef]  (15) to [bend left=70] (16); 
        \draw[ef]  (15) to [bend left=90] (17);

        \draw[ef]  (16) to [bend left=70] (17);     
        \end{tikzpicture}
    }
    \resizebox{.45\textwidth}{!}{
    \begin{tikzpicture}[scale=1.2,mylabel/.style={thin,  align=center,fill=white,bend right}]
		
		\node[vc] (1) at (7*360/18: 4cm) {\tiny AKT};
		\node[vc] (2) at (6*360/18: 4cm) {\tiny AKTp.308};
		\node[vc] (3) at (5*360/18: 4cm) {\tiny AKTp.473};
		\node[vc] (4) at (4*360/18: 4cm) {\tiny BAD};
		\node[vc] (5) at (3*360/18: 4cm) {\tiny BAD.p112};
		\node[vc] (6) at (2*360/18: 4cm) {\tiny BAD.p136};
		\node[vc] (7) at (1*360/18: 4cm) {\tiny BAD.p155};
		\node[vc] (8) at (18*360/18: 4cm) {\tiny BAX};
		\node[vc] (9) at (17*360/18: 4cm) {\tiny BCI.2};
		\node[vc] (10) at (16*360/18: 4cm) {\tiny BCI.XL};
		\node[vc] (11) at (15*360/18: 4cm) {\tiny CCNDI};
		\node[vc] (12) at (14*360/18: 4cm) {\tiny GSK3};
		\node[vc] (13) at (13*360/18: 4cm) {\tiny GSK3.p};
		\node[vc] (14) at (12*360/18: 4cm) {\tiny MYC};
		\node[vc] (15) at (11*360/18: 4cm) {\tiny PTEN};
		\node[vc] (16) at (10*360/18: 4cm) {\tiny PTEN.p};
		\node[vc] (17) at (9*360/18: 4cm) {\tiny TP53};
		\node[vc] (18) at (8*360/18: 4cm) {\tiny XIAP};

        \node at (0,-5) {$U(1)$};
		
		\draw[ef]  (2) to [bend left=70] (3);
		
        \draw[ef]  (3) to (9);
        \draw[ef]  (3) to (15);

        \draw[ef]  (4) to (17);

        \draw[ef]  (6) to (15);
        \draw[ef]  (6) to (16);
        \draw[ef]  (6) to (17);
		
		\draw[ef]  (7) to (15);
        \draw[ef]  (7) to (17);

        \draw[ef]  (8) to (15);
        
        \draw[ef]  (9) to (15);

        \draw[ef]  (14) to [bend left=70] (15);
        
		\draw[ef]  (15) to [bend left=70] (16); 
        
        \draw[ef]  (16) to [bend left=70] (17);     
		
		\end{tikzpicture}
    }

    \resizebox{.45\textwidth}{!}{
    \begin{tikzpicture}[scale=1.2,mylabel/.style={thin,  align=center,fill=white,bend right}]
		
		\node[vc] (1) at (7*360/18: 4cm) {\tiny AKT};
		\node[vc] (2) at (6*360/18: 4cm) {\tiny AKTp.308};
		\node[vc] (3) at (5*360/18: 4cm) {\tiny AKTp.473};
		\node[vc] (4) at (4*360/18: 4cm) {\tiny BAD};
		\node[vc] (5) at (3*360/18: 4cm) {\tiny BAD.p112};
		\node[vc] (6) at (2*360/18: 4cm) {\tiny BAD.p136};
		\node[vc] (7) at (1*360/18: 4cm) {\tiny BAD.p155};
		\node[vc] (8) at (18*360/18: 4cm) {\tiny BAX};
		\node[vc] (9) at (17*360/18: 4cm) {\tiny BCI.2};
		\node[vc] (10) at (16*360/18: 4cm) {\tiny BCI.XL};
		\node[vc] (11) at (15*360/18: 4cm) {\tiny CCNDI};
		\node[vc] (12) at (14*360/18: 4cm) {\tiny GSK3};
		\node[vc] (13) at (13*360/18: 4cm) {\tiny GSK3.p};
		\node[vc] (14) at (12*360/18: 4cm) {\tiny MYC};
		\node[vc] (15) at (11*360/18: 4cm) {\tiny PTEN};
		\node[vc] (16) at (10*360/18: 4cm) {\tiny PTEN.p};
		\node[vc] (17) at (9*360/18: 4cm) {\tiny TP53};
		\node[vc] (18) at (8*360/18: 4cm) {\tiny XIAP};

        \node at (0,-5) {$U(2)$};
		
		\draw[ef]  (1) to [bend left=90] (3);
		\draw[ef]  (1) to (15);

		\draw[ef]  (2) to [bend left=70] (3);
		\draw[ef]  (2) to (15);
        
		\draw[ef]  (3) to (9);
        \draw[ef]  (3) to (17);

        \draw[ef]  (5) to (15);

        \draw[ef]  (6) to (16);
		
		\draw[ef]  (7) to (15);

        \draw[ef]  (11) to (17);

        \draw[ef]  (14) to [bend left=70] (17);
        
		\draw[ef]  (15) to [bend left=70] (16); 
        \draw[ef]  (15) to [bend left=90] (17);

        \draw[ef]  (16) to [bend left=70] (17);     
        \end{tikzpicture}
    }
    \resizebox{.45\textwidth}{!}{
    \begin{tikzpicture}[scale=1.2,mylabel/.style={thin,  align=center,fill=white,bend right}]
		
		\node[vc] (1) at (7*360/18: 4cm) {\tiny AKT};
		\node[vc] (2) at (6*360/18: 4cm) {\tiny AKTp.308};
		\node[vc] (3) at (5*360/18: 4cm) {\tiny AKTp.473};
		\node[vc] (4) at (4*360/18: 4cm) {\tiny BAD};
		\node[vc] (5) at (3*360/18: 4cm) {\tiny BAD.p112};
		\node[vc] (6) at (2*360/18: 4cm) {\tiny BAD.p136};
		\node[vc] (7) at (1*360/18: 4cm) {\tiny BAD.p155};
		\node[vc] (8) at (18*360/18: 4cm) {\tiny BAX};
		\node[vc] (9) at (17*360/18: 4cm) {\tiny BCI.2};
		\node[vc] (10) at (16*360/18: 4cm) {\tiny BCI.XL};
		\node[vc] (11) at (15*360/18: 4cm) {\tiny CCNDI};
		\node[vc] (12) at (14*360/18: 4cm) {\tiny GSK3};
		\node[vc] (13) at (13*360/18: 4cm) {\tiny GSK3.p};
		\node[vc] (14) at (12*360/18: 4cm) {\tiny MYC};
		\node[vc] (15) at (11*360/18: 4cm) {\tiny PTEN};
		\node[vc] (16) at (10*360/18: 4cm) {\tiny PTEN.p};
		\node[vc] (17) at (9*360/18: 4cm) {\tiny TP53};
		\node[vc] (18) at (8*360/18: 4cm) {\tiny XIAP};

        \node at (0,-5) {$U(3)$};
		
		\draw[ef]  (2) to [bend left=70] (3);
		\draw[ef]  (2) to (9);
        \draw[ef]  (2) to [bend right=70] (16);
        \draw[ef]  (2) to [bend right=90] (17);
		
		\draw[ef]  (3) to (13);
        \draw[ef]  (3) to (15);
		\draw[ef]  (3) to (17);

        \draw[ef]  (4) to (17);

        \draw[ef]  (6) to (16);
		
		\draw[ef]  (7) to (17);

        \draw[ef]  (10) to (15);

        \draw[ef]  (15) to [bend left=70] (16); 
        \draw[ef]  (15) to [bend left=90] (17);

		\end{tikzpicture}
    }
    \end{center}
    \caption{Induced multiple-graph from GemBag}
    \label{fig:MultiGemBag}
\end{figure}

\section{Discussion}

We propose a class of graphical models that generalizes both chain graphs and multiple graphs and, for the first time, we establish compatibility between these two types of graph. 
In line with LWF chain graphs, profile undirected graphs can be used for modelling the profile conditional independencies 
resulting from a sequence of non-independent regression models involving all response variables. 
From this perspective, the specification of a class of profile chain graphs represents an interesting generalization to explore profile independencies in a multivariate regression setting.  
The parameterization discussed in Section \ref{sec:Gaussian} for the Gaussian case is quite standard.
Under the assumption of a Multinomial sampling scheme for the multivariate outcome vector, a parameterization based on the log-linear transformation \citep{lauritzen1996graphical} could be used for profile undirected graph models. We developed a Bayesian inferential procedure and a companion EM algorithm for fast inference. Alternative inferential approaches for this type of graph are possible; for example, model comparison within the class of profile undirected graphs can be based on the likelihood ratio test in the case of nested models. These graphical models are smooth and belong to the curved exponential family, so the likelihood ratio test has an asymptotic chi-square distribution. 

We demonstrate the practical utility of this class of models through the analysis of protein network data from multiple subtypes of acute myeloid leukemia. 
In this application, our proposed approach yields more robust network estimates than GemBag, as evidenced by higher balanced accuracy when the experiment is repeated 100 times with 10\% of the data randomly removed at each iteration. 
Although this represents limited empirical evidence, the observed gains in robustness suggest that our method tends to include fewer false-positive edges.
This behavior is consistent with the explicit modeling of external factors acting on the nodes, which allows the method to disentangle genuine conditional dependencies from associations induced by unmodeled heterogeneity. 
As a result, BPUGM learns sparser and more stable graph structures by excluding edges that are not truly supported by the data, whereas GemBag appears to retain weaker residual associations that are less stable across subsamples. 
Overall, these findings indicate that explicitly accounting for external factors can lead to more reliable inference of biologically meaningful network connections.


\section*{Acknowledgement}
The authors gratefully acknowledge Andrea Lazzerini for his contributions.

\section*{Supplementary material}
\label{SM}
The Supplementary Material includes the profile local Markov property, proofs of the theorems and propositions, details of the EM algorithm, the data-generating mechanism used in the simulation studies, and additional simulation results



\bibliographystyle{abbrvnat}
\newpage
\bibliography{paper-ref}

\clearpage
\pagebreak
\begin{center}
\textbf{\large Supplemental Materials: Profile Graphical Models}
\end{center}

\setcounter{equation}{0}
\setcounter{figure}{0}
\setcounter{table}{0}
\setcounter{page}{1}
\makeatletter
\renewcommand{\theequation}{S\arabic{equation}}
\renewcommand{\thefigure}{S\arabic{figure}}
\renewcommand{\bibnumfmt}[1]{[S#1]}
\renewcommand{\citenumfont}[1]{S#1}
\vspace{-1cm}
\appendix


\section{Profile Local Markov property}

The probability distributions $P[Y_{V|\X}]$ of the profile outcome vectors $Y_{V|\X}$ satisfy the \textit{profile undirected Local Markov Property} ($\U$-LMP) wrt the graph $\GU=(V,\EU)$ if, for any vertex $a \in V$

\begin{equation}
\label{localmpU}
Y_{a}(x) \ci Y_{V\setminus\{a ~ \cup ~ nb_x(a)\}}(x)|Y_{nb_x(a)}(x), \qquad x \in \X.
\end{equation}

\section{Proofs}
\begin{proof} \textbf{\textit{of Theorem \ref{t1}.}} 
	Let $\GU=(V,\EU)$ be a profile undirected graph and $D \subseteq V$ be any $x$-disconnected set with $x$-connected components $K_1,\dots,K_r$ such that for every pair $K_i,K_j$ with $i,j=1,\dots,r$, $i \neq j$, for the $\U$-CSMP wrt $\GU$ we have
	\begin{equation}
	\label{cs-g}
	Y_{K_i}(x) \ci Y_{K_j}(x)|Y_{V\setminus \{K_i,K_j\}}(x), \qquad x \in \X.
	\end{equation}
	For any pair $K_i,K_j \subset D$ with $i,j=1,\dots,r$, $i \neq j$, the set $S_{ij} = V\setminus\{ K_i,K_j\}$ is an $x$-separator. Then,  the $\U$-CSMP implies the $\U$-GMP wrt $\GU$. Conversely, consider any $x$-connected set $C$ wrt $\GU$ and let $nb_x(C)=\bigcup_{a \in C}nb_x(a)$ be the neighbour set including $C$ and let $S_C=nb_x(C)\setminus C$   be an $x$-separator for the sets  $C$ and $V\setminus nb_x(C)$, for any $x \in \X$. The $\U$-GMP implies that
	\begin{equation}
	\label{g-cs}
	Y_{C}(x) \ci Y_{V\setminus nb_x(C)}(x)|Y_{S_C}(x), \qquad x \in \X.
	\end{equation}
	Note that $C \cup \{V \setminus nb_x(C)\}$ is an $x$-disconnected set, for $x \in \X$. We distinguish two  cases, whether $V \setminus nb_x(C)$ is $x$-connected or $x$-disconnected. In the first case, the $x$-connected components of $C \cup \{V \setminus nb_x(C)\}$ are $C$ and $V \setminus nb_x(C)$, then the $\U$-CSMP is satisfied. If $V \setminus nb_x(C)$ is $x$-disconnected with $K_1 \cup \dots \cup K_r$ connected components,  the $\U$-GMP also implies that 
	\begin{equation}
	\label{g-cs2}
	Y_{C}(x) \ci Y_{K_1}(x) \ci \dotso \ci Y_{K_r}(x)| Y_{S_C} (x), \qquad x \in \X.
	\end{equation}
	Then  the $\U$-GMP implies the $\U$-CSMP wrt $\GU$.
\end{proof}

\begin{proof} \textbf{\textit{of Proposition \ref{p2}.}}
	Consider a profile undirected graph $\GU=(V,\EU)$ associated to the profile outcome vectors $Y_{V|\X}$ and the induced class $U_{V|\X}$ of multiple undirected graphs. If the probability distributions $P[Y_{V|\X}]$ satisfy the $\U$-CSMP wrt $\GU$, the $\U$-GMP is also satisfied from Theorem \ref{t1}. So, given three disjoint subsets $A,B,C$  of $V$,
	\begin{equation}
	Y_A(x) \bigCI  Y_{B}(x) | Y_{C}(x),
	\end{equation}
	where $A$ and $B$ are $x$-separated by $C$, with $x 	\in \X$. 
	The result follows by Definition \ref{def.multUnd}, since $A$ and $B$ are $x$-separated by $C$ in $\GB$ if and only if they are $x$-separated by $C$ in $U(x) \in U_{V|\X}$, with $x \in \X$.
\end{proof}

\begin{proof} \textbf{\textit{of Proposition \ref{p1}.}} 
	Given an undirected graph $U=(V,E_U)$ associated to a random vector $Y_V$, the global and the pairwise Markov properties are equivalent if the joint probability distribution $P(Y_V)$ is strictly positive; see \citet{lauritzen1996graphical}. The proposition follows by applying this result to the strictly positive probability distribution $P[Y_V(x)]$ of any profile outcome vector $Y_V(x) \in Y_{V|\X}$.
\end{proof}

\begin{proof} \textbf{\textit{of Theorem \ref{thr2}.}}
     { First of all we recall that, given a profile undirected graph $\GU$ associated to the profile outcome vectors $Y_{V|\X}$,  for  Theorem \eqref{t1} the probability distributions $P[Y_{V|\X}]$ satisfies the $\U$-CSMP if and only if the $\U$-GMP is  satisfied. In order to prove the compatibility, we distinguish three cases:
      \begin{itemize}
       \item[a)] if a set $D$ is $x$-disconnected in $\GU$ with $x = \emptyset$, this means that it is a connected set and then the same set $D$ is also connected in any $C_U \in \mathcal{C}_U$ for condition $(i)$ and, therefore, the $\U$-CSMP trivially implies the condition \eqref{eq:LWF global undirected} as no independence statement is implied for $Y_D$ given $X$ and the remaining set of variables;  in this case condition $(ii)$ is  not invoked as no independence statements hold for 
          $Y_D $
          regardless of any $Y_j \in Y_D$ is dependent or independent of $X$ given the remaining set of variables;  
          \item[b)] if a set $D$ is $x$-disconnected in $\GU$ for all $x \in \X$, then the same set $D$ is also disconnected in any $C_U \in \mathcal{C}_U$ for condition $(i)$ and, therefore, the $\U$-CSMP implies the condition \eqref{eq:LWF global undirected} for the class of induced chain graphs; in this case condition $(ii)$ is also not invoked as 
          $$
          Y_{K_1}(x) \ci \dots \ci Y_{K_r}(x)|Y_{V\setminus D}(x), \quad \forall x \in \X \Rightarrow Y_{K_1} \ci \dots \ci Y_{K_r}|\{Y_{V\setminus D}, X\},
          $$
          regardless of any $Y_j \in Y_D$ is dependent or independent of $X$ given the remaining set of variables;
          \item[c)] if a set $D$ is $x$-disconnected in $\GU$ for some $x \in \Z \subset \X$ with $x \neq \emptyset$, then the same set $D$ is  connected in any $C_U \in \mathcal{C}_U$ for condition $(i)$; the $\U$-CSMP implies  condition \eqref{eq:LWF global undirected} as 
          $$
          Y_{K_1}(x) \ci \dots \ci Y_{K_r}(x)|Y_{V\setminus D}(x), \quad  x \in \Z \subset \X \Rightarrow Y_{K_1} \nci \dots \nci Y_{K_r}|\{Y_{V\setminus D}, X\};
          $$
          in this case  $Y_D$ cannot be independent of $X$ given $Y_{V \setminus D}$ as the profile conditional distribution $Y_D(x)|Y_{V \setminus D}(x)$ behaves differently  for any $x \in X$, i.e. $Y_{K_1}(x) \ci \dots \ci Y_{K_r}(x)|Y_{V\setminus D}(x)$ for $ x \in \Z $ and $Y_{K_1}(x) \nci \dots \nci Y_{K_r}(x)|Y_{V\setminus D}(x)$ for $  x \in \X \setminus \Z $;  condition $(ii)$ is then required so that $Y_D \nci X|Y_{V \setminus D}$.
          \end{itemize}
          }
\end{proof}

\section{EM algorithm}
\label{App:EM}

The \textbf{E-step} takes the expectation of the complete data log posterior
\begin{equation}
\label{eqApp:Q1}
\begin{split}
Q(\widehat{\Delta}) = &\Exp_{\theta, R\mid Y, \widehat{\Delta}}  \left\{ \sum_{x=0}^{q-1} \left[ \log \left( \Prob(Y_{v_x}\mid \widehat{\Omega}_x, \widehat{\beta}_x,\widehat{\alpha}) \right) \right]
+\sum_{x=0}^{q-1} \sum_{i<j} \left[ \log\left( \Prob(\widehat{\omega}_{ij,x}\mid r^*_{ij,x} )\right)  \right] 
\right.\\
&+ \left. \sum_{x=0}^{q-1} \sum_{i=1}^p \left[ \log\left( \Prob(\widehat{\beta}_{ix} \mid \theta^*_i) \right) +
\log\left( \Prob(\widehat{\omega}_{ii,x} \mid \tau) \right) \right]\right\} 
\end{split}
\end{equation}

Let $Q=Q(\widehat{\Delta}) = \Exp_{\theta,R \mid Y,\widehat{\Delta}}[\log(\Prob(\widehat{\Delta} ,\theta,R\mid Y_{V_x})]= \langle \log(\Prob(\widehat{\Delta},\theta,R \mid Y_{V_x}) \rangle$ 
\begin{equation}
\begin{split}
\label{eq:Q1}
Q\propto & \left\langle  \frac{1}{2} \sum_{x=0}^{q-1} \left[n_x \log (\det(\widehat{\Omega}_x)) - \text{tr} \left\{ \sum_{k=1}^{n_x} \left( Y_{x,k}-(\widehat{\alpha} + \widehat{\beta}_x) \right)\left( Y_{x,k}-(\widehat{\alpha} + \widehat{\beta}_x) \right)^\top \widehat{\Omega}_x \right\} \right] \right. \\
& \left. +  \sum_{x=0}^{q-1} \sum_{i=1}^p \left[ \log\left( \Prob(\widehat{\beta}_{ix} \mid \theta^*_i) \right) -
 \tau \widehat{\omega}_{ii,x}  \right] + \sum_{x=0}^{q-1} \sum_{i<j} \left[ \log\left( \Prob(\widehat{\omega}_{ij,x}\mid r^*_{ij,x} )\right) \right] \right\rangle,
 \end{split}
\end{equation}

where tr$(A)$ is the trace of matrix $A$ and det$(A)$ its determinant.
Setting a Laplace spike-and-slab prior on $\omega_{ij,x}\mid r_{ij,x}$ as in~\citet{YangXinming2021GGEo}, and a Normal spike-and-slab on $\beta_{ix} \mid \theta_i$ as in~\citet{GeorgeEMVS}, \eqref{eq:Q1} becomes

\begin{equation}
\label{eqApp:Q2}
\begin{split}
Q \propto &   \frac{1}{2} \sum_{x=1}^q \left[n_x \log (\det (\widehat{\Omega}_x)) - \text{tr} \left\{ \sum_{k=1}^{n_x} \left( Y_{x,k}-(\widehat{\alpha} + \widehat{\beta}_x) \right)\left( Y_{x,k}-(\widehat{\alpha} + \widehat{\beta}_x) \right)^\top \widehat{\Omega}_x \right\} \right] \\
& - \sum_{x=1}^q \sum_{i=1}^p \left[ \frac{1}{2} \widehat{\beta}_{ix}^2 \left(\Exp_{\theta\mid \widehat{\Delta}, Y_{V_x}} \left[ \frac{1}{\theta_i \lambda_1}+ \frac{1}{(1-\theta_i) \lambda_0} \right] \right) + \tau \widehat{\omega}_{ii,x}  \right]  \\
&-\sum_{x=1}^q \sum_{i<j} \left[ |\widehat{\omega}_{ij,x}|  \left(\Exp_{R\mid \widehat{\Delta},Y_{V_x}} \left[ \frac{r_{ij,x}}{\nu_1}+ \frac{1-r_{ij,x}}{\nu_0} \right] \right) \right]
\end{split}
\end{equation}

Equation~\eqref{eqApp:Q2} only depends on $(\theta,R)$ through
\begin{equation}
\label{eq:Etheta}
\Exp_{\theta\mid \widehat{\Delta}, Y_{V_x}} \left[ \frac{1}{\theta_i \lambda_1}+ \frac{1}{(1-\theta_i) \lambda_0} \right] = 
\frac{\Exp_{\theta\mid \widehat{\Delta}, Y_{V_x}} [\theta_i]}{\lambda_1}+ \frac{\Exp_{\theta\mid \widehat{\Delta}, Y_{V_x}} [1-\theta_i]}{\lambda_0}=\frac{\theta_i^*}{\lambda_1}+\frac{1-\theta_i^*}{\lambda_0}
\end{equation}
and
 
\begin{equation}
\label{eq:Er}
\Exp_{R\mid \widehat{\Delta},Y_{V_x}} \left[ \frac{r_{ij,x}}{\nu_1}+ \frac{1-r_{ij,x}}{\nu_0} \right] = \frac{ \Exp_{R\mid \widehat{\Delta},Y_{V_x}}[r_{ij,x}]}{\nu_1}+ \frac{ \Exp_{R\mid \widehat{\Delta},Y_{V_x}}[1-r_{ij,x}]}{\nu_0} = \frac{r_{ij,x}^*}{\nu_1}+\frac{1-r_{ij,x}^*}{\nu_0},
\end{equation}

The model hierarchy separates $\theta$ from the data $Y_{V_x}$ through the coefficients $\beta$ and the precision $\Omega$ so that $\Prob[\theta\mid \widehat{\Delta}, Y_{V_x}] = \Prob[\theta\mid \widehat{\Delta}]= \Prob[\theta\mid \beta, {\Omega}]$.
This leads to 
\begin{equation}
\label{eq:EthetaApp}
\begin{split}
\theta_i^*  &=\Exp_{\theta\mid \widehat{\Delta}, Y_{V_x}}[\theta_i]=\Exp_{\theta\mid {\Delta}}[\theta_i]=\Prob[\theta_i=1\mid \beta_i, \omega_{i \cdot}]\\
&= \frac{\Prob[\beta_i, \omega_{i \cdot} \mid \theta_i =1] \Prob[ \theta_i =1]}{\Prob[\beta_i, \omega_{i \cdot} \mid \theta_i =1] \Prob[ \theta_i =1]+\Prob[\beta_i, \omega_{i \cdot} \mid \theta_i =0] \Prob[ \theta_i =0]}\\
&= \frac{\Prob[\beta_i \mid \theta_i =1]\Prob[\omega_{i \cdot} \mid \theta_i =1] p_2}{\Prob[\beta_i \mid \theta_i =1]\Prob[\omega_{i \cdot} \mid \theta_i =1] p_2 + \Prob[\beta_i \mid \theta_i =0]\Prob[\omega_{i \cdot} \mid \theta_i =0] (1-p_2)}\\
&= \frac{p_2 \Prob[\omega_{i \cdot} \mid \theta_i =1] \prod_{x=0}^{q-1}  \Prob_1(\beta_{ix})}{p_2 \Prob[\omega_{i \cdot} \mid \theta_i =1] \prod_{x=0}^{q-1}\Prob_1({\beta}_{ix})+(1-p_2)\Prob[\omega_{i \cdot} \mid \theta_i =0] \prod_{x=0}^{q-1}\Prob_0({\beta}_{ix})}
\end{split}
\end{equation}

The $\Prob[\omega_{i \cdot} \mid \theta_i]$ in \eqref{eq:EthetaApp} is

\begin{equation}
\label{eq:GammaTheta}
\begin{split}
\Prob(\omega_{i \cdot} \mid \theta_i ) =&  \sum_{j} \sum_{\theta_j}  \sum_{r_{ij,0},\dots,r_{ij,q-1}} \sum_{\gamma_{ij}} \left[ \left(\prod_{x=0}^{q-1}\Prob(\omega_{ij,x}\mid r_{ij,x}) \right)  \Prob(r_{ij,0},\dots,r_{ij,x}\mid \theta_i, \theta_j, \gamma_{ij})* \right. \\
& \Prob (\theta_j) \Prob (\gamma_{ij}) \Bigg]  \\
=&\sum_{j} \sum_{\theta_j} \sum_{r_{ij,0},\dots,r_{ij,q-1}} \sum_{\gamma_{ij}} \left[ \left(\prod_{x=0}^{q-1} r_{ij,x} \Prob_1(\omega_{ij,x}) +(1- r_{ij,x}) \Prob_0(\omega_{ij,x}) \right)* \right. \\
&\left( \gamma_{ij} \theta_i \theta_j \left(\prod_{x=0}^{q-1} p_3^{r_{ij,x}} (1-p_3)^{1-r_{ij,x}} \right) + (1-\gamma_{ij}) \left(\prod_{x=0}^{q-1} \delta_0(r_{ij,x})\right)  \right. \\
&+ \gamma_{ij} (1-\theta_i \theta_j) p_4^{r_{ij,0}} (1-p_4)^{1-r_{ij,0}} \delta_{r_{ij,1},\dots, r_{ij,q-1}}(r_{ij,0}) \Bigg)*\\
&\left(p_2^{\theta_j}(1-p_2)^{1-\theta_j}\right)*\left(p_1^{\gamma_{ij}}(1-p_1)^{1-\gamma_{ij}}\right)
 \Bigg]\\
 =&\sum_{j} \sum_{\theta_j} \sum_{r_{ij,0},\dots,r_{ij,q-1}} \left[ \left(\prod_{x=0}^{q-1} r_{ij,x} \Prob_1(\omega_{ij,x}) +(1- r_{ij,x}) \Prob_0(\omega_{ij,x}) \right)* \right.  \\
&\left(p_2^{\theta_j}(1-p_2)^{1-\theta_j}\right)*
\left\{ p_1 \theta_i \theta_j\left(\prod_{x=0}^{q-1} p_3^{r_{ij,x}} (1-p_3)^{1-r_{ij,x}} \right) +  \right. \\
& p_1(1-\theta_i \theta_j) p_4^{r_{ij,0}} (1-p_4)^{1-r_{ij,0}} \delta_{r_{ij,1},\dots,r_{ij,q-1}}(r_{ij,0}) + \\
& \left. (1-p_1)  \left(\prod_{x=0}^{q-1} \delta_0(r_{ij,x})\right)  \right\}
 \end{split}
\end{equation}

Then,
\begin{equation}
\label{eq:GammaTheta1}
\begin{split}
\Prob(\omega_{i \cdot} \mid \theta_i =1 ) =&\sum_{j} \sum_{\theta_j} \sum_{r_{ij,0},\dots,r_{ij,q-1}} \left[ \left(\prod_{x=0}^{q-1} r_{ij,x} \Prob_1(\omega_{ij,x}) +(1- r_{ij,x}) \Prob_0(\omega_{ij,x}) \right)* \right.  \\
&\left(p_2^{\theta_j}(1-p_2)^{1-\theta_j}\right)*
\left\{ p_1 \theta_j\left(\prod_{x=0}^{q-1} p_3^{r_{ij,x}} (1-p_3)^{1-r_{ij,x}} \right) +  \right. \\
& p_1(1- \theta_j) p_4^{r_{ij,0}} (1-p_4)^{1-r_{ij,0}} \delta_{r_{ij,1},\dots,r_{ij,q-1}}(r_{ij,0}) + \\
& \left. (1-p_1)  \left(\prod_{x=0}^{q-1} \delta_0(r_{ij,x})\right)  \right\} \\
=&\sum_{j} \sum_{r_{ij,0},\dots,r_{ij,q-1}} \left[ \left(\prod_{x=0}^{q-1} r_{ij,x} \Prob_1(\omega_{ij,x}) +(1- r_{ij,x}) \Prob_0(\omega_{ij,x}) \right)* \right. \\
&\Bigg\{ p_1 p_2 \left(\prod_{x=0}^{q-1} p_3^{r_{ij,x}} (1-p_3)^{1-r_{ij,x}} \right) + (1-p_1)p_2\left(\prod_{x=0}^{q-1} \delta_0(r_{ij,x})\right)  \\
&+ p_1(1- p_2) p_4^{r_{ij,0}} (1-p_4)^{1-r_{ij,0}} \delta_{r_{ij,1},\dots,r_{ij,q-1}}(r_{ij,0}) \Bigg\}\\
=&\sum_{j} \left[ \prod_{x=0}^{q-1} \Prob_0(\omega_{ij,x}) \Big\{ p_1 p_2 (1-p_3)^q + p_1(1- p_2) (1-p_4) + (1-p_1) p_2\Big\} \right.
\\
&+ \prod_{x=0}^{q-1} \Prob_1(\omega_{ij,x})\left\{p_1 p_2  p_3^q + p_1(1- p_2) p_4 \right\}\\
&+ p_1p_2\sum_{k=1}^{q-1} p_3^k(1-p_3)^{q-k} \sum_{n \in A_q^{(k)}} \prod_{l=0}^{q-1} \Prob_{n_l}(\omega_{ij,l}) \Bigg],
 \end{split}
\end{equation}
with $ A_q^{(k)}:=\{ n =(n_0,\dots,n_{q-1}: n_l \in \{ 0,1\}$ for all $0\leq l \leq q-1$ and $\sum_{l=0}^{q-1} = k ) \}$, the set of $\{0,1\}$-valued binary sequences of length q, with $k$ elements with $\Prob_1(\omega_{ij,k}), k=1,\dots,q-1$. 
We note that the cardinality of $ A_q^{(k)}$ is $\#A_q^{(k)} = \binom{q}{k}$ and $\sum_{k=0}^q \#A_q^{(k)} = \sum_{k=0}^q \binom{q}{k} = 2^q$.

\begin{equation}
\label{eq:GammaTheta0}
\begin{split}
\Prob(\omega_{i \cdot} \mid \theta_i =0 ) =& \sum_{j} \sum_{\theta_j} \sum_{r_{ij,0},\dots,r_{ij,q-1}} \left[ \left(\prod_{x=0}^{q-1} r_{ij,x} \Prob_1(\omega_{ij,x}) +(1- r_{ij,x}) \Prob_0(\omega_{ij,x}) \right)* \right.  \\
&\left(p_2^{\theta_j}(1-p_2)^{1-\theta_j}\right)*
\left\{  (1-p_1)  \left(\prod_{x=0}^{q-1} \delta_0(r_{ij,x})\right)  +  \right. \\
& p_1 p_4^{r_{ij,0}} (1-p_4)^{1-r_{ij,0}} \delta_{r_{ij,1},\dots,r_{ij,q-1}}(r_{ij,0}) \Bigg\} \Bigg] \\
 =& \sum_{j} \sum_{r_{ij,0},\dots,r_{ij,q-1}} \left[ \left(\prod_{x=0}^{q-1} r_{ij,x} \Prob_1(\omega_{ij,x}) +(1- r_{ij,x}) \Prob_0(\omega_{ij,x}) \right)* \right.  \\
& \left\{  (1-p_1)  \left(\prod_{x=0}^{q-1} \delta_0(r_{ij,x})\right)  +  \right. \\
& p_1 p_4^{r_{ij,0}} (1-p_4)^{1-r_{ij,0}} \delta_{r_{ij,1},\dots,r_{ij,q-1}}(r_{ij,0}) \Bigg\} \Bigg] \\
=&\sum_{j} \Bigg[  \prod_{x=0}^{q-1} \Prob_0(\omega_{ij,x}) \left( p_1(1-p_4)+(1-p_1) \right) +\\
&\prod_{x=0}^{q-1} \Prob_1(\omega_{ij,x}) p_1p_4 \Bigg]
 \end{split}
\end{equation}

\newpage
For instance, assuming $q=4$, Equation~\eqref{eq:GammaTheta1} and \eqref{eq:GammaTheta0} are of the form

\begin{equation}
\label{eq:GammaThetaQ4}
\begin{split}
\Prob(\omega_{i \cdot} \mid \theta_i =1 ) =& \sum_{j} \Bigg[ \Prob_0(\omega_{ij,0})\Prob_0(\omega_{ij,1})\Prob_0(\omega_{ij,2})\Prob_0(\omega_{ij,3}) \Big\{ p_1 p_2 (1-p_3)^4 \\
& + p_1(1- p_2) (1-p_4) + (1-p_1) p_2\Big\} \\
&+ \Prob_1(\omega_{ij,0})\Prob_1(\omega_{ij,1})\Prob_1(\omega_{ij,2})\Prob_1(\omega_{ij,3})\Big\{p_1 p_2  p_3^4 + p_1(1- p_2) p_4 \Big\}\\
&+ p_1p_2 \Big\{ p_3(1-p_3)^3 \big[ \Prob_1(\omega_{ij,0})\Prob_0(\omega_{ij,1})\Prob_0(\omega_{ij,2})\Prob_0(\omega_{ij,3}) \\
& + \Prob_0(\omega_{ij,0})\Prob_1(\omega_{ij,1})\Prob_0(\omega_{ij,2})\Prob_0(\omega_{ij,3}) +\Prob_0(\omega_{ij,0})\Prob_0(\omega_{ij,1})\Prob_1(\omega_{ij,2})\Prob_0(\omega_{ij,3}) \\
& + \Prob_0(\omega_{ij,0})\Prob_0(\omega_{ij,1})\Prob_0(\omega_{ij,2})\Prob_1(\omega_{ij,3}) \big] \\
& + p_3^2 (1-p_3)^2 \big[ \Prob_1(\omega_{ij,0})\Prob_1(\omega_{ij,1})\Prob_0(\omega_{ij,2})\Prob_0(\omega_{ij,3}) \\
&+ \Prob_1(\omega_{ij,0})\Prob_0(\omega_{ij,1})\Prob_1(\omega_{ij,2})\Prob_0(\omega_{ij,3}) + \Prob_1(\omega_{ij,0})\Prob_0(\omega_{ij,1})\Prob_0(\omega_{ij,2})\Prob_1(\omega_{ij,3}) \\
& + \Prob_0(\omega_{ij,0})\Prob_1(\omega_{ij,1})\Prob_1(\omega_{ij,2})\Prob_0(\omega_{ij,3}) + \Prob_0(\omega_{ij,0})\Prob_1(\omega_{ij,1})\Prob_0(\omega_{ij,2})\Prob_1(\omega_{ij,3}) \\
&+ \Prob_0(\omega_{ij,0})\Prob_0(\omega_{ij,1})\Prob_1(\omega_{ij,2})\Prob_1(\omega_{ij,3})\big] \\
&+p_3^3(1-p_3) \big[ \Prob_0(\omega_{ij,0})\Prob_1(\omega_{ij,1})\Prob_1(\omega_{ij,2})\Prob_1(\omega_{ij,3}) \\
& + \Prob_1(\omega_{ij,0})\Prob_0(\omega_{ij,1})\Prob_1(\omega_{ij,2})\Prob_1(\omega_{ij,3}) + \Prob_1(\omega_{ij,0})\Prob_1(\omega_{ij,1})\Prob_0(\omega_{ij,2})\Prob_1(\omega_{ij,3}) \\
& + \Prob_1(\omega_{ij,0})\Prob_1(\omega_{ij,1})\Prob_1(\omega_{ij,2})\Prob_0(\omega_{ij,3}) \big] \Big\}
 \Bigg],\\
 \Prob(\omega_{i \cdot} \mid \theta_i =0 ) =&  \sum_{j} \Bigg[ \Prob_0(\omega_{ij,0})\Prob_0(\omega_{ij,1})\Prob_0(\omega_{ij,2})\Prob_0(\omega_{ij,3}) \left( p_1(1-p_4)+(1-p_1) \right) +\\
&\Prob_1(\omega_{ij,0})\Prob_1(\omega_{ij,1})\Prob_1(\omega_{ij,2})\Prob_1(\omega_{ij,3}) p_1p_4 \Bigg].
 \end{split}
\end{equation}

\newpage
The conditional expectation $\Exp_{R\mid \widehat{\Delta},Y_{V_x}}[r_{ij,x}]=r_{ij,x}^*$ is

\begin{equation}
\label{eq:ErAppP0}
\begin{split}
r_{ij,x}^* =&\Exp_{R\mid \widehat{\Delta},Y_{V_x}}[r_{ij,x}]= \Exp_{R\mid \widehat{\Delta}}[r_{ij,x}] =\Prob(r_{ij,x}=1\mid \Delta) \\
=& \Prob(\gamma_{ij}=1\mid \Delta) * \\
&\Big[ \Prob(\theta_i =1 \mid \Delta)\Prob(\theta_j =1 \mid \Delta) \Prob(r_{ij,x}=1\mid \gamma_{ij}=1, \theta_i=1,\theta_j=1, \Delta) +\\
&\Prob(r_{ij,/x}=1\mid\Delta)\Big\{  \Prob(\theta_i =1 \mid \Delta)\Prob(\theta_j =0 \mid \Delta) \Prob(r_{ij,x}=1\mid \gamma_{ij}=1, \theta_i=1,\theta_j=0, \Delta) \\
&+  \Prob(\theta_i =0 \mid \Delta)\Prob(\theta_j =1\mid \Delta)\Prob(r_{ij,x}=1\mid \gamma_{ij}=1, \theta_i=0,\theta_j=1, \Delta)\\
&+\Prob(\theta_i =0 \mid \Delta)\Prob(\theta_j =0 \mid \Delta) \Prob(r_{ij,x}=1\mid \gamma_{ij}=1, \theta_i=0,\theta_j=0, \Delta) \Big\} \Big] 
\end{split}
\end{equation}

$\mathbf{1. \Prob(\gamma_{ij}=1\mid \Delta)}$ in Equation~\eqref{eq:ErAppP0} is
\begin{equation}
\label{eq:Gamma1OmegaApp}
\begin{split}
\Prob(\gamma_{ij}=1\mid\Delta) =& \frac{\Prob(\Delta \mid\gamma_{ij}=1) \Prob(\gamma_{ij}=1)}{\Prob(\Delta \mid\gamma_{ij}=1) \Prob(\gamma_{ij}=1)+\Prob(\Delta \mid\gamma_{ij}=0) \Prob(\gamma_{ij}=0)}\\
=& \frac{\Prob(\Delta\mid\gamma_{ij}=1) p_1}{\Prob(\Delta \mid\gamma_{ij}=1) p_1+\Prob(\Delta \mid\gamma_{ij}=0) (1-p_1)}.
\end{split},
\end{equation}

\begin{equation}
\begin{split}
\Prob(\Delta \mid\gamma_{ij}) =& \sum_{r_{ij,0},\dots,r_{ij,q-1}} \sum_{\theta_i}\sum_{\theta_j} \Bigg[ \Bigg( \prod_{x=0}^{q-1} \Prob(\omega_{ij,x}\mid r_{ij,x})\Bigg) \Prob(r_{ij,0},\dots,r_{ij,x}\mid \theta_i, \theta_j, \gamma_{ij})* \\
& \Bigg( \prod_{x=0}^{q-1} \Prob(\beta_{i,x} \mid \theta_i)
\Prob(\beta_{j,x} \mid \theta_j)\Bigg) \Prob (\theta_i)\Prob (\theta_j) \Bigg]  
\end{split},
\end{equation}

\begin{equation}
\label{eq:OmegaGammaApp}
\begin{split}
\Prob(\Delta \mid\gamma_{ij}) =&\sum_{r_{ij,0},\dots,r_{ij,q-1}} \sum_{\theta_i}\sum_{\theta_j} \Bigg[ \Bigg( \prod_{x=0}^{q-1} r_{ij,x} \Prob_1(\omega_{ij,x}) +(1- r_{ij,x}) \Prob_0(\omega_{ij,x}) \Bigg)*  \\
&\left( \gamma_{ij} \theta_i \theta_j \left(\prod_{x=0}^{q-1} p_3^{r_{ij,x}} (1-p_3)^{1-r_{ij,x}} \right) + (1-\gamma_{ij}) \left(\prod_{x=0}^{q-1} \delta_0(r_{ij,x})\right)  \right. \\
&+ \gamma_{ij} (1-\theta_i \theta_j) p_4^{r_{ij,0}} (1-p_4)^{1-r_{ij,0}} \delta_{r_{ij,1},\dots, r_{ij,q-1}}(r_{ij,0}) \Bigg)*\\
& \Bigg( \prod_{x=0}^{q-1} \Big( \theta_{i} \Prob_1(\beta_{ix}) + (1-\theta_{i}) \Prob_0(\beta_{ix}) \Big) \Big( \theta_{j} \Prob_1(\beta_{jx}) + (1-\theta_{j}) \Prob_0(\beta_{jx}) \Big)\Bigg) *\\
&p_2^{\theta_i}(1-p_2)^{1-\theta_i}p_2^{\theta_j}(1-p_2)^{1-\theta_j}
 \Bigg] \\
  =& \sum_{r_{ij,0},\dots,r_{ij,q-1}} \Bigg[ \Bigg( \prod_{x=0}^{q-1} r_{ij,x} \Prob_1(\omega_{ij,x}) +(1- r_{ij,x}) \Prob_0(\omega_{ij,x}) \Bigg)*\\
&\Bigg\{p_2^2 * \prod_{x=0}^{q-1} \Big(  \Prob_1(\beta_{ix}) \Prob_1(\beta_{jx}) \Big) * \Bigg(\gamma_{ij}   \left(\prod_{x=0}^{q-1} p_3^{r_{ij,x}}  (1-p_3)^{1-r_{ij,x}} \right)\\
& + (1-\gamma_{ij}) \left(\prod_{x=0}^{q-1} \delta_0(r_{ij,x})\right) \Bigg) + \Bigg( (1-p_2)^2 * \prod_{x=0}^{q-1} \Big(  \Prob_0(\beta_{ix}) \Prob_0(\beta_{jx}) \Big) \\
&+p_2(1-p_2) \Bigg( \prod_{x=0}^{q-1} \Big(  \Prob_0(\beta_{ix}) \Prob_1(\beta_{jx}) \Big) +\prod_{x=0}^{q-1} \Big(  \Prob_1(\beta_{ix}) \Prob_0(\beta_{jx}) \Big) \Bigg)\Bigg) *\\
&\Bigg((1-\gamma_{ij}) \left(\prod_{x=0}^{q-1} \delta_0(r_{ij,x})\right) + \gamma_{ij} p_4^{r_{ij,0}} (1-p_4)^{1-r_{ij,0}} \delta_{r_{ij,1},\dots, r_{ij,q-1}}(r_{ij,0}) \Bigg) \Bigg\} \Bigg] 
\end{split} 
\end{equation}
Set $\gamma_{ij}=0$ and let $g_0(\omega_{ij})=\Prob(\Delta \mid\gamma_{ij}=0)$
\begin{equation}
\label{eq:DeltaGamma0}
\begin{split}
g_0(\omega_{ij})=& \sum_{r_{ij,0},\dots,r_{ij,q-1}} \Bigg[ \Bigg( \prod_{x=0}^{q-1} r_{ij,x} \Prob_1(\omega_{ij,x}) +(1- r_{ij,x}) \Prob_0(\omega_{ij,x}) \Bigg)*\\
&\Bigg\{p_2^2 * \prod_{x=0}^{q-1} \Big(  \Prob_1(\beta_{ix}) \Prob_1(\beta_{jx}) \Big) * \prod_{x=0}^{q-1} \delta_0(r_{ij,x}) \\
& + \Bigg(\prod_{x=0}^{q-1} \delta_0(r_{ij,x}) \Bigg) * \Bigg( (1-p_2)^2 * \prod_{x=0}^{q-1} \Big(  \Prob_0(\beta_{ix}) \Prob_0(\beta_{jx}) \Big) \\
&+p_2(1-p_2) \Bigg( \prod_{x=0}^{q-1} \Big(  \Prob_0(\beta_{ix}) \Prob_1(\beta_{jx}) \Big) +\prod_{x=0}^{q-1} \Big(  \Prob_1(\beta_{ix}) \Prob_0(\beta_{jx}) \Big) \Bigg)\Bigg) \Bigg\} \Bigg] \\
=& \prod_{x=0}^{q-1} \Prob_0(\omega_{ij,x})\Bigg[p_2^2 * \prod_{x=0}^{q-1} \Big(  \Prob_1(\beta_{ix}) \Prob_1(\beta_{jx}) \Big) + \\
&  (1-p_2)^2 * \prod_{x=0}^{q-1} \Big(  \Prob_0(\beta_{ix}) \Prob_0(\beta_{jx}) \Big)  +p_2(1-p_2)*\\
& \Bigg( \prod_{x=0}^{q-1} \Big(  \Prob_0(\beta_{ix}) \Prob_1(\beta_{jx}) \Big) +\prod_{x=0}^{q-1} \Big(  \Prob_1(\beta_{ix}) \Prob_0(\beta_{jx}) \Big) \Bigg)\Bigg)\Bigg].
\end{split}
\end{equation}
Set $\gamma_{ij}=1$ and let $g_1(\omega_{ij})=\Prob(\Delta \mid\gamma_{ij}=1)$
\begin{equation}
\begin{split}
g_1(\omega_{ij})=& 
\sum_{r_{ij,0},\dots,r_{ij,q-1}} \Bigg[ \Bigg( \prod_{x=0}^{q-1} r_{ij,x} \Prob_1(\omega_{ij,x}) +(1- r_{ij,x}) \Prob_0(\omega_{ij,x}) \Bigg)*\\
&\Bigg\{p_2^2 * \prod_{x=0}^{q-1} \Big(  \Prob_1(\beta_{ix}) \Prob_1(\beta_{jx}) \Big) *  \left(\prod_{x=0}^{q-1} p_3^{r_{ij,x}}  (1-p_3)^{1-r_{ij,x}} \right)   \\
& + \Bigg( (1-p_2)^2 * \prod_{x=0}^{q-1} \Big(  \Prob_0(\beta_{ix}) \Prob_0(\beta_{jx}) \Big) \\
&+p_2(1-p_2) \Bigg( \prod_{x=0}^{q-1} \Big(  \Prob_0(\beta_{ix}) \Prob_1(\beta_{jx}) \Big) +\prod_{x=0}^{q-1} \Big(  \Prob_1(\beta_{ix}) \Prob_0(\beta_{jx}) \Big) \Bigg)\Bigg) *\\
&\Bigg(p_4^{r_{ij,0}} (1-p_4)^{1-r_{ij,0}} \delta_{r_{ij,1},\dots, r_{ij,q-1}}(r_{ij,0}) \Bigg) \Bigg\} \Bigg] 
\nonumber
\end{split} 
\end{equation}
\begin{equation}
\begin{split}
\label{eq:DeltaGamma1}
g_1(\omega_{ij})=& \prod_{x=0}^{q-1} \Prob_1(\omega_{ij,x}) *\Bigg\{ p_3^q *p_2^2 * \prod_{x=0}^{q-1} \Big(  \Prob_1(\beta_{ix}) \Prob_1(\beta_{jx}) \Big)\\
& + p_4 * \Bigg( (1-p_2)^2 * \prod_{x=0}^{q-1} \Big(  \Prob_0(\beta_{ix}) \Prob_0(\beta_{jx}) \Big) \\
&+p_2(1-p_2) \Bigg( \prod_{x=0}^{q-1} \Big(  \Prob_0(\beta_{ix}) \Prob_1(\beta_{jx}) \Big) +\prod_{x=0}^{q-1} \Big(  \Prob_1(\beta_{ix}) \Prob_0(\beta_{jx}) \Big) \Bigg)\Bigg) \Bigg\} \\
 &+ \prod_{x=0}^{q-1} \Prob_0(\omega_{ij,x}) *\Bigg\{ (1-p_3)^q *p_2^2 * \prod_{x=0}^{q-1} \Big(  \Prob_1(\beta_{ix}) \Prob_1(\beta_{jx}) \Big)\\
& + (1-p_4) * \Bigg( (1-p_2)^2 * \prod_{x=0}^{q-1} \Big(  \Prob_0(\beta_{ix}) \Prob_0(\beta_{jx}) \Big) \\
&+p_2(1-p_2) \Bigg( \prod_{x=0}^{q-1} \Big(  \Prob_0(\beta_{ix}) \Prob_1(\beta_{jx}) \Big) +\prod_{x=0}^{q-1} \Big(  \Prob_1(\beta_{ix}) \Prob_0(\beta_{jx}) \Big) \Bigg)\Bigg) \Bigg\} \\
&+ p_2^2 \prod_{x=0}^{q-1} \Big(  \Prob_0(\beta_{ix}) \Prob_0(\beta_{jx}) \Big) \sum_{k=1}^{q-1} p_3^k(1-p_3)^{q-k} \sum_{n \in A_q^{(k)}} \prod_{l=0}^{q-1} \Prob_{n_l}(\omega_{ij,l}).
\end{split} 
\end{equation}


\newpage
$\mathbf{2. \Prob(r_{ij,x}=1\mid \gamma_{ij}=1, \theta_i, \theta_j, \Delta)=h_{\theta_i,\theta_j}(\omega_{ij,x})}$ is equal to
\begin{equation}
\label{eq:ProbsRapp}
\begin{split}
=&\Prob(\Delta \mid r_{ij,x}=1,\gamma_{ij}=1,\theta_i,\theta_j)\Prob(r_{ij,x}=1 \mid \gamma_{ij}=1,\theta_i,\theta_j)\\
&\Big[ \Prob(\Delta \mid r_{ij,x}=1,\gamma_{ij}=1,\theta_i,\theta_j)\Prob(r_{ij,x}=1 \mid \gamma_{ij}=1,\theta_i,\theta_j) \\
&+ \Prob(\Delta \mid r_{ij,x}=0,\gamma_{ij}=1,\theta_i,\theta_j)\Prob(r_{ij,x}=0 \mid \gamma_{ij}=1,\theta_i,\theta_j) \Big]^{-1} \\
=&\Prob(\omega_{ij,x} \mid r_{ij,x}=1,\gamma_{ij}=1,\theta_i,\theta_j)\Prob(r_{ij,x}=1 \mid \gamma_{ij}=1,\theta_i,\theta_j)\\
&\Big[ \Prob(\omega_{ij,x} \mid r_{ij,x}=1,\gamma_{ij}=1,\theta_i,\theta_j)\Prob(r_{ij,x}=1 \mid \gamma_{ij}=1,\theta_i,\theta_j) \\
&+ \Prob(\omega_{ij,x} \mid r_{ij,x}=0,\gamma_{ij}=1,\theta_i,\theta_j)\Prob(r_{ij,x}=0 \mid \gamma_{ij}=1,\theta_i,\theta_j) \Big]^{-1} \\
=&\Prob(\omega_{ij,x} \mid r_{ij,x}=1)\Prob(r_{ij,x}=1 \mid \gamma_{ij}=1,\theta_i,\theta_j)\\
&\Big[ \Prob(\omega_{ij,x} \mid r_{ij,x}=1)\Prob(r_{ij,x}=1 \mid \gamma_{ij}=1,\theta_i,\theta_j) \\
&+ \Prob(\omega_{ij,x} \mid r_{ij,x}=0)\Prob(r_{ij,x}=0 \mid \gamma_{ij}=1,\theta_i,\theta_j) \Big]^{-1} \\
\end{split}
\end{equation}
Setting $\Prob(r_{ij,0},\dots,r_{ij,1} \mid \gamma_{ij}, \theta_i, \theta_j)$ as defined in Equation~\eqref{eq:rxq2} we compute
\begin{equation}
\label{eq:rijxMarg}
\begin{split}
\Prob(r_{ij,x} \mid \gamma_{ij}=1, \theta_i, \theta_j) =& \sum_{r_{ij,/x}} \Bigg[\theta_i \theta_j \left(\prod_{x=0}^{q-1} p_3^{r_{ij,x}} (1-p_3)^{1-r_{ij,x}} \right)\\
& + (1-\theta_i \theta_j) p_4^{r_{ij,x}} (1-p_4)^{1-r_{ij,x}} \delta_{r_{ij,/ x}}(r_{ij,x}) \Bigg] \\
=& \theta_i \theta_j p_3^{r_{ij,x}} (1-p_3)^{1-r_{ij,x}}\\ 
& + (1-\theta_i \theta_j)\sum_{r_{ij,/x}}  \Big[ p_4^{r_{ij,x}} (1-p_4)^{1-r_{ij,x}} \delta_{r_{ij,/ x}}(r_{ij,x}) \Big]
\\
=& \theta_i \theta_j p_3^{r_{ij,x}} (1-p_3)^{1-r_{ij,x}}+ (1-\theta_i \theta_j) \Big[p_4^{r_{ij,x}} + (1-p_4)^{(1-r_{ij,x})} \Big]\\
=& \theta_i \theta_j \text{Bernoulli} (r_{ij,x} \mid p_3) +
(1-\theta_i \theta_j)  \text{Bernoulli} (r_{ij,x} \mid p_4)
\end{split}
\end{equation}

\begin{equation}
\label{eq:rijxMarg0}
\begin{split}
\Prob(r_{ij,x} \mid \gamma_{ij}=0, \theta_i, \theta_j) =& \sum_{r_{ij,/x}} \prod_{x = 0}^{q-1}\delta_0(r_{ij,x})
\end{split}
\end{equation}

Substituting Equation~\eqref{eq:rijxMarg}, Equation~\eqref{eq:ProbsRapp} 
\begin{equation}
\label{eq:ProbsRapp2}
h_{\theta_i,\theta_j}(\omega_{ij,x}) =\frac{\Prob_1(\omega_{ij,x}) \Big[ \theta_i \theta_j p_3 + (1-\theta_i \theta_j) p_4 \Big]}{\Prob_1(\omega_{ij,x}) \Big[ \theta_i \theta_j p_3 + (1-\theta_i \theta_j) p_4 \Big] + \Prob_0(\omega_{ij,x}) \Big[ \theta_i \theta_j (1-p_3) + (1-\theta_i \theta_j) (1-p_4) \Big]} 
\end{equation}

Thus $r_{ij,x}^*$ is equal to

\begin{equation}
\label{eq:ErAppP2}
\begin{split}
r_{ij,x}^* =& \Prob(\gamma_{ij}=1\mid \Delta) * \\
&\Big[ \Prob(\theta_i =1 \mid \Delta)\Prob(\theta_j =1 \mid \Delta) \Prob(r_{ij,x}=1\mid \gamma_{ij}=1, \theta_i=1,\theta_j=1, \Delta) +\\
&\Prob(r_{ij,/x}=1\mid\Delta)\Big\{  \Prob(\theta_i =1 \mid \Delta)\Prob(\theta_j =0 \mid \Delta) \Prob(r_{ij,x}=1\mid \gamma_{ij}=1, \theta_i=1,\theta_j=0, \Delta) \\
&+  \Prob(\theta_i =0 \mid \Delta)\Prob(\theta_j =1\mid \Delta)\Prob(r_{ij,x}=1\mid \gamma_{ij}=1, \theta_i=0,\theta_j=1, \Delta)\\
&+\Prob(\theta_i =0 \mid \Delta)\Prob(\theta_j =0 \mid \Delta) \Prob(r_{ij,x}=1\mid \gamma_{ij}=1, \theta_i=0,\theta_j=0, \Delta) \Big\} \Big] \\
=& \gamma^*_{ij} \Bigg[ \theta^*_i  \theta^*_j h_1(\omega_{ij,x}) + r^*_{ij,/x} h_{0}(\omega_{ij,x}) \Big\{ \theta^*_i (1- \theta^*_j) + (1-\theta^*_i) \theta^*_j + (1- \theta^*_i) (1- \theta^*_j) \Big\} \Bigg]
\end{split}
\end{equation}
In which
\begin{equation}
\begin{split}
\gamma_{ij}^*= \frac{g_1(\omega_{ij}) p_1}{g_1(\omega_{ij}) p_1+g_0(\omega_{ij}) (1-p_1)}
\end{split},
\end{equation}

From Equation~\eqref{eq:ProbsRapp2} we see that $h_{1,0}(\omega_{ij,x})=h_{0,1}(\omega_{ij,x}) =h_{0,0}(\omega_{ij,x})  =h_{0}(\omega_{ij,x})$. We then defined $h_{1,1}(\omega_{ij,x})=h_{1}(\omega_{ij,x})$, getting: 
\begin{equation}
\begin{split}
h_{1}(\omega_{ij,x}) &=\frac{\Prob_1(\omega_{ij,x}) p_3}{\Prob_1(\omega_{ij,x}) p_3  + \Prob_0(\omega_{ij,x})(1-p_3)} \\
h_{0}(\omega_{ij,x}) &=\frac{\Prob_1(\omega_{ij,x})  p_4}{\Prob_1(\omega_{ij,x})  p_4 + \Prob_0(\omega_{ij,x})  (1-p_4)} 
\end{split}
\end{equation}

In the \textbf{M-step} we maximise $Q(\Delta)$ w.r.t. $(\Omega_x, \beta_x, \alpha)$. 
Without lost of generality, we consider that we have mean centered observations before proceeding, i.e.\ $\alpha = \{0,\dots, 0 \}$.

Let $Y_x \in \R^{n \times p}$ a matrix with column $k$ equal to the $Y_{x,k}$ for individual $k$, $\beta \in \R^{p \times 1} =[\beta_{ix}]_{i \in V} $ and $\1 \in \R^{1 \times n_x}$ a column vector of ones. 
Setting $\Exp_{\theta\mid \widehat{\Delta}, Y_{V_x}} \left[ \frac{1}{\theta_i \lambda_1}+ \frac{1}{(1-\theta_i) \lambda_0} \right] =\frac{\theta_i^*}{\lambda_1}+\frac{1-\theta_i^*}{\lambda_0}$ and $\Exp_{R\mid \widehat{\Delta},Y_{V_x}} \left[ \frac{r_{ij,x}}{\nu_1}+ \frac{1-r_{ij,x}}{\nu_0} \right] = \frac{r_{ij,x}^*}{\nu_1}+\frac{1-r_{ij,x}^*}{\nu_0}$, Equation~\eqref{eqApp:Q2} can be seen as:

\begin{equation}
\begin{split}
Q \propto &   \frac{1}{2} \sum_{x=0}^{q-1} \left[n_x \log (\det (\widehat{\Omega}_x)) - \text{tr} \left\{  \left( Y_x^{\top} - \1 \widehat{\beta}_x \right)  \left( Y_x^{\top} - \1 \widehat{\beta}_x \right)^\top \widehat{\Omega}_x \right\} \right] \\
& - \sum_{x=0}^{q-1} \sum_{i=1}^p \left[ \frac{1}{2} \widehat{\beta}_{ix}^2 \left(\frac{\theta_i^*}{\lambda_1}+\frac{1-\theta_i^*}{\lambda_0} \right) + \tau \widehat{\omega}_{ii,x}  \right]  \\
&-\sum_{x=0}^{q-1} \sum_{i<j} \left[ |\widehat{\omega}_{ij,x}|  \left(\frac{r_{ij,x}^*}{\nu_1}+\frac{1-r_{ij,x}^*}{\nu_0} \right) \right] \nonumber
\end{split}
\end{equation}

\begin{equation}
\begin{split}
Q \propto
& \frac{1}{2} \sum_{x=0}^{q-1} \left[n_x \log (\det (\widehat{\Omega}_x)) - \text{tr} \left\{  \left( Y_x^{\top} Y_x - 2  \widehat{\beta}_x \1 Y_x^{\top} -  \widehat{\beta}_x \1 \1^{\top} \widehat{\beta}^{\top}_x \right) \widehat{\Omega}_x \right\} \right] \\
& - \sum_{x=0}^{q-1} \sum_{i=1}^p \left[ \frac{1}{2} \widehat{\beta}_{ix}^2 \left(\frac{\theta_i^*}{\lambda_1}+\frac{1-\theta_i^*}{\lambda_0} \right) + \tau \widehat{\omega}_{ii,x}  \right]  \\
&-\sum_{x=0}^{q-1} \sum_{i<j} \left[ |\widehat{\omega}_{ij,x}|  \left(\frac{r_{ij,x}^*}{\nu_1}+\frac{1-r_{ij,x}^*}{\nu_0} \right) \right]
\end{split}
\end{equation}

To $Q$ w.r.t. $\beta_x$, we set the partial derivative $\beta_x$, to 0:

\begin{equation}
\label{eqApp:Dbeta}
\frac{\partial Q}{\partial \beta_x} = \widehat{\Omega}^\top_x Y_x^\top \1^\top - \widehat{\Omega}^\top_x \widehat{\beta}_x \1 \1^\top - \diag \left \{ \frac{\theta_1^*}{\lambda_1}+\frac{1-\theta_1^*}{\lambda_0} \dots  \frac{\theta_p^*}{\lambda_1}+\frac{1-\theta_p^*}{\lambda_0} \right\} \widehat{\beta}_x =0 ,
\end{equation}
solving Equation~\eqref{eqApp:Dbeta}

\begin{equation}
\label{eqApp:Mbeta}
\widehat{\beta}_x = \left(n_x \widehat{\Omega}_x +  D_{\Theta^*} \right)^{-1} \widehat{\Omega}_x Y_x^\top \1^\top ,
\end{equation}
with $D_{\Theta^*}  = \diag \left \{ \frac{\theta_1^*}{\lambda_1}+\frac{1-\theta_1^*}{\lambda_0} \dots  \frac{\theta_p^*}{\lambda_1}+\frac{1-\theta_p^*}{\lambda_0} \right\} $.
Equation~\eqref{eqApp:Mbeta} has the form of a ridge regression estimator with penalty $ D_{\Theta^*}$ .

Maximising Q w.r.t $\widehat{\Omega}_x$ for each $x=0,\dots,q-1$, implies optimizing the following objective function:
\begin{equation}
\begin{split}
Q(\widehat{\Omega}_x)= 
& \frac{n_x}{2}\log (\det (\widehat{\Omega}_x)) - \frac{n_x}{2} \text{tr} \left\{ \widehat{S}_x  \widehat{\Omega}_x \right\} - \sum_{i=1}^p  \tau \widehat{\omega}_{ii,x} -\sum_{i<j} \left[ |\widehat{\omega}_{ij,x}|  \left(\frac{r_{ij,x}^*}{\nu_1}+\frac{1-r_{ij,x}^*}{\nu_0} \right) \right],
\end{split}
\end{equation}
with $\widehat{S}_x = \frac{1}{n_x} \sum_{k=1}^{n_x}(Y_{x,k}-\widehat{\beta}_x)(Y_{x,k}-\widehat{\beta}_x)^\top$. 

We optimize $Q(\widehat{\Omega}_x)$ subject to the constraints that $\widehat{\Omega}_x \succ 0$ and $|| \widehat{\Omega}_x ||_2 \leq B $, with a reasonably large $B$ to obtain an objective function $Q(\widehat{\Omega}_x)$ strictly convex, and guaranteed that the local solution $\widehat{\Omega}_x$ is the unique solution as in \citet{YangXinming2021GGEo}.  
In order to optimize $Q(\widehat{\Omega}_x)$, we follow the algorithm suggested in \citet{GanLingrui2019BRfG}.

\section{Simulations: Data generating mechanism}
We generate observations from a set $Y_{V|\X}$ of random vectors associated to a profile undirected graph $\GU$ with $p=20,50$ and $100$ nodes and $q=4$ levels of $X$, such that $x \in \X = \{0,1,2,3\}$. 
Following \citet{peterson2015bayesian}, we first construct $\Omega_0$, the precision matrix of the baseline level $x=0$. We set $\Omega_0$ to be a $p \times p$ symmetric matrix with main diagonal entries $\omega_{aa,0}=a$, with $a=1, \dots, p$, and off-diagonal entries $\omega_{(a+1)a,0}=\omega_{a(a+1),0}=0.5$ with $a=1, \dots, 19$  and  $\omega_{(a+2)a,0}=\omega_{a(a+2),0}=0.4$ with $a=1, \dots, 18$ . For all $a \in V$, we set  both $\alpha_{a}$ and $\zeta_{ax}$ to zero. For $x=\{1,2,3\}$, we also set $\zeta_{ax}=0$ for $a=5, \dots 20$ and $\zeta_{ax}=1$ for $a=1, \dots 4$, i.e., the external factor $X$ affects only the first four response variables. The remaining precision matrices $\Omega_x$ for $x=\{1,2,3\}$ are obtained as follow, first we set $\Omega_x = \Omega_0$, then with probability 0.5 we set to zero its non-zero entries.
We then change the sparcity level of the precision matrices, varying the $s$ parameter from 0.0010 to 0.0050 to have increasing number of non-zero elements. 
Data are generated by drawing a random sample of size $n_x = 50$ from the distribution $\mathcal{N}(\beta_{x},\Sigma_x)$ where $\Sigma_x=\Omega_x^{-1}$ and $\beta_{x}=\Sigma_x\zeta_{x}$, for all $x \in \X$. 
\section{More results}
\begin{table}[ht]
\centering
\caption{Accuracy, sensitivity, specificity, balanced accuracy and AUC over  \textbf{100} datasets, K=4, N=50 Scenario 1: different graph}
\scalebox{.85}{
\begin{tabular}{|c|c|c|c|c|c|c|c|c|c|c|c|}
\multicolumn{1}{l}{}                                       &\multicolumn{1}{l|}{}                                       & \multicolumn{2}{c|}{\cellcolor[HTML]{1F4E78}{\color[HTML]{FFFFFF} \textbf{Accuracy}}}                                                                               & \multicolumn{2}{c|}{\cellcolor[HTML]{1F4E78}{\color[HTML]{FFFFFF} \textbf{Sensitivity}}}                                                                      & \multicolumn{2}{c||}{\cellcolor[HTML]{1F4E78}{\color[HTML]{FFFFFF} \textbf{Specificity}}}                                                                     & \multicolumn{2}{c|}{\cellcolor[HTML]{1F4E78}{\color[HTML]{FFFFFF} \textbf{AUC}}}                                                                      \\ 
\multicolumn{1}{l}{}                                       & \multicolumn{1}{l|}{}                                       & \multicolumn{1}{c|}{\cellcolor[HTML]{1F4E78}{\color[HTML]{FFFFFF} \textbf{Mean}}} & \multicolumn{1}{c|}{\cellcolor[HTML]{1F4E78}{\color[HTML]{FFFFFF} \textbf{SE}}} & \multicolumn{1}{c|}{\cellcolor[HTML]{1F4E78}{\color[HTML]{FFFFFF} \textbf{Mean}}} & \multicolumn{1}{c|}{\cellcolor[HTML]{1F4E78}{\color[HTML]{FFFFFF} \textbf{SE}}} & \multicolumn{1}{c|}{\cellcolor[HTML]{1F4E78}{\color[HTML]{FFFFFF} \textbf{Mean}}} & \multicolumn{1}{c||}{\cellcolor[HTML]{1F4E78}{\color[HTML]{FFFFFF} \textbf{SE}}} & \multicolumn{1}{c|}{\cellcolor[HTML]{1F4E78}{\color[HTML]{FFFFFF} \textbf{Mean}}} & \multicolumn{1}{c|}{\cellcolor[HTML]{1F4E78}{\color[HTML]{FFFFFF} \textbf{SE}}}  \\ 
\multicolumn{1}{l}{}                                       &\multicolumn{1}{l|}{}                                       & \multicolumn{8}{c|}{\textbf{$p=20$}}                                                                                   \\ 
\hline
\rowcolor[HTML]{DAE8FC} 
\cellcolor{white}
\multirow{5}{*}{\textbf{S = 0.010}}
&\textbf{BPUGM} & 0.914 & 0.000 & 0.837 & 0.026 & 0.914 & 0.000 & 0.875 & 0.006 \\ 
  &\textbf{GemBag} 
  & 0.996 & 0.000 & 0.572 & 0.064 & 1.000 & 0.000 & 0.786 & 0.016 \\ 
  &\textbf{FGL}  & 0.983 & 0.001 & 0.603 & 0.096 & 0.986 & 0.001 & 0.794 & 0.023 \\  
  &\textbf{GGL} & 0.982 & 0.001 & 0.665 & 0.083 & 0.984 & 0.001 & 0.825 & 0.020 \\ 
  \hline
\rowcolor[HTML]{DAE8FC} 
\cellcolor{white}
\multirow{5}{*}{\textbf{S = 0.025}}
&\textbf{BPUGM} & 0.899 & 0.000 & 0.431 & 0.007 & 0.915 & 0.000 & 0.673 & 0.002 \\  
  &\textbf{GemBag}
  & 0.972 & 0.000 & 0.182 & 0.007 & 1.000 & 0.000 & 0.591 & 0.002 \\ 
  &\textbf{FGL} & 0.943 & 0.001 & 0.275 & 0.022 & 0.967 & 0.002 & 0.621 & 0.004 \\ 
  &\textbf{GGL} & 0.948 & 0.001 & 0.298 & 0.021 & 0.971 & 0.002 & 0.635 & 0.004 \\ 
  \hline
\rowcolor[HTML]{DAE8FC} 
\cellcolor{white}
\multirow{5}{*}{\textbf{S = 0.050}}
&\textbf{BPUGM} & 0.804 & 0.000 & 0.195 & 0.002 & 0.915 & 0.000 & 0.555 & 0.000 \\  
  &\textbf{GemBag} 
  & 0.851 & 0.000 & 0.033 & 0.000 & 0.999 & 0.000 & 0.516 & 0.000 \\ 
  &\textbf{FGL} & 0.831 & 0.001 & 0.110 & 0.008 & 0.962 & 0.003 & 0.536 & 0.001 \\ 
  &\textbf{GGL} & 0.828 & 0.001 & 0.129 & 0.007 & 0.955 & 0.003 & 0.542 & 0.001 \\ 
  \hline
  \multicolumn{1}{l}{}                                       &\multicolumn{1}{l|}{}                                       & \multicolumn{8}{c|}{\textbf{$p=50$}}                                                                                   \\ 
\hline
\rowcolor[HTML]{DAE8FC} 
\cellcolor{white}
\multirow{5}{*}{\textbf{S = 0.010}}
&\textbf{BPUGM} & 0.948 & 0.000 & 0.790 & 0.051 & 0.949 & 0.000 & 0.869 & 0.013 \\ 
  &\textbf{GemBag}
  & 0.999 & 0.000 & 0.623 & 0.150 & 1.000 & 0.000 & 0.811 & 0.037 \\ 
  &\textbf{FGL} & 0.998 & 0.000 & 0.353 & 0.190 & 0.998 & 0.000 & 0.676 & 0.047 \\  
  &\textbf{GGL} & 0.997 & 0.000 & 0.490 & 0.194 & 0.997 & 0.000 & 0.744 & 0.048 \\ 
  \hline
\rowcolor[HTML]{DAE8FC} 
\cellcolor{white}
\multirow{5}{*}{\textbf{S = 0.0025}}
&\textbf{BPUGM} & 0.948 & 0.000 & 0.596 & 0.018 & 0.949 & 0.000 & 0.773 & 0.004 \\ 
  &\textbf{GemBag}
  & 0.998 & 0.000 & 0.362 & 0.023 & 1.000 & 0.000 & 0.681 & 0.006 \\
  &\textbf{FGL}  & 0.995 & 0.000 & 0.340 & 0.075 & 0.997 & 0.000 & 0.668 & 0.018 \\ 
  &\textbf{GGL} & 0.995 & 0.000 & 0.380 & 0.061 & 0.997 & 0.000 & 0.688 & 0.015 \\
  \hline
\rowcolor[HTML]{DAE8FC} 
\cellcolor{white}
\multirow{5}{*}{\textbf{S = 0.0050}}
&\textbf{BPUGM} & 0.946 & 0.000 & 0.429 & 0.005 & 0.950 & 0.000 & 0.690 & 0.001 \\ 
  &\textbf{GemBag}
  & 0.993 & 0.000 & 0.240 & 0.007 & 1.000 & 0.000 & 0.620 & 0.002 \\ 
  &\textbf{FGL} & 0.988 & 0.000 & 0.140 & 0.015 & 0.995 & 0.000 & 0.568 & 0.003 \\
  &\textbf{GGL} & 0.991 & 0.000 & 0.271 & 0.033 & 0.997 & 0.000 & 0.634 & 0.008 \\
  \hline
  \multicolumn{1}{l}{}                                       &\multicolumn{1}{l|}{}                                       & \multicolumn{8}{c|}{\textbf{$p=100$}}                                                                                   \\ 
\hline
\rowcolor[HTML]{DAE8FC} 
\cellcolor{white}
\multirow{5}{*}{\textbf{S = 0.010}}
&\textbf{BPUGM} & 0.965 & 0.000 & 0.520 & 0.008 & 0.966 & 0.000 & 0.743 & 0.002 \\ 
  &\textbf{GemBag} & 0.999 & 0.000 & 0.378 & 0.013 & 1.000 & 0.000 & 0.689 & 0.003 \\ 
  &\textbf{FGL} & 0.999 & 0.000 & 0.223 & 0.047 & 1.000 & 0.000 & 0.611 & 0.012 \\
  &\textbf{GGL} & 0.999 & 0.000 & 0.277 & 0.043 & 1.000 & 0.000 & 0.638 & 0.011 \\
  \hline
\rowcolor[HTML]{DAE8FC} 
\cellcolor{white}
\multirow{5}{*}{\textbf{S = 0.0025}}
&\textbf{BPUGM} & 0.965 & 0.000 & 0.546 & 0.004 & 0.966 & 0.000 & 0.756 & 0.001 \\
  &\textbf{GemBag} & 0.998 & 0.000 & 0.368 & 0.006 & 1.000 & 0.000 & 0.684 & 0.002 \\
  &\textbf{FGL} & 0.997 & 0.000 & 0.285 & 0.031 & 0.999 & 0.000 & 0.642 & 0.008 \\  
  &\textbf{GGL} & 0.997 & 0.000 & 0.374 & 0.024 & 0.999 & 0.000 & 0.687 & 0.006 \\
  \hline
\rowcolor[HTML]{DAE8FC} 
\cellcolor{white}
\multirow{5}{*}{\textbf{S = 0.0050}}
&\textbf{BPUGM} & 0.959 & 0.000 & 0.235 & 0.001 & 0.966 & 0.000 & 0.600 & 0.000 \\
  &\textbf{GemBag} & 0.991 & 0.000 & 0.123 & 0.000 & 1.000 & 0.000 & 0.561 & 0.000 \\ 
  &\textbf{FGL} & 0.989 & 0.000 & 0.100 & 0.004 & 0.998 & 0.000 & 0.549 & 0.001 \\
  &\textbf{GGL} & 0.990 & 0.000 & 0.127 & 0.003 & 0.998 & 0.000 & 0.563 & 0.001 \\
  \hline
\end{tabular}}
\end{table}
   
\begin{table}[ht]
\centering
\caption{Accuracy, sensitivity, specificity, balanced accuracy and AUC over  \textbf{100} datasets, K=4, N=50, Scenario 2: $\{G(0)=G(1)\} \neq \{G(2)=G(3)\}$}
\scalebox{.85}{
\begin{tabular}{|c|c|c|c|c|c|c|c|c|c|}
\multicolumn{1}{l}{}                                       &\multicolumn{1}{l|}{}                                       & \multicolumn{2}{c|}{\cellcolor[HTML]{1F4E78}{\color[HTML]{FFFFFF} \textbf{Accuracy}}}                                                                               & \multicolumn{2}{c|}{\cellcolor[HTML]{1F4E78}{\color[HTML]{FFFFFF} \textbf{Sensitivity}}}                                                                      & \multicolumn{2}{c||}{\cellcolor[HTML]{1F4E78}{\color[HTML]{FFFFFF} \textbf{Specificity}}}                                                                     & \multicolumn{2}{c|}{\cellcolor[HTML]{1F4E78}{\color[HTML]{FFFFFF} \textbf{AUC}}}                                                                      \\ 
\multicolumn{1}{l}{}                                       & \multicolumn{1}{l|}{}                                       & \multicolumn{1}{c|}{\cellcolor[HTML]{1F4E78}{\color[HTML]{FFFFFF} \textbf{Mean}}} & \multicolumn{1}{c|}{\cellcolor[HTML]{1F4E78}{\color[HTML]{FFFFFF} \textbf{SE}}} & \multicolumn{1}{c|}{\cellcolor[HTML]{1F4E78}{\color[HTML]{FFFFFF} \textbf{Mean}}} & \multicolumn{1}{c|}{\cellcolor[HTML]{1F4E78}{\color[HTML]{FFFFFF} \textbf{SE}}} & \multicolumn{1}{c|}{\cellcolor[HTML]{1F4E78}{\color[HTML]{FFFFFF} \textbf{Mean}}} & \multicolumn{1}{c||}{\cellcolor[HTML]{1F4E78}{\color[HTML]{FFFFFF} \textbf{SE}}} & \multicolumn{1}{c|}{\cellcolor[HTML]{1F4E78}{\color[HTML]{FFFFFF} \textbf{Mean}}} & \multicolumn{1}{c|}{\cellcolor[HTML]{1F4E78}{\color[HTML]{FFFFFF} \textbf{SE}}}  \\ 
\multicolumn{1}{l}{}                                       &\multicolumn{1}{l|}{}                                       & \multicolumn{8}{c|}{\textbf{$p=20$}}                                                                                   \\ 
\hline
\rowcolor[HTML]{DAE8FC} 
\cellcolor{white}
\multirow{5}{*}{\textbf{S = 0.010}}
&\textbf{BPUGM} & 0.913 & 0.000 & 0.837 & 0.025 & 0.914 & 0.000 & 0.875 & 0.006 \\ 
  &\textbf{GemBag}
  & 0.996 & 0.000 & 0.567 & 0.052 & 1.000 & 0.000 & 0.783 & 0.013 \\ 
  &\textbf{FGL} & 0.974 & 0.002 & 0.708 & 0.069 & 0.976 & 0.002 & 0.842 & 0.016 \\ 
  &\textbf{GGL} & 0.968 & 0.002 & 0.708 & 0.084 & 0.970 & 0.002 & 0.839 & 0.019 \\ 
  \hline
\rowcolor[HTML]{DAE8FC} 
\cellcolor{white}
\multirow{5}{*}{\textbf{S = 0.025}}
&\textbf{BPUGM} & 0.894 & 0.000 & 0.399 & 0.006 & 0.914 & 0.000 & 0.657 & 0.001 \\  
  &\textbf{GemBag}
  & 0.967 & 0.000 & 0.174 & 0.006 & 1.000 & 0.000 & 0.587 & 0.002 \\ 
  &\textbf{FGL} & 0.944 & 0.001 & 0.254 & 0.023 & 0.972 & 0.002 & 0.613 & 0.004 \\ 
  &\textbf{GGL} & 0.946 & 0.001 & 0.266 & 0.017 & 0.974 & 0.002 & 0.620 & 0.003 \\ 
  \hline
\rowcolor[HTML]{DAE8FC} 
\cellcolor{white}
\multirow{5}{*}{\textbf{S = 0.050}}
&\textbf{BPUGM} & 0.804 & 0.000 & 0.213 & 0.001 & 0.915 & 0.000 & 0.564 & 0.000 \\ 
  &\textbf{GemBag}
  & 0.849 & 0.000 & 0.048 & 0.000 & 0.999 & 0.000 & 0.524 & 0.000 \\ 
  &\textbf{FGL} & 0.845 & 0.000 & 0.083 & 0.005 & 0.987 & 0.001 & 0.535 & 0.001 \\ 
  &\textbf{GGL} & 0.837 & 0.001 & 0.097 & 0.006 & 0.976 & 0.001 & 0.537 & 0.001 \\
  \hline
  \multicolumn{1}{l}{}                                       &\multicolumn{1}{l|}{}                                       & \multicolumn{8}{c|}{\textbf{$p=50$}}                                                                                   \\ 
\hline
\rowcolor[HTML]{DAE8FC} 
\cellcolor{white}
\multirow{5}{*}{\textbf{S = 0.0010}}
&\textbf{BPUGM} & 0.948 & 0.000 & 0.815 & 0.069 & 0.948 & 0.000 & 0.881 & 0.017 \\ 
  &\textbf{GemBag}
  & 0.999 & 0.000 & 0.505 & 0.220 & 1.000 & 0.000 & 0.752 & 0.055 \\
  &\textbf{FGL} & 0.998 & 0.000 & 0.225 & 0.158 & 0.998 & 0.000 & 0.612 & 0.039 \\ 
  &\textbf{GGL} & 0.995 & 0.000 & 0.415 & 0.197 & 0.995 & 0.000 & 0.705 & 0.049 \\ 
  \hline
\rowcolor[HTML]{DAE8FC} 
\cellcolor{white}
\multirow{5}{*}{\textbf{S = 0.0025}}
&\textbf{BPUGM} & 0.948 & 0.000 & 0.601 & 0.014 & 0.949 & 0.000 & 0.775 & 0.004 \\ 
  &\textbf{GemBag}
  & 0.998 & 0.000 & 0.383 & 0.017 & 1.000 & 0.000 & 0.691 & 0.004 \\
  &\textbf{FGL} & 0.996 & 0.000 & 0.531 & 0.099 & 0.998 & 0.000 & 0.765 & 0.025 \\ 
  &\textbf{GGL} & 0.995 & 0.000 & 0.446 & 0.079 & 0.997 & 0.000 & 0.722 & 0.020 \\
  \hline
\rowcolor[HTML]{DAE8FC} 
\cellcolor{white}
\multirow{5}{*}{\textbf{S = 0.0050}}
&\textbf{BPUGM} & 0.946 & 0.000 & 0.429 & 0.005 & 0.950 & 0.000 & 0.690 & 0.001 \\ 
  &\textbf{GemBag}
  & 0.993 & 0.000 & 0.240 & 0.007 & 1.000 & 0.000 & 0.620 & 0.002 \\ 
  &\textbf{FGL}  & 0.989 & 0.000 & 0.136 & 0.014 & 0.996 & 0.000 & 0.566 & 0.003 \\ 
  &\textbf{GGL} & 0.992 & 0.000 & 0.270 & 0.034 & 0.998 & 0.000 & 0.634 & 0.008 \\ 
  \hline
  \multicolumn{1}{l}{}                                       &\multicolumn{1}{l|}{}                                       & \multicolumn{8}{c|}{\textbf{$p=100$}}                                                                                   \\ 
\hline
\rowcolor[HTML]{DAE8FC} 
\cellcolor{white}
\multirow{5}{*}{\textbf{S = 0.010}}
&\textbf{BPUGM} & 0.965 & 0.000 & 0.549 & 0.008 & 0.965 & 0.000 & 0.757 & 0.002 \\
  &\textbf{GemBag} & 0.999 & 0.000 & 0.415 & 0.014 & 1.000 & 0.000 & 0.707 & 0.003 \\ 
  &\textbf{FGL} & 0.999 & 0.000 & 0.528 & 0.038 & 0.999 & 0.000 & 0.764 & 0.009 \\ 
  &\textbf{GGL} & 0.999 & 0.000 & 0.507 & 0.040 & 0.999 & 0.000 & 0.753 & 0.010 \\ 
  \hline
\rowcolor[HTML]{DAE8FC} 
\cellcolor{white}
\multirow{5}{*}{\textbf{S = 0.0025}}
&\textbf{BPUGM} & 0.965 & 0.000 & 0.594 & 0.004 & 0.965 & 0.000 & 0.780 & 0.001 \\
  &\textbf{GemBag} & 0.998 & 0.000 & 0.412 & 0.008 & 1.000 & 0.000 & 0.706 & 0.002 \\
  &\textbf{FGL} & 0.998 & 0.000 & 0.423 & 0.044 & 0.999 & 0.000 & 0.711 & 0.011 \\  
  &\textbf{GGL} & 0.998 & 0.000 & 0.413 & 0.045 & 0.999 & 0.000 & 0.706 & 0.011 \\ 
  \hline
\rowcolor[HTML]{DAE8FC} 
\cellcolor{white}
\multirow{5}{*}{\textbf{S = 0.0050}}
&\textbf{BPUGM} & 0.959 & 0.000 & 0.269 & 0.001 & 0.966 & 0.000 & 0.617 & 0.000 \\
  &\textbf{GemBag} & 0.991 & 0.000 & 0.147 & 0.001 & 1.000 & 0.000 & 0.573 & 0.000 \\ 
  &\textbf{FGL} & 0.990 & 0.000 & 0.240 & 0.012 & 0.997 & 0.000 & 0.619 & 0.003 \\
  &\textbf{GGL} & 0.989 & 0.000 & 0.186 & 0.005 & 0.997 & 0.000 & 0.592 & 0.001 \\
  \hline
\end{tabular}}
\end{table}

\begin{table}[ht]
\centering
\caption{Accuracy, sensitivity, specificity, balanced accuracy and AUC over  \textbf{100} datasets, K=4, N=50, Scenario 3: $\{G(0)=G(1)=G(2)\} \neq \{G(3)\}$}
\scalebox{.85}{
\begin{tabular}{|c|c|c|c|c|c|c|c|c|c|}
\multicolumn{1}{l}{}                                       &\multicolumn{1}{l|}{}                                       & \multicolumn{2}{c|}{\cellcolor[HTML]{1F4E78}{\color[HTML]{FFFFFF} \textbf{Accuracy}}}                                                                               & \multicolumn{2}{c|}{\cellcolor[HTML]{1F4E78}{\color[HTML]{FFFFFF} \textbf{Sensitivity}}}                                                                      & \multicolumn{2}{c||}{\cellcolor[HTML]{1F4E78}{\color[HTML]{FFFFFF} \textbf{Specificity}}}                                                                     & \multicolumn{2}{c|}{\cellcolor[HTML]{1F4E78}{\color[HTML]{FFFFFF} \textbf{AUC}}}                                                                      \\ 
\multicolumn{1}{l}{}                                       & \multicolumn{1}{l|}{}                                       & \multicolumn{1}{c|}{\cellcolor[HTML]{1F4E78}{\color[HTML]{FFFFFF} \textbf{Mean}}} & \multicolumn{1}{c|}{\cellcolor[HTML]{1F4E78}{\color[HTML]{FFFFFF} \textbf{SE}}} & \multicolumn{1}{c|}{\cellcolor[HTML]{1F4E78}{\color[HTML]{FFFFFF} \textbf{Mean}}} & \multicolumn{1}{c|}{\cellcolor[HTML]{1F4E78}{\color[HTML]{FFFFFF} \textbf{SE}}} & \multicolumn{1}{c|}{\cellcolor[HTML]{1F4E78}{\color[HTML]{FFFFFF} \textbf{Mean}}} & \multicolumn{1}{c||}{\cellcolor[HTML]{1F4E78}{\color[HTML]{FFFFFF} \textbf{SE}}} & \multicolumn{1}{c|}{\cellcolor[HTML]{1F4E78}{\color[HTML]{FFFFFF} \textbf{Mean}}} & \multicolumn{1}{c|}{\cellcolor[HTML]{1F4E78}{\color[HTML]{FFFFFF} \textbf{SE}}}  \\ 
\multicolumn{1}{l}{}                                       &\multicolumn{1}{l|}{}                                       & \multicolumn{8}{c|}{\textbf{$p=20$}}  \\
\hline
\rowcolor[HTML]{DAE8FC} 
\cellcolor{white}
\multirow{5}{*}{\textbf{S = 0.010}}
&\textbf{BPUGM} & 0.913 & 0.000 & 0.843 & 0.039 & 0.913 & 0.000 & 0.878 & 0.010 \\ 
  &\textbf{GemBag}
  & 0.998 & 0.000 & 0.677 & 0.091 & 1.000 & 0.000 & 0.839 & 0.023 \\ 
  &\textbf{FGL} & 0.990 & 0.000 & 0.698 & 0.167 & 0.991 & 0.000 & 0.844 & 0.042 \\ 
  &\textbf{GGL} & 0.985 & 0.001 & 0.703 & 0.164 & 0.986 & 0.001 & 0.844 & 0.040 \\ 
  \hline
\rowcolor[HTML]{DAE8FC} 
\cellcolor{white}
\multirow{5}{*}{\textbf{S = 0.025}}
&\textbf{BPUGM} & 0.892 & 0.000 & 0.297 & 0.005 & 0.915 & 0.000 & 0.606 & 0.001 \\
  &\textbf{GemBag}
  & 0.967 & 0.000 & 0.139 & 0.004 & 0.999 & 0.000 & 0.569 & 0.001 \\ 
  &\textbf{FGL}  & 0.947 & 0.001 & 0.200 & 0.011 & 0.976 & 0.001 & 0.588 & 0.002 \\ 
  &\textbf{GGL} & 0.948 & 0.001 & 0.198 & 0.011 & 0.978 & 0.001 & 0.588 & 0.002 \\ 
  \hline
\rowcolor[HTML]{DAE8FC} 
\cellcolor{white}
\multirow{5}{*}{\textbf{S = 0.050}}
&\textbf{BPUGM} & 0.797 & 0.000 & 0.156 & 0.001 & 0.914 & 0.000 & 0.535 & 0.000 \\  
  &\textbf{GemBag}
  & 0.849 & 0.000 & 0.023 & 0.000 & 1.000 & 0.000 & 0.511 & 0.000 \\ 
  &\textbf{FGL} & 0.832 & 0.001 & 0.076 & 0.005 & 0.970 & 0.002 & 0.523 & 0.000 \\ 
  &\textbf{GGL} & 0.826 & 0.001 & 0.085 & 0.006 & 0.960 & 0.002 & 0.523 & 0.000 \\ 
  \hline
  \multicolumn{1}{l}{}                                       &\multicolumn{1}{l|}{}                                       & \multicolumn{8}{c|}{\textbf{$p=50$}}                                                                                   \\ 
\hline
\rowcolor[HTML]{DAE8FC} 
\cellcolor{white}
\multirow{5}{*}{\textbf{S = 0.0010}}
&\textbf{BPUGM} & 0.948 & 0.000 & 0.740 & 0.194 & 0.948 & 0.000 & 0.844 & 0.048 \\ 
  &\textbf{GemBag}
  & 0.999 & 0.000 & 0.260 & 0.194 & 1.000 & 0.000 & 0.630 & 0.049 \\
  &\textbf{FGL} & 0.999 & 0.000 & 0.040 & 0.039 & 1.000 & 0.000 & 0.520 & 0.010 \\ 
  &\textbf{GGL} & 0.997 & 0.000 & 0.140 & 0.122 & 0.997 & 0.000 & 0.569 & 0.030 \\
  \hline
\rowcolor[HTML]{DAE8FC} 
\cellcolor{white}
\multirow{5}{*}{\textbf{S = 0.0025}}
&\textbf{BPUGM} & 0.948 & 0.000 & 0.626 & 0.017 & 0.949 & 0.000 & 0.788 & 0.004 \\
  &\textbf{GemBag}
  & 0.998 & 0.000 & 0.394 & 0.023 & 1.000 & 0.000 & 0.697 & 0.006 \\
  &\textbf{FGL} & 0.997 & 0.000 & 0.440 & 0.088 & 0.998 & 0.000 & 0.719 & 0.022 \\
  &\textbf{GGL} & 0.995 & 0.000 & 0.438 & 0.074 & 0.997 & 0.000 & 0.717 & 0.018 \\
  \hline
\rowcolor[HTML]{DAE8FC} 
\cellcolor{white}
\multirow{5}{*}{\textbf{S = 0.0050}}
&\textbf{BPUGM} & 0.946 & 0.000 & 0.429 & 0.005 & 0.950 & 0.000 & 0.690 & 0.001 \\ 
  &\textbf{GemBag}
  & 0.993 & 0.000 & 0.240 & 0.007 & 1.000 & 0.000 & 0.620 & 0.002 \\
  &\textbf{FGL} & 0.989 & 0.000 & 0.136 & 0.014 & 0.996 & 0.000 & 0.566 & 0.003 \\
  &\textbf{GGL} & 0.992 & 0.000 & 0.270 & 0.034 & 0.998 & 0.000 & 0.634 & 0.008 \\ 
  \hline
  \multicolumn{1}{l}{}                                       &\multicolumn{1}{l|}{}                                       & \multicolumn{8}{c|}{\textbf{$p=100$}}                                                                                   \\ 
\hline
\rowcolor[HTML]{DAE8FC} 
\cellcolor{white}
\multirow{5}{*}{\textbf{S = 0.010}}
&\textbf{BPUGM} & 0.965 & 0.000 & 0.515 & 0.008 & 0.965 & 0.000 & 0.740 & 0.002 \\
  &\textbf{GemBag} & 0.999 & 0.000 & 0.397 & 0.013 & 1.000 & 0.000 & 0.698 & 0.003 \\  
  &\textbf{FGL} & 0.999 & 0.000 & 0.485 & 0.039 & 0.999 & 0.000 & 0.742 & 0.010 \\
  &\textbf{GGL} & 0.998 & 0.000 & 0.467 & 0.037 & 0.999 & 0.000 & 0.733 & 0.009 \\ 
  \hline
\rowcolor[HTML]{DAE8FC} 
\cellcolor{white}
\multirow{5}{*}{\textbf{S = 0.0025}}
&\textbf{BPUGM} & 0.965 & 0.000 & 0.545 & 0.005 & 0.966 & 0.000 & 0.755 & 0.001 \\
  &\textbf{GemBag} & 0.998 & 0.000 & 0.393 & 0.008 & 1.000 & 0.000 & 0.696 & 0.002 \\
  &\textbf{FGL} & 0.998 & 0.000 & 0.419 & 0.064 & 0.999 & 0.000 & 0.709 & 0.016 \\ 
  &\textbf{GGL} & 0.998 & 0.000 & 0.405 & 0.049 & 0.999 & 0.000 & 0.702 & 0.012 \\ 
  \hline
\rowcolor[HTML]{DAE8FC} 
\cellcolor{white}
\multirow{5}{*}{\textbf{S = 0.0050}}
&\textbf{BPUGM} & 0.963 & 0.000 & 0.492 & 0.001 & 0.966 & 0.000 & 0.729 & 0.000 \\
  &\textbf{GemBag} & 0.996 & 0.000 & 0.354 & 0.003 & 1.000 & 0.000 & 0.677 & 0.001 \\ 
  &\textbf{FGL} & 0.992 & 0.000 & 0.493 & 0.046 & 0.995 & 0.000 & 0.744 & 0.011 \\
  &\textbf{GGL} & 0.994 & 0.000 & 0.525 & 0.012 & 0.997 & 0.000 & 0.761 & 0.003 \\
  \hline
\end{tabular}}
\end{table}

\begin{table}[ht]
\centering
\caption{Accuracy, sensitivity, specificity, balanced accuracy and AUC over  \textbf{100} datasets, K=4, N=50, Scenario 4: Same graph}
\scalebox{.85}{
\begin{tabular}{|c|c|c|c|c|c|c|c|c|c|}
\multicolumn{1}{l}{}                                       &\multicolumn{1}{l|}{}                                       & \multicolumn{2}{c|}{\cellcolor[HTML]{1F4E78}{\color[HTML]{FFFFFF} \textbf{Accuracy}}}                                                                               & \multicolumn{2}{c|}{\cellcolor[HTML]{1F4E78}{\color[HTML]{FFFFFF} \textbf{Sensitivity}}}                                                                      & \multicolumn{2}{c||}{\cellcolor[HTML]{1F4E78}{\color[HTML]{FFFFFF} \textbf{Specificity}}}                                                                     & \multicolumn{2}{c|}{\cellcolor[HTML]{1F4E78}{\color[HTML]{FFFFFF} \textbf{AUC}}}                                                                      \\ 
\multicolumn{1}{l}{}                                       & \multicolumn{1}{l|}{}                                       & \multicolumn{1}{c|}{\cellcolor[HTML]{1F4E78}{\color[HTML]{FFFFFF} \textbf{Mean}}} & \multicolumn{1}{c|}{\cellcolor[HTML]{1F4E78}{\color[HTML]{FFFFFF} \textbf{SE}}} & \multicolumn{1}{c|}{\cellcolor[HTML]{1F4E78}{\color[HTML]{FFFFFF} \textbf{Mean}}} & \multicolumn{1}{c|}{\cellcolor[HTML]{1F4E78}{\color[HTML]{FFFFFF} \textbf{SE}}} & \multicolumn{1}{c|}{\cellcolor[HTML]{1F4E78}{\color[HTML]{FFFFFF} \textbf{Mean}}} & \multicolumn{1}{c||}{\cellcolor[HTML]{1F4E78}{\color[HTML]{FFFFFF} \textbf{SE}}} & \multicolumn{1}{c|}{\cellcolor[HTML]{1F4E78}{\color[HTML]{FFFFFF} \textbf{Mean}}} & \multicolumn{1}{c|}{\cellcolor[HTML]{1F4E78}{\color[HTML]{FFFFFF} \textbf{SE}}}  \\ 
\multicolumn{1}{l}{}                                       &\multicolumn{1}{l|}{}                                       & \multicolumn{8}{c|}{\textbf{$p=20$}}  \\
\hline
\rowcolor[HTML]{DAE8FC} 
\cellcolor{white}
\multirow{5}{*}{\textbf{S = 0.010}}
&\textbf{BPUGM} & 0.915 & 0.000 & 0.840 & 0.020 & 0.915 & 0.000 & 0.878 & 0.005 \\ 
  &\textbf{GemBag}
  & 0.996 & 0.000 & 0.642 & 0.052 & 1.000 & 0.000 & 0.821 & 0.013 \\ 
  &\textbf{FGL}   & 0.966 & 0.001 & 0.575 & 0.092 & 0.970 & 0.001 & 0.773 & 0.020 \\ 
  &\textbf{GGL} & 0.977 & 0.001 & 0.772 & 0.078 & 0.980 & 0.001 & 0.876 & 0.019 \\ 
  \hline
\rowcolor[HTML]{DAE8FC} 
\cellcolor{white}
\multirow{5}{*}{\textbf{S = 0.025}}
&\textbf{BPUGM} & 0.898 & 0.000 & 0.502 & 0.007 & 0.915 & 0.000 & 0.709 & 0.002 \\ 
  &\textbf{GemBag}
  & 0.969 & 0.000 & 0.264 & 0.008 & 1.000 & 0.000 & 0.632 & 0.002 \\ 
  &\textbf{FGL} & 0.907 & 0.003 & 0.420 & 0.026 & 0.928 & 0.003 & 0.674 & 0.004 \\ 
  &\textbf{GGL} & 0.923 & 0.003 & 0.470 & 0.023 & 0.943 & 0.004 & 0.706 & 0.005 \\ 
  \hline
\rowcolor[HTML]{DAE8FC} 
\cellcolor{white}
\multirow{5}{*}{\textbf{S = 0.050}}
&\textbf{BPUGM} & 0.815 & 0.000 & 0.286 & 0.002 & 0.919 & 0.000 & 0.602 & 0.000 \\ 
  &\textbf{GemBag}
  & 0.852 & 0.000 & 0.095 & 0.001 & 1.000 & 0.000 & 0.547 & 0.000 \\ 
  &\textbf{FGL}  & 0.792 & 0.001 & 0.319 & 0.009 & 0.884 & 0.002 & 0.602 & 0.001 \\ 
  &\textbf{GGL} & 0.799 & 0.001 & 0.312 & 0.009 & 0.894 & 0.003 & 0.603 & 0.001 \\ 
  \hline
  \multicolumn{1}{l}{}                                       &\multicolumn{1}{l|}{}                                       & \multicolumn{8}{c|}{\textbf{$p=50$}}                                                                                   \\ 
\hline
\rowcolor[HTML]{DAE8FC} 
\cellcolor{white}
\multirow{5}{*}{\textbf{S = 0.0010}}
&\textbf{BPUGM} & 0.950 & 0.000 & 0.812 & 0.038 & 0.950 & 0.000 & 0.881 & 0.009 \\  
  &\textbf{GemBag}
  & 0.999 & 0.000 & 0.693 & 0.089 & 1.000 & 0.000 & 0.846 & 0.022 \\
  &\textbf{FGL} & 0.998 & 0.000 & 0.193 & 0.057 & 0.999 & 0.000 & 0.596 & 0.014 \\ 
  &\textbf{GGL} & 0.999 & 0.000 & 0.590 & 0.212 & 0.999 & 0.000 & 0.795 & 0.053 \\
  \hline
\rowcolor[HTML]{DAE8FC} 
\cellcolor{white}
\multirow{5}{*}{\textbf{S = 0.0025}}
&\textbf{BPUGM} & 0.949 & 0.000 & 0.603 & 0.015 & 0.950 & 0.000 & 0.777 & 0.004 \\ 
  &\textbf{GemBag}
  & 0.998 & 0.000 & 0.384 & 0.018 & 1.000 & 0.000 & 0.692 & 0.004 \\ 
  &\textbf{FGL} & 0.994 & 0.000 & 0.193 & 0.032 & 0.997 & 0.000 & 0.595 & 0.008 \\ 
  &\textbf{GGL} & 0.997 & 0.000 & 0.390 & 0.081 & 0.999 & 0.000 & 0.694 & 0.020 \\ 
  \hline
\rowcolor[HTML]{DAE8FC} 
\cellcolor{white}
\multirow{5}{*}{\textbf{S = 0.0050}}
&\textbf{BPUGM} & 0.946 & 0.000 & 0.429 & 0.005 & 0.950 & 0.000 & 0.690 & 0.001 \\
  &\textbf{GemBag}
  & 0.993 & 0.000 & 0.240 & 0.007 & 1.000 & 0.000 & 0.620 & 0.002 \\
  &\textbf{FGL} & 0.988 & 0.000 & 0.143 & 0.015 & 0.995 & 0.000 & 0.569 & 0.003 \\
  &\textbf{GGL} & 0.992 & 0.000 & 0.275 & 0.033 & 0.997 & 0.000 & 0.636 & 0.008 \\ 
  \hline
  \multicolumn{1}{l}{}                                       &\multicolumn{1}{l|}{}                                       & \multicolumn{8}{c|}{\textbf{$p=100$}}                                                                                   \\ 
\hline
\rowcolor[HTML]{DAE8FC} 
\cellcolor{white}
\multirow{5}{*}{\textbf{S = 0.010}}
&\textbf{BPUGM} & 0.966 & 0.000 & 0.609 & 0.010 & 0.966 & 0.000 & 0.787 & 0.002 \\ 
  &\textbf{GemBag} & 0.999 & 0.000 & 0.466 & 0.012 & 1.000 & 0.000 & 0.733 & 0.003 \\  
  &\textbf{FGL} & 0.998 & 0.000 & 0.127 & 0.015 & 0.999 & 0.000 & 0.563 & 0.004 \\
  &\textbf{GGL} & 0.999 & 0.000 & 0.444 & 0.069 & 1.000 & 0.000 & 0.722 & 0.017 \\
  \hline
\rowcolor[HTML]{DAE8FC} 
\cellcolor{white}
\multirow{5}{*}{\textbf{S = 0.0025}}
&\textbf{BPUGM} & 0.965 & 0.000 & 0.591 & 0.003 & 0.966 & 0.000 & 0.779 & 0.001 \\
  &\textbf{GemBag} & 0.998 & 0.000 & 0.432 & 0.005 & 1.000 & 0.000 & 0.716 & 0.001 \\
  &\textbf{FGL} & 0.995 & 0.000 & 0.221 & 0.019 & 0.997 & 0.000 & 0.609 & 0.005 \\ 
  &\textbf{GGL} & 0.997 & 0.000 & 0.530 & 0.046 & 0.999 & 0.000 & 0.764 & 0.012 \\
  \hline
\rowcolor[HTML]{DAE8FC} 
\cellcolor{white}
\multirow{5}{*}{\textbf{S = 0.0050}}
&\textbf{BPUGM} & 0.960 & 0.000 & 0.313 & 0.001 & 0.966 & 0.000 & 0.640 & 0.000 \\ 
  &\textbf{GemBag} & 0.991 & 0.000 & 0.190 & 0.001 & 1.000 & 0.000 & 0.595 & 0.000 \\
  &\textbf{FGL} & 0.986 & 0.000 & 0.146 & 0.003 & 0.995 & 0.000 & 0.571 & 0.001 \\
  &\textbf{GGL} & 0.990 & 0.000 & 0.264 & 0.009 & 0.998 & 0.000 & 0.631 & 0.002 \\ 
  \hline
\end{tabular}}
\end{table}

\end{document}